%% file: main.tex
\def\anon{0}
\def\draft{1}

\documentclass[11pt]{article}
\usepackage{color,colortbl,latexsym,amsthm,amsmath,amssymb,path,enumitem,soul,tikz,fullpage,subcaption}
\usetikzlibrary{shapes.geometric}
\usepackage[colorlinks,citecolor=blue,linkcolor=blue,urlcolor=blue,pagebackref]{hyperref}
\usepackage{algorithm}
\usepackage{algorithmicx}
\usepackage{algpseudocode}
\usepackage{soul}
\usepackage{multirow}

\usepackage{relsize}
\usepackage{diagbox}
\usepackage{comment}
\usepackage[font={small,sf}]{caption}
\newcommand{\ignore}[1]{}

\makeatother
\makeatletter
\newenvironment{breakablealgorithm}
  {%
   \begin{center}
     \refstepcounter{algorithm}%
     \hrule height.8pt depth0pt \kern2pt%
     \renewcommand{\caption}[2][\relax]{%
       {\raggedright\textbf{\ALG@name~\thealgorithm} ##2\par}%
       \ifx\relax##1\relax %
         \addcontentsline{loa}{algorithm}{\protect\numberline{\thealgorithm}##2}%
       \else %
         \addcontentsline{loa}{algorithm}{\protect\numberline{\thealgorithm}##1}%
       \fi
       \kern2pt\hrule\kern2pt
     }
  }{%
     \kern2pt\hrule\relax%
   \end{center}
  }
\makeatother

\setlist[itemize]{itemsep=-0.1em}
\setlist[enumerate]{itemsep=-0.1em}

\makeatletter
\def\th@plain{%
  \thm@notefont{}%
  \itshape %
}
\def\th@definition{%
  \thm@notefont{}%
  \normalfont %
}
\makeatother

\newtheorem{theorem}{Theorem}[section]
\newtheorem{lemma}[theorem]{Lemma}
\newtheorem{corollary}[theorem]{Corollary}

\newtheorem{hypothesis}[theorem]{Hypothesis}
\newtheorem{question}{Question}
\newtheorem{claim}[theorem]{Claim}
\newtheorem{definition}[theorem]{Definition}

\newtheorem{remark}[theorem]{Remark}
\newtheorem{proposition}[theorem]{Proposition}

\newtheorem{example}[theorem]{Example}
\usepackage{graphicx,psfrag}

\usepackage{listings,xcolor}
\usepackage{dsfont}

\newcommand{\FF}{\ensuremath{\mathbb F}}

\renewcommand{\epsilon}{\varepsilon}
\newcommand{\eps}{\varepsilon}

\def \poly {\mathop{\rm poly}} %

\def\tO{\tilde{O}}

\def\eps{{\varepsilon}}

\ifnum\draft=1
\newcommand{\mina}[1]{\textcolor{red}{(Mina: #1)}}
\newcommand{\surya}[1]{\textcolor{magenta}{(Surya: #1)}}
\newcommand{\yinzhan}[1]{\textcolor{green}{(Yinzhan: #1)}}
\newcommand{\virgi}[1]{\textcolor{red}{(Virginia: #1)}}
\else
    \newcommand{\mina}[1]{}
    \newcommand{\surya}[1]{}
    \newcommand{\yinzhan}[1]{}
    \newcommand{\virgi}[1]{}
\fi 
\newcommand\cliquelist[1]{(#1)\text{-}\mathsf{Clique}\text{-}\mathsf{Listing}}
\newcommand\cliquedet[1]{(#1)\text{-}\mathsf{Clique}\text{-}\mathsf{Detection}}
\newcommand{\exactclique}[1]{\mathsf{Exact}\text{-}#1\text{-}\mathsf{Clique}}
\newcommand\MM{\mathsf{MM}}
\newcommand\sparse{{\tt Sparse}}
\newcommand\dense{{\tt Dense}}

\title{Towards Optimal Output-Sensitive Clique Listing\\
{{or: Listing Cliques from Smaller Cliques}}}
\ifnum\anon=1
    \author{Anonymous}
\else
    \author{\normalsize Mina Dalirrooyfard\thanks{Morgan Stanley Research. \href{mailto:minad@mit.edu}{\texttt{minad@mit.edu}}. While at MIT, supported by a Google Faculty Research Award and an Akamai MIT CS Theory Group Fellowship.} \and \normalsize Surya Mathialagan\thanks{MIT. \href{mailto:smathi@mit.edu}{\texttt{smathi@mit.edu}}. Supported by the Siebel Scholars program, by DARPA under Agreement No. HR00112020023 and by NSF grant CNS-2154149.}
    \and \normalsize Virginia Vassilevska Williams\thanks{MIT. \href{mailto:virgi@mit.edu}{\texttt{virgi@mit.edu}}. Partially supported by NSF Career Award CCF-1651838, NSF Grant CCF-2129139, a Sloan Research Fellowship and a Google Faculty Research Award.} 
    \and \normalsize  Yinzhan Xu\thanks{MIT. \href{mailto:xyzhan@mit.edu}{\texttt{xyzhan@mit.edu}}. Supported by NSF Grant CCF-2129139.}}
\fi 

\date{}
\begin{document}

\maketitle 
\pagenumbering{gobble} 
\begin{abstract}
We study the problem of finding and listing $k$-cliques in an $m$-edge, $n$-vertex graph, for constant $k\geq 3$. This is a fundamental problem of both theoretical and practical importance.

Our first contribution is an algorithmic framework for finding $k$-cliques that gives the first improvement in 19 years over the old runtimes for $4$ and $5$-clique finding, as a function of $m$ [Eisenbrand and Grandoni, TCS'04]. 
With the current bounds on matrix multiplication, our algorithms run in $O(m^{1.66})$ and $O(m^{2.06})$ time, respectively, for $4$-clique and $5$-clique finding.

Our main contribution is an output-sensitive algorithm for listing $k$-cliques, for any constant $k\geq 3$. We complement the algorithm with tight lower bounds based on standard fine-grained assumptions. Previously, the only known conditionally optimal output-sensitive algorithms were for the case of $3$-cliques given by Bj\"{o}rklund, Pagh, Vassilevska W. and Zwick [ICALP'14]. If the matrix multiplication exponent $\omega$ is $2$, and if the number of $k$-cliques $t$ is large enough, the running time of our algorithms is 
$$\tilde{O}\left(\min\{m^{\frac{1}{k-2}}t^{1 - \frac{2}{k(k-2)}},n^{\frac{2}{k-1}}t^{1-\frac{2}{k(k-1)}}\}\right),$$
and this is {\em tight} under the Exact-$k$-Clique Hypothesis. This running time naturally extends the running time obtained by Bj\"{o}rklund, Pagh, Vassilevska W. and Zwick for $k=3$.

Our framework is very general in that it gives $k$-clique listing algorithms whose running times can be measured in terms of the number of $\ell$-cliques $\Delta_\ell$ in the graph for any $1\leq \ell<k$. This generalizes the typical parameterization in terms of $n$ (the number of $1$-cliques) and $m$ (the number of $2$-cliques).

If $\omega$ is $2$, and if the size of the output, $\Delta_k$, is sufficiently large, then for every $\ell<k$, the running time of our algorithm for listing $k$-cliques is 
$$\tilde{O}\left(\Delta_\ell^{\frac{2}{\ell (k - \ell)}}\Delta_k^{1-\frac{2}{k(k-\ell)}}\right).$$ 
We also show that this runtime is {\em optimal} for all $1 \leq \ell < k$ under the Exact $k$-Clique hypothesis. 
\end{abstract}

\newpage

\tableofcontents{}
\newpage 
\pagenumbering{arabic} 
\section{Introduction}
\input{1-intro}

\section{Preliminaries}
\label{sec:prelim}

\input{2-prelim}

\section{Detecting Cliques}\label{sec:detection}

\input{3-detection}

\section{Lower Bounds for Listing Cliques}\label{sec:lower_bounds}
\input{4-lower-bound}

\section{Optimal Listing Algorithms for Graphs with Many \texorpdfstring{$k$}{k}-Cliques}
\label{sec:upper-bound}
\input{5-upper-bound}

\section{Extending the Algorithm to Graphs with Fewer \texorpdfstring{$k$}{k}-Cliques}
\label{sec:general-list}
\input{6-general-list}

\section{6-Clique Madness}
\label{sec:6clique}
\input{7-6-clique}

\bibliographystyle{alpha}
\bibliography{ref}

\end{document}

%% file: 1-intro.tex
Finding, counting and listing cliques in graphs are fundamental tasks with numerous applications. In any type of network (social, biological, financial, web, maps, etc.) clique listing is used to find patterns such as communities, spam-link farms, motifs, correlated genes and more (see \cite{SchankW05,listingcliqueswww} and the many citations within).

As finding a clique of maximum size has long been known to be NP-hard \cite{Karp72}, the focus in numerous practical works (see \cite{listingcliqueswww,listingcliquesdensest,listingcliquesnucleus,count5via3,trilistlatapy,ChibaN85,SchankW05,ShunT15,ChuC11}) 
is on listing cliques of small size such as triangles and $4$-cliques. 

More generally, in an $n$-node $m$-edge graph, for a constant $k \geq 3$  (independent of $n$ and $m$), we want to find, count or list the $k$-cliques in $G$.
Chiba and Nishizeki \cite{ChibaN85} presented an algorithm that for any constant $k\geq 3$ can list all $k$-cliques in a graph in $O(m\alpha^{k-2})$ time, where $\alpha\leq O(\sqrt{m})$ is the {\em arboricity} of the given graph. This algorithm is among the most efficient clique-listing approaches in practice (see e.g. \cite{listingcliqueswww} and the references within).

Purely in terms of $m$, Chiba and Nishizeki's algorithm runs in $O(m^{k/2})$ time. Since $O(m^{k/2})$ is also the maximum number of $k$-cliques in an $m$-edge graph, this algorithm is optimal, as long as the graph has $\Theta(m^{k/2})$ cliques (e.g. when the graph itself is a clique). However, when the graph has $t$ $k$-cliques, where $t$ is $o(m^{k/2})$, the optimality argument no longer works. In fact, it has been known for almost 40 years \cite{nesetril1985complexity} that when $t=1$, a much faster runtime is possible using fast matrix multiplication.

This motivates the study of {\bf output-sensitive} algorithms for $k$-clique listing: algorithms whose running time depends on the number of $k$-cliques in the output. An even more desirable version of an output-sensitive algorithm is one that can also take as input some parameter $t$, and can list up to $t$ $k$-cliques in the graph. When $t$ is much smaller than the number of $k$-cliques in the graph, such an algorithm could potentially be more efficient. These two versions are actually runtime-equivalent up to logarithmic factors for most natural running times (we provide a proof in Section~\ref{sec:prelim} for completeness). We thus use these  two notions interchangeably.

Bj{\"o}rklund, Pagh, Vassilevska W. and Zwick \cite{bjorklund2014listing} designed such output-sensitive algorithms for triangle listing with runtime $\tilde{O}(n^{\omega}+n^{\frac{3(\omega-1)}{5-\omega}}t^{\frac{2(3-\omega)}{5-\omega}})$ and $\tilde{O}(m^{\frac{2\omega}{\omega+1}}+m^{\frac{3(\omega-1)}{\omega+1}}t^{\frac{3-\omega}{\omega+1}})$\footnote{We use $\tilde{O}$ to hide polylog factors.}, where $\omega<2.372$~\cite{duan2023, VXXZ24} is the exponent of matrix multiplication and $t$ is the number of triangles listed. If $\omega=2$, the runtimes simplify to $\tilde{O}(n^2+nt^{2/3})$ and $\tilde{O}(m^{4/3}+mt^{1/3})$, and these are shown to be conditionally optimal for any $t= \Omega(n^{1.5})$ and $t= \Omega(m)$ respectively under the popular $3$SUM hypothesis \cite{patrascu2010towards,kopelowitz2016higher} and the even more believable Exact Triangle hypothesis \cite{williams2020monochromatic}. 
There have also been many recent works focusing on output-sensitive cycle-listing algorithms. The works of \cite{abboud2022listing, jin2023removing} show $O(\min\{n^2 + t, m^{4/3} + t\})$ algorithms for listing $t$ 4-cycles, and the work of ~\cite{jin2024listing} shows $\tilde{O}(n^2 + t)$ algorithm for listing $t$ 6-cycles. Moreover, matching conditional lower bounds for $4$-cycle listing were shown under the 3SUM hypothesis~\cite{jin2023removing,abboud2023stronger3sum}, which was subsequently strengthened to hold under the Exact Triangle hypothesis \cite{CX24}. 

While the output-sensitive questions for triangle listing and 4-cycle listing are is well-understood by now, no similar conditionally optimal results are known for $k$-clique listing when $k\geq 4$. 

\begin{question}\label{q2}
    What is the best output-sensitive algorithm for $k$-clique listing for $k>3$? 
\end{question}

When analyzing algorithms, researchers look at a variety of {\bf parameters} to understand performance: the size of the input (typically $n$ and $m$ for graph problems), the size of the output (the number of $k$-cliques), and other natural parameters of the input (e.g. the arboricity, as in \cite{ChibaN85}). In this work, we study clique-listing algorithms parameterized by $\Delta_\ell$, the number of $\ell$-cliques in the graph for $\ell < k$.

To motivate this, let us consider the first non-trivial algorithm for $k$-clique finding by Ne\v{s}etril and Poljak \cite{nesetril1985complexity}.
For simplicity, assume that $k$ is divisible by $3$. First, the algorithm enumerates all $k/3$-cliques in the input graph $G$, and forms a new graph $H$ whose nodes represent the $k/3$-cliques of $G$ and whose edges connect two $k/3$-cliques that together form a $2k/3$-clique. 
The triangles of $H$ correspond to $k$-cliques in $G$, and so Ne\v{s}etril and Poljak reduce $k$-clique finding, counting and listing in $G$ to finding, counting and listing (respectively) of {\em triangles} in $H$\footnote{Note the reduction also works for counting and listing because every $k$-clique is represented by exactly $\binom{k}{k/3,k/3,k/3}$ triangles.}. As there are $O(n^{k/3})$ $k/3$-cliques in $G$, and since triangle finding or counting in $N$-node graphs can be done in $O(N^\omega)$ time \cite{itairodeh}, \cite{nesetril1985complexity} gave an $O(n^{\omega k/3})$ time algorithm for $k$-clique finding or counting in $n$-node graphs.
Eisenbrand and Grandoni \cite{eisenbrand2004complexity} extended Ne\v{s}etril and Poljak's reduction to obtain a $k$-clique runtime of $O(n^{\beta(k)})$ where $\beta(k)=\omega(\lceil{k/3}\rceil, \lceil{(k-1)/3}\rceil, \lfloor{k/3}\rfloor)$, and $\omega(a,b,c)$ is the exponent of multiplying an $n^a\times n^b$ matrix by an $n^b\times n^c$ matrix. As the runtime of $k$-clique detection has remained unchallenged for several decades, the  hypothesis that these algorithms are optimal has been used to provide conditional lower bounds in several works (e.g.~\cite{AbboudBW18,BackursT17,BringmannW17}). Throughout the paper, we consider the word-RAM model of computation with $O(\log n)$ bit words. 

\begin{hypothesis}[$k$-Clique Hypothesis]
\label{hyp:k_clique}
    On a word-RAM model with $O(\log n)$ bit words, detecting a $k$-clique in an $n$-node graph requires $n^{\beta(k) - o(1)}$ time , where $\beta(k)=\omega(\lceil{k/3}\rceil, \lceil{(k-1)/3}\rceil, \lfloor{k/3}\rfloor)$.
\end{hypothesis}

Now, suppose $G$ has a small number $q$ of $k/3$-cliques and suppose we can list these $k/3$-cliques quickly, then Ne\v{s}etril and Poljak's algorithm would run in only $O(q^{\omega})$ additional time which can be much faster than $O(n^{\omega k/3})$.

More generally, if a graph has a small number $\Delta_\ell$ of $\ell$-cliques for $\ell<k$, a simple generalization of Ne\v{s}etril and Poljak's reduction would reduce $k$-clique to $k/\ell$-clique in a graph with $\Delta_\ell$ nodes (assuming $k$ is divisible by $\ell$ for simplicity). If one can list the $\ell$-cliques fast, then $k$-clique finding, listing and detection can all be done faster in graphs with small $\Delta_\ell$.

In other words, for $k$-clique problems, the number of $\ell$-cliques $\Delta_\ell$, where $\ell<k$ is arguably the most natural parameter. The usual input parameters $n$ and $m$ can be viewed as the special cases $\Delta_1$ and $\Delta_2$.
We are not the first to suggest this natural parameterization of the input. In fact, small $\Delta_\ell$ values have been exploited to obtain faster $k$-clique algorithms in experimental algorithmics: e.g., \cite{count5via3} and \cite{osti_1141233} count $k$-cliques faster in graphs with a small number of triangles. Motivated by these practical results, we are the first to consider the following question within theoretical computer science:

\begin{question}\label{q3}
    Can we get a general conditionally optimal algorithms for output-sensitive $k$-clique listing in terms of the number $\Delta_\ell$ of $\ell$-cliques for any $\ell<k$?
\end{question}

\subsection{Our Contributions}
We present a systematic study of clique finding and listing, and provide answers to both Questions \ref{q2} and \ref{q3}. We give the first output-sensitive algorithms for listing $k$-cliques for $k \geq 4$. We also give the first general algorithms for detecting and listing $k$-cliques in terms of the number of $\ell$-cliques, and the first fine-grained lower bounds for the listing problem for general $k$. Our lower bounds show that our algorithms are tight for a non-trivial range of the number of $k$-cliques to output. We summarize our contributions in Table~\ref{tab:contributions}. $\cliquedet{k,\ell}$ and $\cliquelist{k,\ell}$ refer to \emph{detecting} and \emph{listing} $k$-cliques respectively given a list of all $\ell$-cliques. Here, $t$ is the number of $k$-cliques we are asked to list.

\begingroup
\renewcommand{\arraystretch}{1.2} 
    \begin{table}[h]
        \centering
        \scalebox{0.85}{
        \begin{tabular}{p{2.7cm}|c|c}
        &\textbf{Results} & \textbf{References}\\
        \hline 
        \hline 
             \multirow{2}{=}{\textbf{Detection}} & New $\cliquedet{k,\ell}$ framework &  Section~\ref{sec:detection}\\
             & Improved $\cliquedet{4,2}$ and $\cliquedet{5,2}$ & Theorem~\ref{intro:detection-m} \\
             \hline 
             \textbf{Lower bounds} & Conditional lower bounds for $\cliquelist{k,\ell}$ & Theorems~\ref{thm:lb_intro_mn}, \ref{thm:lb_intro}\\
             \hline 
             \multirow{5}{=}{\textbf{Listing}} & Optimal algorithms for $(4,1)$ and $\cliquelist{5,1}$ & Theorems~\ref{thm:4_1_opt}, \ref{thm:5_1_opt}\\
             & Nearly-everywhere optimal algorithms for $(4,\ell)$, $\cliquelist{5,\ell}$ & Theorems~\ref{thm:4_2}, \ref{thm:5_2}\\
             & Optimal $\cliquelist{k,\ell}$ algorithms for large $t$ & Theorems~\ref{thm:intro-k-large-t-listing-mn}, \ref{thm:intro-k-l-large-t-listing}\\
             & Generalized $\cliquelist{k,\ell}$ algorithm for all $t$ & Section~\ref{sec:general-list}\\
             & Refined analysis for $\cliquelist{6,1}$ & Section~\ref{sec:6clique}
        \end{tabular}}
        \caption{Summary of our contributions. $\cliquedet{k,\ell}$ and $\cliquelist{k,\ell}$ refer to \emph{detecting} and \emph{listing} $k$-cliques respectively given a list of all $\ell$-cliques. Here, $t$ is the number of $k$-cliques we are asked to list.}
        \label{tab:contributions}
    \end{table}
\endgroup

\paragraph{Improved 4 and 5-clique detection in sparse graphs.}
We provide a general algorithmic framework for detecting cliques. 
As special cases of the framework, we give the first improvement over the
the runtime of Eisenbrand and Grandoni \cite{eisenbrand2004complexity} for $4$ and $5$-clique detection in sparse graphs (we show this in Examples~\ref{ex:clique-det-4-2} and \ref{ex:clique-det-5-2} in Section~\ref{sec:detection_examples}).

 \begin{theorem}\label{intro:detection-m}
     There is an $O(m^{1.657})$ time algorithm for 4-clique detection and an $O(m^{2.057})$ time algorithm for 5-clique detection in $m$-edge graphs.
 \end{theorem}
We compare the explicit values of~\cite{eisenbrand2004complexity}'s exponent and our improved exponents in Table~\ref{table:improved_det_4_5} in terms of the current bounds for square and rectangular matrix multiplication \cite{VXXZ24}. 
    
    \begin{table}[ht]
        \centering
        \begin{tabular}{c|c|c}
         $k$ &  Previous exponent \cite{eisenbrand2004complexity} & Our exponent (Theorem~\ref{intro:detection-m})\\
         \hline 
         4 & 1.668 & 1.657\\
         5 & 2.096 & 2.057
        \end{tabular}
        \caption{The table contains exponents $c$ such that $4$ and $5$-clique detection is in $O(m^c)$ time. For $\cliquedet{4, 2}$ and $\cliquedet{5, 2}$, the previous exponent was given by $\beta(k) \cdot \beta(k-1)/(\beta(k) + \beta(k-1) - 1)$ \cite{eisenbrand2004complexity}, where $\beta(k)$ is the exponent of $k$-clique detection (as in Hypothesis~\ref{hyp:k_clique}). We give the runtime of their algorithm with the current bounds on square and rectangular matrix multiplication \cite{VXXZ24}. }\label{table:improved_det_4_5}
    \end{table}

\paragraph{Lower bounds for $k$-clique listing.}
Prior works \cite{patrascu2010towards,kopelowitz2016higher,williams2020monochromatic} give fine-grained lower bounds for listing triangles in an $n$-node, $m$-edge graph: triangle-listing requires $n^{1-o(1)}t^{2/3}$ time in $n$-node graphs, and requires $m^{1-o(1)}t^{1/3}$ in $m$-edge graphs time, under standard fine-grained hypotheses. The lower bounds imply tightness of the known algorithms \cite{bjorklund2014listing}  if $t$ is large enough: $t= \Omega(n^{1.5})$ or $t= \Omega(m)$ respectively.

The lower bounds of \cite{patrascu2010towards,kopelowitz2016higher} are under the $3$SUM hypothesis. 
Extending these to lower bounds for $k$-clique listing seems difficult. Instead we focus on the approach of \cite{williams2020monochromatic} who showed hardness under the Exact-Triangle hypothesis which states that finding a triangle of weight sum 0 in an $n$-node edge-weighted graph requires $n^{3-o(1)}$ time in the word-RAM model. The Exact-Triangle hypothesis is one of the most believable hypotheses in fine-grained complexity, as it is implied by both the $3$SUM hypothesis and the APSP hypothesis (see \cite{vsurvey}).

A natural generalization of the Exact-Triangle hypothesis is the Exact-$k$-Clique hypothesis (which coincides with the Exact-Triangle hypothesis for $k=3$):

\begin{hypothesis}[Exact-$k$-Clique hypothesis]\label{hyp:exact_k_clique}
For a constant $k\geq 3$, let $\exactclique{k}$ be the problem that given an $n$-node graph with edge weights in $\{-n^{100k},\dots, n^{100k}\}$, asks to determine whether the graph contains a $k$-clique whose edges sum to $0$.
Then, $\exactclique{k}$ requires $n^{k-o(1)}$ time, on the word-RAM model of computation with $O(\log n)$ bit words. 
\end{hypothesis}

The Exact-$k$-Clique hypothesis is among the popular hardness hypotheses in fine-grained complexity.
Most recently, it has been used to give hardness for the Orthogonal Vectors problem in moderate dimensions \cite{abboud2018more} and join queries in databases \cite{BringmannCM22}. Moreover, due to known reductions  (see e.g. \cite{vsurvey}), the Exact-$k$-Clique hypothesis is at least as believable as  the Max-Weight-$k$-Clique hypothesis which is used in many previous papers (e.g. \cite{AbboudWW14,BackursDT16,BackursT17,LincolnWW18,BringmannGMW20}). 

Under the Exact-$k$-Clique hypothesis we prove lower bounds for $k$-clique listing for all $k \geq 3$. These are the first lower bounds for output-sensitive clique listing for $k \geq 4$.

\begin{theorem}\label{thm:lb_intro_mn}
    For any $k \ge 3$, and $\gamma \in [0, k]$, listing $t$ $k$-cliques in a graph with $n$ vertices, and in a graph with $m$ nodes requires 
    \[
        \left(n^{\frac{2}{k-1}}t^{1-\frac{2}{k(k-1)}}\right)^{1-o(1)} \quad\text{and} \quad\left(m^{\frac{1}{k-2}}t^{1 - \frac{2}{k(k-2)}}\right)^{1-o(1)}
    \]
    time respectively under the Exact-$k$-Clique hypothesis. 
\end{theorem}
This is a special case of Theorem~\ref{thm:lower_bound} in the main body. 
For $k=3$ this is the same lower bound as previously proven \cite{patrascu2010towards,kopelowitz2016higher, williams2020monochromatic}.
Shortly, we will present algorithms that match our lower bound for all $k,m,n$ and for large $t$ if $\omega =2$, implying that our lower bound is tight. This is in fact {\bf the first output-sensitive lower bound} for $k$-clique listing problems for $k \geq 4$, and the first such lower bound for {\em any } graph pattern of size at least 5. 

\paragraph{Optimal algorithms for 4 and 5-clique listing.} For the special cases of $k = 4, 5$, we give algorithms parametrized by the number of vertices $n$ and number of $k$-cliques $t$ which are  conditionally {\bf optimal} if $\omega = 2$. We prove these results in Corollary~\ref{cor:4-1-opt} and Corollary~\ref{cor:5_1_opt}.

Similar to \cite{bjorklund2014listing}, we state our runtimes in terms of $\omega$. %
In our analysis, we compute rectangular matrix multiplication by truncating it to multiple instances of square matrix multiplication. If one is interested in better numerical values, one could instead use the best upper bound on rectangular matrix multiplication \cite{VXXZ24} in these steps.

\begin{theorem}\label{thm:4_1_opt}
    Given a graph on $n$ nodes, one can list $t$ 4-cliques in $$\tilde{O}\left(n^{\omega + 1} + n^{\frac{4(\omega - 1)(2\omega - 3)}{\omega^2 - 5 \omega + 12}}t^{1 - \frac{(\omega - 1)(2\omega - 3)}{\omega^2 - 5 \omega + 12}}\right)$$ time. If $\omega = 2$, the runtime is $\tilde{O}(n^3 + n^{2/3}t^{5/6}).$
\end{theorem}

Recall that the $4$-Clique hypothesis, which is a special case of  Hypothesis~\ref{hyp:k_clique} when $k = 4$, gives a lower bound of $n^{3-o(1)}$ if $\omega = 2$. Moreover, Theorem~\ref{thm:lb_intro_mn} gives a lower bound of $(n^{2/3}t^{5/6})^{1-o(1)}$. Therefore, this $4$-clique listing algorithm is indeed conditionally optimal.

\begin{theorem}\label{thm:5_1_opt}
    Given a graph on $n$ nodes, one can list $t$ 5-cliques in $$\tilde{O}\left(n^{\omega + 2} + n^{\frac{5(\omega - 1)(2\omega - 3)(3\omega - 5)}{48-47\omega + 16\omega^2 - \omega^3}}t^{1 - \frac{(\omega - 1)(2\omega - 3)(3\omega - 5)}{48-47\omega + 16\omega^2 - \omega^3}}\right)$$
    time. If $\omega = 2$, the runtime is $\tilde{O}(n^4 +n^{1/2}t^{9/10}).$
\end{theorem}
Recall that the 5-Clique hypothesis from Hypothesis~\ref{hyp:k_clique} gives us a lower bound of $n^{4-o(1)}$ if $\omega = 2$. Moreover,  Theorem~\ref{thm:lb_intro_mn} gives a lower bound of $(n^{1/2}t^{9/10})^{1-o(1)}$. Therefore, this $5$-clique listing algorithm is also conditionally optimal. 

\paragraph{Nearly-everywhere optimal algorithms for 4 and 5-clique listing in sparse graphs.} 
In the case of sparse graphs, we obtain conditionally optimal runtimes for $4$ and $5$-clique listing for almost all values of $t$ if $\omega = 2.$ The runtimes are stated in the following theorems and are pictorially depicted in Figure~\ref{fig:4_5_sparse_runtime}.

\begin{figure}[h]
    \centering
    \includegraphics[width=0.4\textwidth]{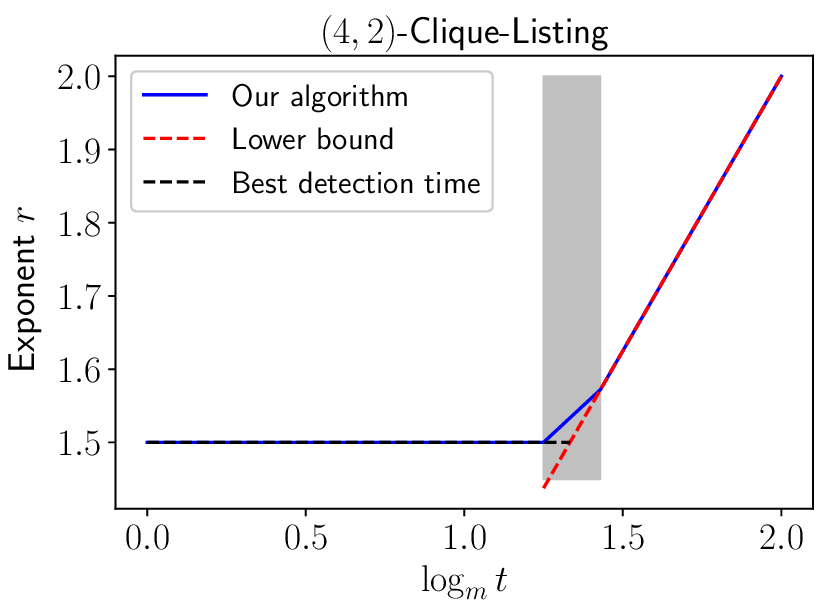}
    \hspace{0.05\textwidth}
    \includegraphics[width=0.4\textwidth]{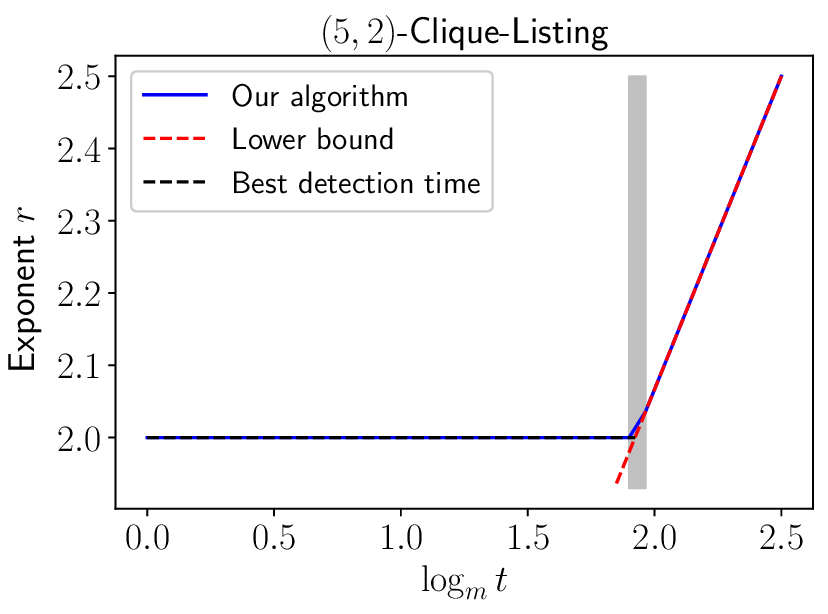}
    \caption{Upper and lower bounds for
    $4$ and $5$-cliques in graphs with $m$ edges, if $\omega = 2$. Here, $r$ is such that one can list 
    $t$ $4$-cliques or $5$-cliques respectively, in $\tilde{O}(m^r)$ time.
    The blue line corresponds to our upper bound from Theorems~\ref{thm:4_2} and \ref{thm:5_2}, the dashed red line denotes our lower bound from Theorem~\ref{thm:lb_intro}, and the dashed black line corresponds to the lower bound from Hypothesis~\ref{hyp:k_clique}. The shaded region highlights the portions of the algorithms which are not conditionally optimal.}
    \label{fig:4_5_sparse_runtime} 
\end{figure}

\begin{theorem}\label{thm:4_2}
    If $\omega = 2$, one can list $t$ 4-cliques in a graph with $m$ edges in time
    \begin{align*}
        \begin{cases}
                \tilde{O}(m^{3/2}) & \text{if $t \leq m^{5/4}$},\\
                \tilde{O}(mt^{2/5})& \text{if $m^{5/4}\leq t \leq m^{10/7}$},\\
                \tilde{O}(m^{1/2}t^{3/4}) & \text{if $t \geq m^{10/7}$}.
        \end{cases}
    \end{align*}
\end{theorem}
This algorithm matches the lower bound in Hypothesis~\ref{hyp:k_clique} when $t \leq m^{5/4}$, and it matches our lower bound of Theorem~\ref{thm:lb_intro_mn} when $t \geq m^{10/7}$.

\begin{theorem}\label{thm:5_2}
    If $\omega = 2$, one can list $t$ 5-cliques in a graph with $m$ edges in time
    \begin{align*}
        \begin{cases}
                \tilde{O}(m^{2}) & \text{if $t \leq m^{19/10}$},\\
                \tilde{O}(m^{17/18}t^{10/18}) & \text{if $m^{19/10}\leq t \leq m^{55/28}$},\\
                \tilde{O}(m^{1/3}t^{13/15}) & \text{if $t \geq m^{55/28}$}.
        \end{cases}
    \end{align*}
\end{theorem} 

 This algorithm matches the runtime of the lower bound in Hypothesis~\ref{hyp:k_clique} when $t \leq m^{19/10}$, and it matches our lower bound from Theorem~\ref{thm:lb_intro_mn} when $t \geq m^{55/28}$.

Theorem~\ref{thm:4_2} and Theorem~\ref{thm:5_2} are proved in Section~\ref{sec:4_5_l_listing}.

\paragraph{Optimal algorithms for listing many $k$-cliques.} More generally, we consider the problem of listing $k$-cliques for $k \geq 3$. For instance, consider the problem of listing 6-cliques in sparse graphs with $m$ edges. If we adapt the existing approach for $k$-clique detection~\cite{nesetril1985complexity, eisenbrand2004complexity} and directly reduce it to triangle listing in a graph with $m$ nodes and then use \cite{bjorklund2014listing}, we get an $\tO(m^2+mt^{2/3})$ runtime when $\omega = 2$. In comparison, the lower bound from Theorem~\ref{thm:lb_intro} is $(m^{1/4}t^{11/12})^{1-o(1)}$. When $t$ is close to maximum (as $t \to O(m^3)$), the $\tO(m^2+mt^{2/3})$ runtime is polynomially higher than the lower bound. Therefore, we cannot only rely on such reductions. 
    
Nevertheless, we give a conditionally {\em tight} algorithm for graphs with many $k$-cliques, provided that $\omega=2$ for sufficiently large number of cliques. In particular, the runtime of the algorithm in the theorem below matches the lower bound of Theorem~\ref{thm:lb_intro_mn}. 

\begin{theorem}[Informal]
    \label{thm:intro-k-large-t-listing-mn}
    If $\omega = 2$, 
    there is an algorithm for $k$-clique listing  which runs in time 
        \[ \tilde{O}\left(\min\left\{n^{\frac{2}{k-1}} t^{1 - \frac{2}{k(k-1)}}, m^{\frac{1}{k-2}}t^{1 - \frac{2}{k(k-2)}}\right\}\right)\]
    when $t$ is large. 
\end{theorem}
We give more explicit bounds on $t$ and the runtimes in terms of $\omega$ in Sections~\ref{sec:k_1_opt} and \ref{sec:k_l_opt}. In other words, we have an algorithm which \textbf{match the lower bound} in Theorem~\ref{thm:lb_intro_mn} for graphs with many $k$-cliques.

\paragraph{General listing algorithm for all $t$.} 
In Section~\ref{sec:general-list}, we give a general black-box approach (by non-trivially adapting previous reductions \cite{nesetril1985complexity, eisenbrand2004complexity}) that uses our (conditionally) optimal algorithm for a large number of $k$-cliques $t$ to obtain a fast algorithm that works for {\em all} $t$. The main advantage of this approach is its simplicity and generality. In particular, we obtain an intuitive and simple analysis of the runtime for all $k,t$. In Section~\ref{sec:general-list}, we show a comparison of our lower bounds and the runtime of our general algorithm in some examples. We illustrate the runtime of the general algorithm for some specific cases in Figure~\ref{fig:general_examples}. 
\begin{figure}[ht]
    \centering
    \includegraphics[width=0.38\textwidth]{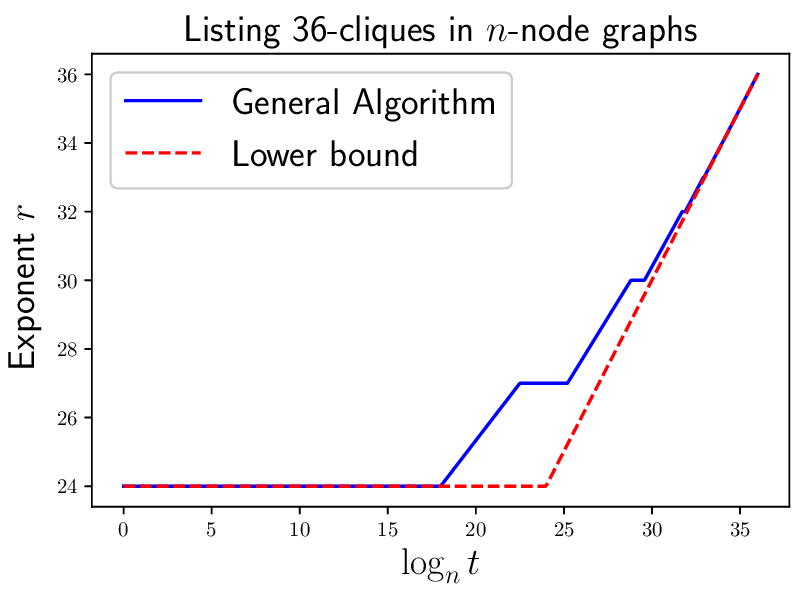}
    \hspace{0.05\textwidth}
    \includegraphics[width=0.38\textwidth]{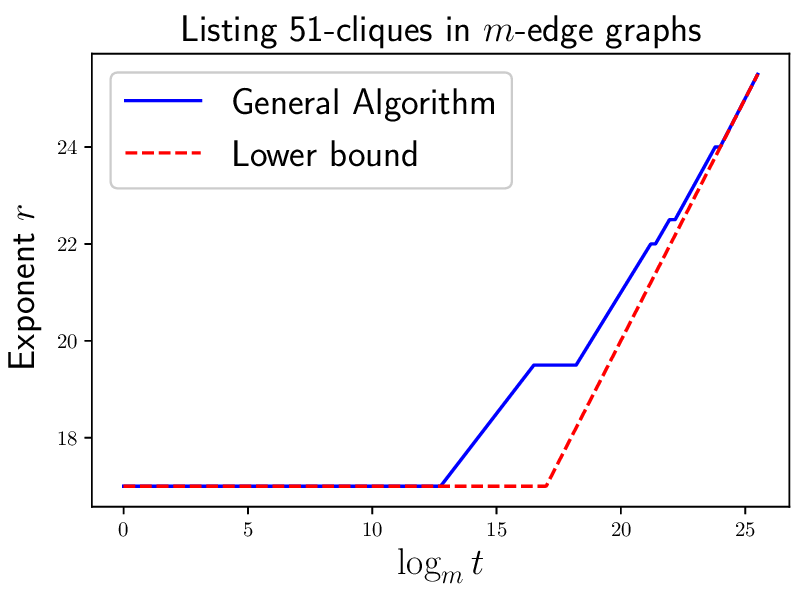}
    \caption{Upper and lower bounds listing 36-cliques in $n$-node graphs, and 51-cliques in $m$-edge graphs if $\omega = 2$. Here, the exponent $r$ is such that one can list $t$ $36$-cliques or $51$-cliques respectively, in $\tilde{O}(n^r)$ and $\tilde{O}(m^r)$ time respectively. The blue line corresponds to our upper bound from the general listing algorithm, and the dashed red line denotes the lower bounds from Hypothesis~\ref{hyp:k_clique} and Theorem~\ref{thm:lb_intro_mn}.}
    \label{fig:general_examples}
\end{figure}

\paragraph{Improved algorithm for 6-clique listing.} We note that our generic algorithm trades simplicity for optimality, and it is not always the best algorithm one can obtain for fixed $k$. 

In Section~\ref{sec:6clique}, we give a more refined algorithm for 6-clique listing in terms of $n$ and $t$ if $\omega = 2$ to illustrate how one might obtain a tighter runtime bound for specific $k$. In Figure~\ref{fig:6_1}, we compare our ``general'' bound, our best bound and our lower bounds to illustrate the improvement in the algorithm. However, since the number of terms and parameters in the runtime increases significantly with $k$, we do not do this refined analysis for all $k$.
 \begin{figure}[ht]
    \centering
    \includegraphics[width=0.45\textwidth]{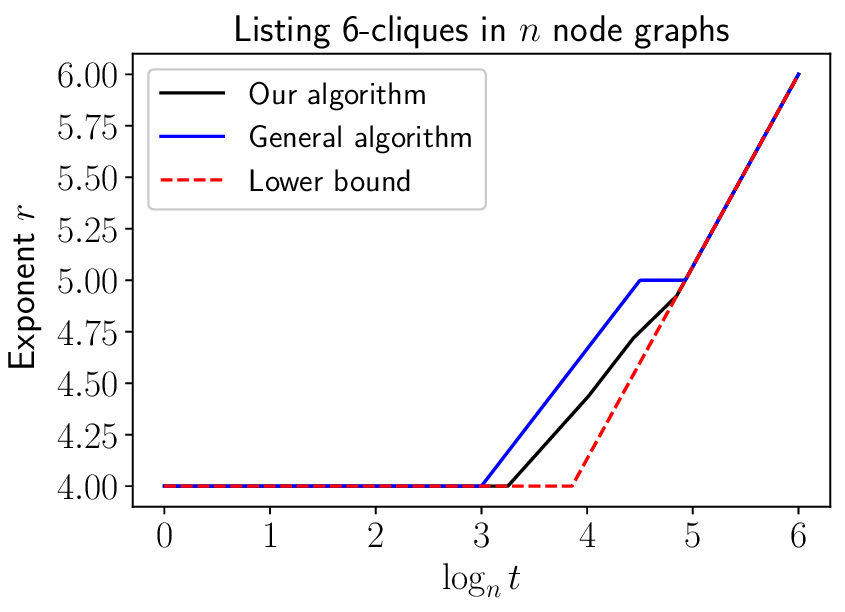}
    \caption{Upper and lower bounds for
    listing 6-cliques in $n$-node graphs if $\omega = 2$. Here, $r$ is such that one can list 
    $t$ $6$-cliques in $\tilde{O}(n^r)$ time.
    The blue line corresponds to the upper bound of our general listing algorithm, the black line corresponds to the upper bound of our refined algorithm, and the dashed red line denotes lower bound from Theorem~\ref{thm:lb_intro_mn} and Hypothesis~\ref{hyp:k_clique}.}
    \label{fig:6_1}
\end{figure}

\paragraph{Listing cliques from smaller cliques.} In fact, our frameworks are much more general and it extends to the problems of finding and listing $k$-cliques given a list of all $\ell$-cliques in the graph, for $\ell \geq 1$. We use the notation $\Delta_\ell$ to denote the number of $\ell$-cliques in the graph.

Let $\cliquedet{k, \ell}$ be the problem of detecting a $k$-clique in a graph $G$, given the list of all $\ell$-cliques in the graph for some $\ell\in \{1,\ldots,k-1\}$. Our framework applies to $\cliquedet{k, \ell}$ for any $k \ge 3, 1 \le \ell < k$. We note that while we only mention $k$-clique detection, we can use well-known techniques to also find $k$-cliques in the same runtime up to a log factor (see Section~\ref{sec:prelim:basic}). Moreover, our algorithm can also be used to count the number of cliques with the same runtime.

In Table~\ref{tab:detection_runtime} we present the exponents of our runtimes for $\cliquedet{k, \ell}$ for small values of $k$ and $\ell$ assuming $\omega = 2$. See Table~\ref{tab:det_exponent} for the runtime in terms of the current bound on $\omega$. 
For $\ell = 1$, we captures the best known $k$-clique detection algorithm and hence matches Hypothesis~\ref{hyp:k_clique}.
Although our general framework is simple, it is actually quite powerful, and allows us to obtain the first improvement in almost 20 years over the runtime of Eisenbrand and Grandoni \cite{eisenbrand2004complexity}, as discussed in Theorem~\ref{intro:detection-m}. 

\begin{table}[ht]
\centering 
\begin{tabular}{c|c|c|c|c|c|c|c|c|c|c}
     \backslashbox{$\ell$}{$k$} & 3 & 4 & 5 & 6 & 7 & 8 & 9 & 10 & 11 & 12 \\
     \hline 
     1 & 2 & 3 & 4 & 4 & 5 & 6 & 6 & 7 & 8 & 8\\
     2 & 4/3 & 3/2 & 2 & 2 & 5/2 & 3 & 3 & 7/2 & 4 & 4\\
     3 & - & 6/5 & 4/3 & 3/2 & 7/4 & 2 & 2 & 7/3 & 8/3 & 8/3\\
     4 & - & - & 8/7 & 6/5 & 7/5 & 3/2 & 8/5 & 9/5 & 2 & 2\\
     5 & - & - & - & 12/11 & 7/6 & 9/7 & 4/3 & 3/2 & 5/3 & 12/7
\end{tabular}
\caption{Our $\cliquedet{k, \ell}$ exponents if $\omega = 2$. The $(k, \ell)$th entry corresponds to the exponent $\alpha$ such that the runtime to detect a $k$-clique is $\tilde{O}(\Delta_\ell^\alpha)$, where $\Delta_\ell$ is the number of $\ell$-cliques. }
\label{tab:detection_runtime}
\end{table}

Let $\cliquelist{k, \ell}$ be the problem of listing all $k$-cliques in a graph $G$, given all the $\ell$-cliques of $G$. Equivalently, it is the problem of listing $t$ $k$-cliques in a graph given all the $\ell$-cliques, where $t$ is an input to the problem (see a proof in Section~\ref{sec:prelim}). 

Under the Exact-$k$-Clique hypothesis we prove  lower bounds for $\cliquelist{k, \ell}$ for \emph{all} $k\geq 3, 1\leq \ell<k$. This is Theorem~\ref{thm:lower_bound} in the main body. In fact, Theorem~\ref{thm:lb_intro_mn} is a special case of this theorem. 

\begin{theorem}\label{thm:lb_intro} 
For any $k \ge 3, 1 \le \ell < k$, and $\gamma \in [0, k/\ell]$, $\cliquelist{k, \ell}$ in a graph with $\Delta_\ell$ given $\ell$-cliques and $t = \tilde{\Theta}(\Delta_\ell^\gamma)$ $k$-cliques requires
$$\left(\Delta_\ell^{\frac{2}{\ell(k-\ell)}} t^{1 - \frac{2}{k(k-\ell)}}\right)^{1-o(1)}$$ time under the  Exact-$k$-Clique hypothesis.
\end{theorem}

Moreover, we give a conditionally {\em tight} algorithm for graphs with many $k$-cliques, provided that $\omega=2$. In particular, the runtime of the algorithm in the theorem below matches the lower bound of Theorem~\ref{thm:lb_intro}. 

    \begin{theorem}[Informal]
    \label{thm:intro-k-l-large-t-listing}
    If $\omega = 2$, 
    there exists an algorithm for $\cliquelist{k, \ell}$ which runs in time $$\tilde{O}\left(\Delta_\ell^{\frac{2}{\ell(k-\ell)}}\Delta_k^{1 - \frac{2}{k(k-\ell)}}\right)$$ for $\Delta_k \geq \Delta_\ell^{\gamma_{k, \ell}}$ where
    $\gamma_{k, \ell} = \frac{k(k^2 - 2k - 1)}{\ell(k^2 - k - \ell - 1)}.$
    \end{theorem}
Theorem~\ref{thm:intro-k-large-t-listing-mn} is a special case of this theorem.

\subsection{Our Techniques}
In this section, we highlight our main techniques used in the  algorithms and lower bounds. 
\paragraph{Detection algorithms.} 
The previous algorithms for $k$-clique detection in $n$-node graphs \cite{nesetril1985complexity, eisenbrand2004complexity} can be viewed as reductions to triangle detection, as mentioned earlier. Here is how they work  when $k$ is not necessarily divisible by $3$. For some integers $a, b, c \in [1, k]$ where $a+b+c = k$, the algorithm creates a tripartite graph on node parts $A, B, C$ with $n^a, n^b, n^c$ nodes respectively, which represent tuples of $a, b, c$ nodes respectively. It also suffices to keep only the tuples of nodes that form a clique in the original graph. For every node $(u_1, \ldots, u_a) \in A$ and every node $(v_1, \ldots, v_b) \in B$, the algorithm adds an edge between them if and only if $u_1, \ldots, u_a, v_1, \ldots, v_b$ form an $(a+b)$-clique in the original graph. It similarly adds edges between $B, C$ and between $A, C$. It is not difficult to see that there is a triangle in the new graph if and only if there is a $k$-clique in the original graph, so we can simply detect triangles by multiplying an $|A| \times |B|$ matrix with a $|B| \times |C|$ matrix. 

We generalize this approach to $k$-clique detection in terms of the number of $\ell$-cliques for $\ell < k$.

Suppose we are given a list of all $\ell$-cliques in the graph, and we want to find a $k$-clique. Let $a, b, c \in [1, k]$ be as before where $a+b+c = k$.
Let $A$, $B$, and $C$, respectively, be the sets of $a$-, $b$- and $c$-cliques in the graph. We would like to bound their sizes in terms of $\Delta_\ell$. Let us focus on bounding $|A|$; bounding $|B|,|C|$ is done similarly.

For $a \ge \ell$, a (probably folklore) bound shows that $\Delta_a \le O(\Delta_\ell^{a / \ell})$ (we also provide a proof for completeness in Section~\ref{sec:prelim}). 

For $a < \ell$, we set a parameter $\Lambda$ and consider two types of $a$-cliques: ``low-degree" ones that are contained in $<\Lambda$ $\ell$-cliques, and ``high-degree'' ones that are contained in $\geq \Lambda$ $\ell$-cliques. There are at most $O(\Delta_\ell/\Lambda)$ high-degree $a$-cliques.

Consider a low-degree $a$-clique $K$ and its neighborhood consisting of the nodes adjacent to all nodes of $K$. We can {\em recurse} on the neighborhood: find a $(k-a)$-clique, given the list of $(\ell - a)$-cliques formed by excluding $K$ from all $\ell$-cliques that contain $K$. 
We can bound the recursion runtime using the fact that $K$ has low degree. Since we have handled all low-degree $a$-cliques, we can set $A$ to be only the $O(\Delta_\ell/\Lambda)$ high-degree $a$-cliques. Similarly, we can get bounds on $|B|$ and $|C|$.

Finally, following previous $k$-clique detection algorithms \cite{nesetril1985complexity, eisenbrand2004complexity}, we perform a rectangular matrix multiplication between an $|A| \times |B|$ matrix and a $|B| \times |C|$ matrix. By analyzing the recursive steps and setting parameters appropriately, we obtain our detection runtimes. As we show in Examples \ref{ex:clique-det-4-2} and \ref{ex:clique-det-5-2}, our recursion and its analysis are more careful than in prior work, allowing us to obtain improved runtimes for $4$ and $5$-clique detection.

We give some explicit examples of this algorithm in Section~\ref{sec:detection_examples}. We also analyze the asymptotic efficiency of this algorithm in Section~\ref{sec:k-h_detect_bound} and Section~\ref{sec:Cl_l_detectionbound}.

\paragraph{Lower bounds for listing.} 
We obtain our lower bound in Theorem~\ref{thm:lb_intro} for listing from the Exact-$k$-Clique hypothesis. Our lower bound technique can be seen as a generalization of the reduction from Exact Triangle to triangle listing problems in \cite{williams2020monochromatic}. 

We note that there is also a different generalization of the technique of \cite{williams2020monochromatic} that shows a conditional lower bound for the $k$-Set-Intersection problem~\cite{BringmannCM22}. We briefly describe the problem. 
{At a very high level, the lower bound of~\cite{BringmannCM22} applies to the following hypergraph problem: the nodes are partitioned into $k+1$ parts: $V_1, \ldots, V_{k}$ (these correspond to the sets) and $U$ (this corresponds to the universe). There are hyperedges among the nodes in $V_1, \ldots, V_{k}$ (corresponding to $k$-set-intersection queries) and there are edges between $U$ and $V_i$ for $i\in [k]$ (corresponding to elements belonging to each set). Given this hypergraph, the problem asks for each hyperedge, whether its nodes share a common neighbor in $U$ (i.e., whether the sets intersect). As the lower bound of ~\cite{BringmannCM22} is for a problem in a hypergraph with hyperedges of cardinality $>2$, it does not directly apply to our applications. Hypergraph problems are generally harder than their graph counterparts (see e.g. \cite{LincolnWW18}), and there is no easy way to convert a hardness proof for hypergraphs into one for graphs without increasing the instance size significantly. }

Now, we describe the high-level ideas of our reduction.
Without loss of generality, we can assume the input instance of Exact-$k$-Clique is a $k$-partite graph on nodes $V_1 \sqcup \cdots \sqcup V_k$, where each $V_i$ contains $n$ nodes. 
At a high level, we first hash the edge weights so that they behave random enough. For simplicity, we assume all edge weights are independently uniformly at random from $[-n^k, n^k]$ in this overview (we deal with the randomness properly in our proof). Then we split $[-n^k, n^k]$ equally into $s$ contiguous intervals, each of size $O(n^k / s)$ for some parameter $s$. We then enumerate combinations of intervals $(L_{i,j})_{1 \le i < j \le k}$, and consider the subgraph where we only keep edges between $V_i$ and $V_j$ whose weight is in $L_{i, j}$. Note that a subgraph cannot contain a $k$-clique of weight $0$ if $0 \not \in \sum_{1 \le i < j \le k} L_{i, j}$ (we denote the sum of two intervals as the sumset of them). Therefore, we only need to consider combinations of intervals where $0 \in \sum_{1 \le i < j \le k} L_{i, j}$. If we  choose the first $\binom{k}{2}-1$ intervals $(L_{i,j})_{1 \le i < j \le k, (i, j) \ne (k-1, k)}$, the final interval must intersect $-\sum_{1 \le i < j \le k, (i, j) \ne (k-1, k)} L_{i, j}$, which has size $O(\frac{n^k}{s})$. Therefore, there are only $O(1)$ choices for the final interval, and the total number of combinations of intervals we need to consider is $O(s^{\binom{k}{2}-1})$. 

For each combination of intervals, we form the subgraph only containing edges with weights in the intervals, and we list all the $k$-cliques in this subgraph. The expected number of $\ell$-cliques in the subgraph is $O(n^\ell / s^{\binom{\ell}{2}})$ and the expected number of $k$-cliques is $O(n^k / s^{\binom{k}{2}})$. For simplicity, we assume these upper bounds always hold in this overview (instead of only holding in expectation). Also, we can list all the $\ell$-cliques in the subgraphs efficiently, i.e., in nearly linear time in their number, which is faster than $n^k$ when $s$ is small enough. 

Then suppose we have an $O\left(\left( \Delta_\ell^{\frac{2}{\ell(k-\ell)}} t^{1-\frac{2}{k(k-\ell)}}\right)^{1-\epsilon}\right)$ time algorithm for listing all $k$-cliques in a graph with $t$ $k$-cliques and with a given list of $\Delta_\ell$ $\ell$-cliques. We can list all $k$-cliques in all the subgraphs in time
$$\tO\left(s^{\binom{k}{2}-1} \left( \left(n^\ell / s^{\binom{\ell}{2}}\right)^{\frac{2}{\ell(k-\ell)}} \left(n^k / s^{\binom{k}{2}}\right)^{1-\frac{2}{k(k-\ell)}}\right)^{1-\epsilon}\right)=\tO\left(n^{k-k\epsilon} \left(s^{\binom{k}{2}-1}\right)^{\epsilon}\right),$$
which is $\tilde{O}(n^{k-\epsilon'})$ time for  $\epsilon'>0$ for sufficiently small $s$, and  violates the Exact-$k$-Clique hypothesis. 

\paragraph{Listing algorithms for graphs with a large number $t$ of $k$-cliques.} 
Here we discuss how we obtain our optimal algorithm  for $\cliquelist{k, \ell}$ in Theorem~\ref{thm:intro-k-l-large-t-listing}, for all $\ell < k$ and large enough $t$. We give the full algorithm in Section~\ref{sec:upper-bound}. The framework works for all values of $t$, but the runtime is conditionally optimal only for large $t$. We will later explain how to improve upon the framework for small $t$.

As a first step, we obtain output-sensitive algorithms for $k$-clique listing in terms of $n$ ($\ell=1$). We then use these algorithms in a black-box way for $\ell\geq 2$. 

Bj\"{o}rklund, Pagh,  Vassilevska W. and Zwick~\cite{bjorklund2014listing} gave an algorithm for triangle listing using a  \textit{dense-sparse} paradigm. We generalize this algorithm to  $k \geq 4$. Let $t$ be the number of $k$-cliques in the graph which we want to list.

\begin{itemize}
    \item \textbf{Dense algorithm:} When the input graph has many edges, we use sampling and rectangular matrix multiplication to find all the edges that occur in at most $\lambda$ $k$-cliques, for some parameter $\lambda$. We then list all $k$-cliques incident to such edges, and can then delete these edges to obtain a graph with at most $O(t/\lambda)$ edges. We then call the algorithm for sparse graphs.
    \item \textbf{Sparse algorithm:} When the input graph has few edges, we list all $k$-cliques incident to nodes with degree at most $x$ by {\em listing} $(k-1)$-cliques in their neighborhoods, for some parameter $x$. We are then left with a graph with at most $O(m/x)$ nodes, at which point we call the dense algorithm.
\end{itemize}

The key change from the framework of \cite{bjorklund2014listing} is in the sparse algorithm.
There, \cite{bjorklund2014listing} uses brute-force to list triangles through low-degree nodes. We on the other hand, recursively use $\cliquelist{k-1, 1}$ algorithms to list the $(k-1)$-cliques in the neighborhoods of low-degree nodes. This makes our algorithm efficient, but also complicates the analysis significantly.

For $\ell \geq 2$, we exploit recursion even more:
we
recursively use algorithms for both $k$-clique listing in terms of nodes, and $(k-1)$-clique listing in terms of $(\ell-1)$-cliques. At a high level, we first find all nodes that are contained in at most $y$ $\ell$-cliques, for some parameter $y$. Then, in the neighborhoods of such nodes, we can find all $(k-1)$-cliques based on the list of all $(\ell-1)$-cliques in the neighborhood. We can then delete all the low-degree nodes. The resulting graph now only has $O(\Delta_\ell/y)$ nodes. Now, we can call the $k$-clique listing algorithm in terms of $n$.

Because of the extra recursion, the analysis gets more complicated, but we are able to keep the algorithms relatively simple. Thus we get the best of both worlds: simplicity and optimality (at least for large $t$).

The reason why our $\cliquelist{k, 1}$ algorithm is only optimal for large $t$ is that our dense algorithm has an inherent cost of $\Omega(n^{k-1})$ due to the rectangular matrix multiplication that we use. This bottleneck extends to $\cliquelist{k, \ell}$ for all $\ell$ as well since all of these algorithms call $\cliquelist{k, 1}$.

\paragraph{Generalizing the listing algorithm to all values of $t$.}
In Section~\ref{sec:general-list}, we explain how to improve upon our listing framework above when $t$ is smaller. While our general runtime analysis for arbitrary $k,t$ and $\ell$ quickly gets complicated, here we will focus on a small example, to give intuition.

Let us  consider the example of $6$-clique listing in an $n$-node graph $G$ assuming $\omega = 2$. The algorithm in Theorem~\ref{thm:intro-k-l-large-t-listing}  has runtime $\tO(n^{\frac{2}{5}}t^{\frac{14}{15}})$ only when $t \geq n^{4+\frac{13}{14}}$, and otherwise runs in $\tO(n^5)$ time\footnote{Clearly, when $t$ is smaller, the runtime  can only be smaller or equal, so for any $t < n^{4+\frac{13}{14}}$, the runtime of this algorithm is $\tO(n^{\frac{2}{5}}(n^{4+\frac{13}{14}})^{\frac{14}{15}}) = \tO(n^5)$ when $\omega = 2$. } which is worse than the $6$-clique detection runtime $\tO(n^4)$. 

We improve the runtime for $t$ smaller than the threshold of $ n^{4+\frac{13}{14}}$ by instead following the techniques of \cite{nesetril1985complexity, eisenbrand2004complexity}. We create a new graph $G'$ whose nodes correspond to the pairs of nodes of the original graph $G$, i.e. the new graph has $n^2$ nodes. We then add an edge between two nodes $(a, b)$ and $(c, d)$ if  $(a, b, c, d)$ forms a $4$-clique in the original graph.
Now, we run the triangle listing algorithm (in~\cite{bjorklund2014listing} or Theorem~\ref{thm:intro-k-l-large-t-listing}) in the new graph. This has runtime $\tO(n^2t^{2/3})$ when $t \ge (n^2)^{1.5} = n^3$. This also allows us to obtain an algorithm for all $t\leq n^3$, running in time $\tO(n^4)$, the  $6$-clique detection runtime, which is tight under Hypothesis~\ref{hyp:k_clique}.

The corresponding runtime is depicted in blue in Figure~\ref{fig:6_1}.

More generally, for larger $k$, we  create a new graph where the nodes represent $\ell'$-cliques in the original graph. Then, we list $\lceil k/\ell' \rceil$-cliques in the new graph. The best $\ell'$ varies for different $t$, and this gives us the trade-offs as seen in Figure~\ref{fig:general_examples}.

Roughly speaking, the algorithm can be viewed as using different dimensions of rectangular matrix multiplication depending on the value of $t$. For example, in the case of $k = 6$, the algorithm for large $t \geq n^{4 + \frac{13}{14}}$ uses $\tilde{O}(\lambda)$ matrix multiplications of size roughly $n \times n^4/\lambda$ by $n^4/\lambda \times n$ for some parameter $\lambda \geq 1$, and this requires at least $\Omega(n^{5})$ time. For $n^3 \le t \leq n^{4 + \frac{13}{14}}$, the algorithm uses $\tilde{O}(\rho)$ matrix multiplications of size $n^2 \times n^2/\rho$ by $n^2/\rho \times n^2$ for some parameter $\rho \geq 1$, which requires at least $\Omega(n^4)$ time. 

\subsection{Organization} 

In Section~\ref{sec:prelim}, we give necessary definitions and standard algorithms. In Section~\ref{sec:detection}, we show our framework for detecting cliques. In Section~\ref{sec:lower_bounds}, we show our lower bound for listing cliques, proving Theorem~\ref{thm:lb_intro}. In Section~\ref{sec:upper-bound}, we show our optimal algorithm for clique listing in graphs with many $k$-cliques, and we extend this algorithm to graphs with fewer $k$-cliques in Section~\ref{sec:general-list}. Finally, we show a more efficient algorithm for $6$-clique listing in Section~\ref{sec:6clique}.

%% file: 2-prelim.tex
\paragraph{Notation.} 
Throughout this paper, we denote the number of nodes in a graph by $n$, the number of edges by $m$, and the number of $\ell$-cliques by $\Delta_\ell$. For an $\ell'$-clique $K$ for some $1 \le \ell' \le \ell$, we use $\Delta_\ell(K)$ to denote the number of $\ell$-cliques containing $K$. For the special case of $\ell = 2$, we use $\deg(v) := \Delta_2(v)$. For integer $k$, we use $K_k$ to denote a $k$-clique. 

For a nonnegative integer $n$, we use $[n]$ to denote $\{1, 2, \ldots, n\}$. 

\paragraph{Matrix multiplication.} We use $\omega < 2.372$ to denote the matrix multiplication exponent~\cite{duan2023, VXXZ24}. For any constants $a, b, c \ge 0$, we use $\omega(a, b, c)$ to denote the exponent of multiplying an $n^a \times n^b$ matrix by an $n^b \times n^c$ matrix. The current best bounds for rectangular matrix multiplication are given by \cite{VXXZ24}. 

We denote by $\MM(A, B, C)$ the runtime of multiplying an $A \times B$ by a $B \times C$ matrix. If $A \leq B \leq C$, we can loosely bound $\MM(A, B, C)$ in terms of $\omega$ as follows:
$$\MM(A, B, C) \leq O\left(A^\omega \cdot \frac{BC}{A^2}\right) = O(A^{\omega - 2}BC).$$
This bound is obtained by splitting the matrix multiplication into $\frac{B}{A} \cdot \frac{C}{A}$ instances of square matrix multiplication of size $A$, and it is in general weaker than the bound in \cite{LU18}. 

\paragraph{H\"{o}lder's Inequality.} To analyze the runtime of our algorithms, we often utilize a reformulation of H\"{o}lder's inequality.

\begin{lemma}[H\"{o}lder's Inequality]
    Given $p, q \in (1, \infty)$ such that $1/p + 1/q = 1$, the following inequality holds for any $x_1, \ldots, x_k, y_1, \ldots, y_k \geq 0$:
$$\sum_{k=1}^n x_ky_k \leq \left( \sum_{k=1}^n x_k^p\right)^{1/p} \left(\sum_{k=1}^n y_k^q\right)^{1/q}.$$
\end{lemma}
We restate H\"{o}lder's Inequality as follows. This is the version that we use in our runtime analyses. 
\begin{corollary}\label{cor:holders_useful}
    Given $\alpha, \beta \in (0, 1)$ such that $\alpha + \beta = 1$, the following inequality holds for any $x_1, \ldots, x_k, y_1, \ldots, y_k \geq 0$:
    $$\sum_{k=1}^n x_k^\alpha y_k^\beta \leq \left( \sum_{k=1}^n x_k\right)^{\alpha} \left(\sum_{k=1}^n y_k\right)^{\beta}$$
\end{corollary}

\subsection{Problem Definitions} 
Now, we define the main clique problems that we consider in this paper. 

\begin{definition}[$\cliquedet{k, \ell}$]
Given a graph $G=(V, E)$ and  the list $L$ of all $\ell$-cliques in $G$, decide whether $G$ contains a $k$-clique. 
\end{definition}

\begin{definition}[$\cliquelist{k, \ell}$]
Given a graph $G=(V, E)$ and  the list $L$ of all $\ell$-cliques in $G$, list all $k$-cliques in $G$. 
\end{definition}

In $\cliquelist{k, \ell}$, we use $t$ to denote the total number of $k$-cliques in the graph. However, as we will show in Section~\ref{sec:prelim:basic}, we can equivalently (up to $\tO(1)$ factor) use $t$ to denote the number of $k$-cliques we wish to list. 

\subsection{Basic Clique Listing Algorithms}
\label{sec:prelim:basic}
Next, we give some standard algorithms and reductions. 

\begin{lemma}
    Suppose $\cliquedet{k, \ell}$ can be solved in time $D(\Delta_\ell)$. Then, given the list of all $\ell$-cliques in a graph, one can find a $k$-clique in $\tilde{O}(D(\Delta_\ell))$ time.
\end{lemma}
\begin{proof}
    Let the input graph be $G = (V, E)$, with list $L$ of all $\ell$-cliques. Without loss of generality, we may assume that $|V| \leq O(\Delta_\ell)$ by deleting all nodes that are not in any $\ell$-clique (since such a node cannot be in a $k$-clique).
    
    If $|V| \leq k$, brute force and check if the graph has a $k$-clique. Otherwise, run $\cliquedet{k, \ell}$ on the graph. If it has a $k$-clique, arbitrarily partition $V$ into $k+1$ sets, $V_1, \dots, V_{k+1}$. Now, for each $i \in [k+1],$ consider the subgraph on node set $V_{-i}=\cup_{j\in[k+1]\setminus\{i\}} V_j$, and run $\cliquedet{k, \ell}$. Note that each such subgraph contains at most $\Delta_\ell$ $\ell$-cliques. 
    For some $i$, it must be the case that the graph on $V_{-i}$ contains a $k$-clique as we partitioned into $k+1$ parts and a $k$-clique has only $k$ nodes. Recurse on exactly one such subgraph $V_{-i}$ on which the detection algorithm returned ``YES''.
    
    Since the depth of this recursion is $O(\log_{\frac{k+1}{k}} |V|) = O(\log \Delta_\ell)$, and we call $\cliquedet{k, \ell}$ on $O(1)$ instances at each step, we have a runtime of $\tilde{O}(D(\Delta_\ell))$ as desired.
\end{proof}

\begin{lemma}
\label{lem:simple_list_ub}
$\cliquelist{k, \ell}$ can be solved in time $\tO(\Delta_\ell^{k/\ell})$.
\end{lemma}

\begin{proof}

We will prove by induction on the following stronger statement: 
given a list $L$ of $\ell$-cliques in a graph (not necessarily all $\ell$-cliques), one can list all $k$-cliques covered by these $\ell$-cliques in the graph in  time  $\tilde{O}\left(|L|^{k/\ell}\right)$, where a $k$-clique $K$ is covered by a list $L$ of $\ell$-cliques if every $\ell$-clique subgraph of $K$ lies in $L$.

When $\ell = 1$, it suffices to use brute-force to list all $k$-cliques. Now suppose $\ell > 1$. 

First, we find  all $(\ell-1)$-cliques that are contained in at most $x$ (and at least $1$) $\ell$-cliques in the list in $O(|L|)$ time. If an $(\ell-1)$-clique $K$ is contained in $y \leq x$ $\ell$-cliques in the list, then we  can list all $k$-cliques containing $K$ in $\tilde{O}(y^{k-\ell+1}) \le \tilde{O}(y \cdot x^{k-\ell})$ time by brute-force. Over all such $(\ell-1)$-cliques, the total running time is thus $\tilde{O}(|L| x^{k-\ell})$.

The number of $(\ell-1)$-cliques that are contained in at least one of the $\ell$-cliques in $L$ and  are not considered above is $O(|L| / x)$. Let $L'$ be the list of these $O(|L|/x)$ $(\ell-1)$-cliques. 
If a $k$-clique $K$ is not found above, then all of its $(\ell-1)$-clique subgraphs are in the list $L'$, i.e., $K$ is covered by $L'$. 
By induction, we can find the list of all $k$-cliques that are covered by $(\ell-1)$-cliques in $L'$ in $\tilde{O}((|L| / x)^{k/(\ell-1)})$ time. This combined with the $k$-cliques listed in the previous case gives all the $k$-cliques covered by $L$. 

Setting $x = |L|^{1/\ell}$ gives the desired $\tilde{O}(|L|^{k/\ell})$ time, and thus completes the induction. 
\end{proof}

The proof of Lemma~\ref{lem:simple_list_ub} also implies that the number of $k$-cliques in a graph with $\Delta_\ell$ $\ell$-cliques is $O(\Delta_\ell^{k/\ell})$.

\begin{lemma}
\label{lem:list_exponent_monotone}
Fix $1 \le \ell < k$. Suppose there is a $T(\Delta_\ell, x)$ time algorithm for $\cliquelist{k, \ell}$ where the total number of $k$-cliques is $\Theta(\Delta_\ell^x)$. Then for any $x' < x$, $\cliquelist{k, \ell}$ on graphs where the total number of $k$-cliques is $\Theta(\Delta_\ell^{x'})$ can be solved in $O(T(\Delta_\ell, x))$ time. 
\end{lemma}
\begin{proof}
First of all, by Lemma~\ref{lem:simple_list_ub}, $x \le \frac{k}{\ell}$. 
We then add a complete $k$-partite graph to the graph where the number of nodes in each part is $\Delta_\ell^{x/k}$. This way, the number of $k$-cliques in the graph is increased by $\Delta_\ell^{x}$, and the number of $\ell$-cliques is increased by $\Delta_\ell^{\ell x / k} \le \Delta_\ell$. Thus, the number of $k$-cliques in the graph is $(\Delta'_\ell)^x$, where $\Delta'_\ell = \Theta(\Delta_\ell)$ is the new number of $\ell$-cliques in the graph. Therefore, we can run the $T(\Delta_\ell, x)$ time algorithm on the new graph in $\Theta(T(\Delta_\ell, x))$ time. Once we list all the $k$-cliques in the new graph, we can return those that belong to the original graph. 
\end{proof}

Let $f(\Delta_\ell, t)$ be the runtime of $\cliquelist{k, \ell}$ when the graph has (an unknown number of) $t$ cliques in total, and let $g(\Delta_\ell, t)$ be the runtime of listing $\min\{\Delta_k, t\}$ distinct $k$-cliques, given the list of all $\ell$-cliques in the graph and a specified $t$ as input. We assume $f(\tO(\Delta_\ell), \tO(t)) = \tO(f(\Delta_\ell, t))$ and $g(\tO(\Delta_\ell), \tO(t)) = \tO(g(\Delta_\ell, t))$. This is true for all of our algorithms as well as any algorithm that has at most a polynomial dependence on $\Delta_\ell$ and $t$. 

The following lemma shows that $f(\Delta_\ell, t) =  \tilde{\Theta}(g(\Delta_\ell, t))$. Therefore, we use both of these two notions interchangeably for the definition of $\cliquelist{k,\ell}$. In particular, given an instance of $\cliquelist{k, \ell}$ with an unknown number of $k$-cliques, the proof of Lemma~\ref{lem:equivalence_specify} allows us to assume that we know an $2$-approximation of $\Delta_k$, with only $\tO(1)$ loss in the running time.

\begin{lemma}
\label{lem:equivalence_specify}
    $f(\Delta_\ell, t) =  \tilde{\Theta}(g(\Delta_\ell, t))$.
\end{lemma}
\begin{proof}
    We first show $f(\Delta_\ell, t) =  \tO(g(\Delta_\ell, t))$. Let $\mathcal{A}$ be an algorithm for listing a specified number of $k$-cliques. 
    Given an $n$-node graph $G$ and the list of $\Delta_\ell$ cliques, we run $O(\log n)$ instances of $\mathcal{A}$ in parallel. More specifically, we specify these instances to list $2^0, 2^1, \ldots, 2^{\lceil \log(n^k+1)\rceil}$ $k$-cliques respectively. We wait until one of the instances finishes listing all $k$-cliques in the graph. Suppose $t$ is the actual number of $k$-cliques in the graph, and we specify $\mathcal{A}$ to list $2^{\lceil \log t\rceil}$  $k$-cliques, then it will finish within $g(\Delta_\ell, O(t))=\tO(g(\Delta_\ell, t))$ time. 
    Since we run $O(\log n)$ instances in parallel, the overall running time is $\tO(g(\Delta_\ell, t))$. 

    Next, we show $g(\Delta_\ell, t) =  \tO(f(\Delta_\ell, t))$. Let $\mathcal{B}$ be an algorithm for $\cliquelist{k, \ell}$. Given a graph $G = (V, E)$, a list of all $\ell$-cliques and a parameter $t$, we need to list $\min\{t, \Delta_k\}$ $k$-cliques in $\tO(f(\Delta_\ell, t))$ time. First, we run $\mathcal{B}$ for $\tO(f(\Delta_\ell, 2^k t))=\tO(f(\Delta_\ell, t))$ time. By Lemma~\ref{lem:list_exponent_monotone}, if $\Delta_k \le 2^k t$, $\mathcal{B}$ will finish in $\tO(f(\Delta_\ell, 2^k t)$) time, and we are done. Now, we assume the number of $k$-cliques in $G$ is at least $2^k t$.
    
    We create a $k$-partite graph $G'=(V',E')$ as follows. Let $V' = V_1 \sqcup \dots \sqcup V_k$ where each $V_i$ is a copy of $V$. Let $v_i$ be the copy of node $v \in V$ in $V_i$. Add edges $(u_i, v_j)$ between nodes $u_i \in V_i$ and $v_j \in V_j$ if and only if $i\neq j$ and $(u, v) \in E$. Clearly, the number of $k$-cliques in $G'$ is at least $2^k (k!) t$, and we need to list $(k!) t$ distinct $k$-cliques in $G'$ in order to produce $t$ distinct $k$-cliques in $G$. Also, the number of $\ell$-cliques in $G'$ is $O(\Delta_\ell)$. Then we partition each $V_i$ arbitrarily into two sets $V_{i, 0}$ and $V_{i,1}$ of size $n/2$. We run $\mathcal{B}$ on each of the $2^k$ induced subgraphs on the sets $V_{1,b_1}, V_{2, b_2}, \dots, V_{k, b_k}$, where $b_i \in \{0, 1\}$ for $\tO(f(O(\Delta_\ell), 2^k (k!) t)) = \tO(f(\Delta_\ell, t))$ time. By the pigenhole principle, one of the subgraphs contain at least $(k!) t$ $k$-cliques. If $\mathcal{B}$ finishes on that subgraph, we are done. Otherwise, $\mathcal{B}$ does not finish on that subgraph, and by Lemma~\ref{lem:list_exponent_monotone}, that subgraph must have more than $2^k (k!) t$ distinct $k$-cliques, so we can recurse on that induced subgraph. Overall, the running time is $\tO(f(\Delta_\ell, t))$ because the recursion depth is $O(\log n)$.
\end{proof}

\cite{bjorklund2014listing} gave similar reductions from listing a specified number of  $t$ triangles to listing all $\Delta_3$ triangles in $n$-node or $m$-edge graphs. Their reduction is more efficient than ours when $t$ is much smaller than $\Delta_3$. However, their reduction requires an algorithm for {\em counting} the number of triangles. We instead provide a black box reduction that does not rely on counting, that works for arbitrary $k,\ell$, and is
 more self-contained and efficient enough for our purpose.

%% file: 3-detection.tex
In this section, we first describe our algorithm for $\cliquedet{k, \ell}$, and then analyze its running time in some interesting cases. 

Throughout this section, we use $g(k, \ell)$ to denote our algorithm's running time exponent on the number of $\ell$-cliques of $\cliquedet{k, \ell}$, i.e., our algorithm for $\cliquedet{k, \ell}$ runs in $\tO(\Delta_\ell^{g(k, \ell)})$ time. 

\subsection{General Detection Framework}
Now we describe a generic algorithm for $\cliquedet{k, \ell}$ for $k \ge 3$ (for $k=2$, we trivially list all edges in the graph, so $g(2, 1) = 2$) in Algorithm~\ref{alg:generic_detection}.

\begin{algorithm}
\caption{Generic $
\cliquedet{k, \ell}$ algorithm.}\label{alg:generic_detection}
\begin{algorithmic}
\item \textbf{Input:} Graph $G = (V, E)$ and the list $L$ of all $\ell$-cliques. 
\item \textbf{Output:} Output {\sc yes} if $G$ contains a $k$-cliques, and {\sc no} otherwise.
\item \textbf{The Algorithm:} 
\begin{itemize}
    \item Let integers $k \ge a \geq b \geq c \ge 1$ be such that $k = a + b + c$ (the algorithm chooses $a, b, c$ optimally). Then goal is then to bound the number of $d$-cliques for $d \in \{a, b, c\}$. 
    \begin{itemize}
        \item If $d \ge \ell$, we can use Lemma~\ref{lem:simple_list_ub} to upper bound the number of $d$-cliques with $S_d = \tilde{\Theta}(\Delta_\ell^{d/\ell})$, and add these $d$-cliques to a list $L_d$ in the same time. 
        \item If $d < \ell$, for every $d$-clique $K$ with $\Delta_
        \ell (K) \leq \Delta_\ell^{x_d}$ (for some parameter $x_d \in [0, 1]$ to be chosen), we check if $K$ is in a $k$-clique by recursively running $\cliquedet{k-d, \ell-d}$ in its neighbourhood. Then, let $L_d$ denote the set of remaining $d$-cliques. Then, $S_{d} := |L_d|  = \Theta(\Delta_\ell^{1-x_d})$. The running time of this step is 
        \begin{align*}
            \tO\left(\sum_{\substack{K: d\text{-clique}\\\Delta_\ell(K) \le \Delta_\ell^{x_d}}} \Delta_\ell(K)^{g(k-d, \ell - d)}\right) & \le \tO\left(\sum_{\substack{K: d\text{-clique}\\\Delta_\ell(K) \le \Delta_\ell^{x_d}}} \Delta_\ell(K) \cdot \Delta_\ell^{x_d (g(k-d, \ell - d) - 1)}\right)\\
            & \le \tO\left(\Delta_\ell^{1+x_d (g(k-d, \ell - d) - 1)}\right). 
        \end{align*}
    \end{itemize}
    \item Finally, we conduct a usual matrix multiplication of dimensions $S_a, S_b, S_c$ in time $\MM(S_a, S_b, S_c)$ as follows. If we find a $k$-clique, output {\sc yes}, otherwise we output {\sc no.}
    \begin{itemize}
        \item Create a matrix $X$ whose rows are indexed by $a$-cliques in $L_a$ and columns are indexed by  $b$-cliques in $L_b$. Set $A[K_a, K_b] = 1$ if  the nodes of $K_a$ and $K_b$ form an $(a+ b)$-clique, and $0$ otherwise.
        \item Create a matrix $Y$ whose rows are indexed by $b$-cliques in $L_b$ and columns are indexed by $c$-cliques in $L_c$, and set the entries similarly.
        \item Compute $Z = XY$. For each pair of  remaining $a$-clique $K_a$ and $c$-clique $K_c$ that form an $(a + c)$-clique, check if $Z[K_a, K_c] > 0$. If such an entry exists, output {\sc yes}. Otherwise, output {\sc no}.
    \end{itemize}
    
\end{itemize}
\end{algorithmic}
\end{algorithm}

The correctness of this algorithm is immediate. 
We also remark that the algorithm can  be used to count the number of $k$-cliques, by replacing all the recursive calls with the counting version of the algorithm, using the matrix multiplication to count the number of $k$-cliques in the remaining graph, and properly summing up and scaling the numbers. Clearly, the counting version of the algorithm will have the same running time. 

\subsection{Examples}
\label{sec:detection_examples}
Let us give some explicit examples to illustrate the algorithm. 

\paragraph{$\cliquedet{k,1}$.} The simplest example is $\cliquedet{k,1}$ for $k \ge 3$. Let $\lfloor k/3 \rfloor \leq c \leq b \leq a \leq \lceil k/3\rceil$ be integers such that $a+b+c = k$, which is one of the possible choices of $a, b, c$ for the algorithm.
Note that $c = \lfloor k/3 \rfloor, b =  \lceil (k-1)/3\rceil, a = \lceil k/3 \rceil$. Since $a, b, c \ge \ell = 1$, the algorithm would choose to use Lemma~\ref{lem:simple_list_ub} to bound the number of cliques of sizes $a, b, c$ as $n^a, n^b, n^c$ respectively. Thus, the running time of the algorithm is $\tO(n^{\omega(a, b, c)}) = \tO(n^{\beta(k)})$, matching the previous running time \cite{eisenbrand2004complexity}.

\paragraph{$\cliquedet{k,\ell}$ for $\ell \leq \lfloor k / 3\rfloor.$}
Similar as above, let $c = \lfloor k/3 \rfloor, b =  \lceil (k-1)/3\rceil, a = \lceil k/3 \rceil$ and the algorithm would choose to use Lemma~\ref{lem:simple_list_ub} to bound the number of cliques of sizes $a, b, c$. Thus, the running time of the algorithm is  $\tO(\Delta_\ell^{\omega(a/\ell, b/\ell, c/\ell)}) \leq \tO(\Delta_\ell^{\omega(\lceil k/3 \rceil,   \lceil (k-1)/3\rceil, \lfloor k/3 \rfloor)/\ell})$. 
This running time is optimal barring improvements for $\cliquedet{k, 1}$:

\begin{table}[ht]
    \centering
    \begin{tabular}{c|c|c|c|c|c|c|c|c|c|c}
         \backslashbox{$\ell$}{$k$} &  3 & 4 & 5 & 6 & 7 & 8 & 9 & 10 & 11 & 12\\
         \hline 
         1 & 2.372 & 3.251 & 4.086 & 4.744 & 5.590 & 6.397 & 7.115 & 7.952 & 8.745 & 9.487\\
         2 & 1.407 & 1.657 & 2.057 & 2.372 & 2.795 & 3.199 & 3.558 & 3.976 & 4.373 & 4.744\\
         3 & - & 1.248 & 1.422 & 1.668 & 1.918 & 2.149 & 2.372 & 2.651 & 2.915 & 3.163\\
         4 & - & - & 1.174 & 1.298 & 1.487 & 1.657 & 1.840 & 2.028 & 2.205 & 2.372\\
         5 & - & - & - & 1.130 & 1.232 & 1.377 & 1.503 & 1.660 & 1.811 & 1.953
    \end{tabular}
    \caption{Our $\cliquedet{k, \ell}$ exponent for various values of $k, \ell$ with the best current bound on $\omega$ and rectangular matrix multiplication \cite{VXXZ24}. See also \cite{van2019dynamic} for a way to bound $\omega(a, b, c)$ for arbitrary $a, b, c > 0$ from values of $\omega(1, x, 1)$. 
    The $(k, \ell)$th entry corresponds to the exponent $\alpha$ such that the runtime to detect a $k$-clique is $\tilde{O}(\Delta_\ell^\alpha)$ , where $\Delta_\ell$ is the number of $\ell$-cliques. 
    }
    \label{tab:det_exponent}
\end{table}

\begin{proposition}
Fix any positive integers $k \ge 3$ and $\ell \le \lfloor k/3\rfloor$, and let $\beta(k) = \omega(\lceil k/3 \rceil,   \lceil (k-1)/3\rceil, \lfloor k/3 \rfloor)$.
If $\cliquedet{k, 1}$ requires $n^{\beta(k) - o(1)}$ time, then
$\cliquedet{k, \ell}$ requires $\Delta_\ell^{\beta(k)/\ell-o(1)}$ time. 
\end{proposition}
\begin{proof}
Suppose for the sake of contradiction that $\cliquedet{k, \ell}$ has an  $O(\Delta_\ell^{\beta(k)/\ell-\eps}$) time algorithm $\mathcal{A}$ for some $\eps > 0$. Then given a $\cliquedet{k, 1}$ instance, we can first use Lemma~\ref{lem:simple_list_ub} to list all $\ell$-cliques in $O(n^\ell)$ time, and the number of $\ell$-cliques is bounded by $O(n^\ell)$. Then we can use $\mathcal{A}$ to solve the $\cliquedet{k, 1}$ instance in $O((n^\ell)^{\beta(k)/\ell-\eps})=n^{\beta(k)-\eps \ell}$ time, a contradiction.
\end{proof}

\begin{example}[$\cliquedet{3, 2}$]
\label{ex:clique-det-3-2}
\em
In this case, the algorithm can only choose $a=b=c=1$, and it would naturally choose $x_a=x_b=x_c$. The time it takes to bound the number of $1$-cliques (nodes) is $\tO(\Delta_2^{1+x_a (g(2, 1) - 1)}) = \tO(m^{1+x_a})$. Then we have $S_a, S_b, S_c \le \Theta(m^{1-x_a})$. Thus, the running time for the matrix multiplication of dimensions $S_a, S_b, S_c$ is $\tO(m^{(1-x_a)\omega})$. Overall, the running time is $\tO(m^{\frac{2\omega}{\omega+1}})$ by setting $x_a = \frac{\omega-1}{\omega+1}$. This is essentially Alon, Yuster and Zwick \cite{alon1997finding}'s triangle detection algorithm for sparse graphs. 
\end{example}

\begin{example}[$\cliquedet{4, 2}$]
\label{ex:clique-det-4-2}
\em
In this case, the algorithm can only choose $a=2, b=c=1$, and it would naturally choose $x_b=x_c$. The algorithm uses Lemma~\ref{lem:simple_list_ub} to (trivially) bound the number of edges as $m$. 
The time it takes to bound the number of nodes is $\tO(\Delta_2^{1+x_b (g(3, 1) - 1)}) = \tO(m^{1+x_b(\omega-1)})$. Then we have $S_a \le \Theta(m), S_b, S_c \le \Theta(m^{1-x_b})$. Thus, the running time for the matrix multiplication of dimensions $S_a, S_b, S_c$ is $\tO(m^{\omega(1, 1-x_b, 1-x_b)})$. The algorithm chooses $x_b$ so that $1+x_b(\omega-1) = \omega(1, 1-x_b, 1-x_b)$. If we simply bound $\omega(1, 1-x_b, 1-x_b)$ by $x_b + \omega(1-x_b)$, we can get $g(4, 2) \le \frac{\omega+1}{2}$ by setting $x_b = \frac{1}{2}$. For the current best bound of square and rectangular matrix multiplication~\cite{VXXZ24}, we can set $x_b = 0.478$ to get an upper bound $g(4, 2) \le 1.657$. As seen in Table~\ref{table:improved_det_4_5}, this is an improvement over the previous best algorithm of Eisenbrand and Grandoni~\cite{eisenbrand2004complexity}. The key difference between our algorithm and \cite{eisenbrand2004complexity}'s algorithm is that, after they perform a similar first stage, they recursively call a $\cliquedet{4, 1}$ algorithm on graphs with $S_b$ nodes, losing the information that the graph has $S_a = m$ edges to begin with. We instead utilize this information with rectangular matrix multiplication to get a better running time. 

\end{example}

\begin{example}[$\cliquedet{5, 2}$]
\label{ex:clique-det-5-2}
\em
In this case, let the algorithm  choose $a=b=2, c=1$ (the choice $a=3, b=c=1$ gives a worse bound). The algorithm uses Lemma~\ref{lem:simple_list_ub} to (trivially) bound the number of edges as $m$. 
The time it takes to bound the number of nodes is $\tO(\Delta_2^{1+x_c (g(4, 1) - 1)}) = \tO(m^{1+x_c(\omega(1, 2, 1)-1)})$. Then we have $S_a,S_b \le \Theta(m), S_c \le \Theta(m^{1-x_c})$. Thus, the running time for the matrix multiplication of dimensions $S_a, S_b, S_c$ is $\tO(m^{\omega(1, 1, 1-x_c)})$. The algorithm chooses $x_c$ so that $1+x_c(\omega(1, 2, 1)-1) = \omega(1, 1, 1-x_c)$. If we simply bound $\omega(1, 2, 1)$ by $\omega + 1$ and $\omega(1, 1, 1-x_c)$ by $2x_c + (1-x_c)\omega$, we can get $g(5, 2) \le \frac{\omega+2}{2}$ by setting $x_c = \frac{1}{2}$. For the current best bound of rectangular matrix multiplication~\cite{VXXZ24}, we can set $x_c = 0.469$ to get an upper bound $g(5, 2) \le 2.057$. As seen in Table~\ref{table:improved_det_4_5}, this is an improvement over the previous best known algorithm of Eisenbrand and Grandoni \cite{eisenbrand2004complexity}.
\end{example}

\begin{example}[More Small Examples]
\label{ex:more-small-examples}
\em
See Tables~\ref{tab:detection_runtime} and \ref{tab:det_exponent} for more examples of the running times of our algorithm. These running times were obtained by finding the optimal values of $a, b, c$ using dynamic programming. 

From previous examples, one might wonder whether the algorithm always sets $a, b, c$ as close to $k/3$ as possible. The following example shows that it is not the case (for $\omega = 2$). 

In $\cliquedet{8, 4}$, if the algorithm chooses $a=4, b = c = 2$, then the running time is 
$$\tO\left(\Delta_4^{1+x_b(g(6, 2)-1)}+\Delta_4^{1+x_c(g(6, 2)-1)} + \Delta_4^{\omega(1, 1-x_b, 1-x_c)}\right).$$
By setting $x_b=x_c = \frac{1}{2}$, this running time is bounded by $\tO(\Delta_4^{3/2})$ when $\omega = 2$ (See Table~\ref{tab:det_exponent} for the value of $g(6, 2)$ when $\omega = 2$). 

However, if the algorithm chooses a more balanced choice $a=b=3, c = 2$, then the running time is $$\tO\left(\Delta_4^{1+x_a(g(5, 1)-1)}+\Delta_4^{1+x_b(g(5, 1)-1)}+\Delta_4^{1+x_c(g(6, 2)-1)} + \Delta_4^{\omega(1-x_a, 1-x_b, 1-x_c)}\right).$$
One optimal way to set the parameters when $\omega = 2$ is $x_a = x_b = \frac{1}{5}$ and $x_c = \frac{3}{5}$, which only gives an $\tO(\Delta_4^{8/5})$ running time when $\omega = 2$ (See Table~\ref{tab:det_exponent} for the values of $g(5, 1)$ and $g(6, 2)$ when $\omega = 2$). 
\end{example}

\subsection{Upper Bound for \texorpdfstring{$\cliquedet{k, k - h}$}{(k, k-h)-Clique-Detection}}
\label{sec:k-h_detect_bound}

In this section, we analyze the running time of our algorithm for $\cliquedet{k, k - h}$ for some constant $h=O(1)$. For convenience, let $e_h(k) = g(k, k - h)$.

We start with the following lemma.
\begin{lemma}
\label{lem:det_exponent_monotone}
For every $k > h$, $e_h(k + 1) \le e_h(k)$. 
\end{lemma}
\begin{proof}
We prove the statement by induction. We skip the base case $k = h+1$ as it works similarly as the induction step
(except for $h=1$, in which case $e_h(2) = 2$ and $e_h(3) = \frac{2\omega}{\omega+1} \le e_h(2)$, as the algorithm handles $\cliquedet{2, 1}$ specially).  Suppose the statement is already true for all smaller $k$. 

Let $\ell = k - h$ and $\ell' = k + 1 - h$. Suppose for $\cliquedet{k, \ell}$, the optimal parameters are $a, b, c, x_a, x_b, x_c, S_a, S_b, S_c$ ($x_d$ is relevant only if $d < \ell$ for $d \in \{a, b, c\}$).  Consider $\cliquedet{k+1, \ell+1}$ with parameters $a' = a+1, b' =b, c'=c$ and $x'_{a'}, x'_{b'}, x'_{c'}, S'_{a'}, S'_{b'}, S'_{c'}$ to be determined. Let $\Delta_\ell$ be the number of $\ell$-cliques in the $\cliquedet{k, \ell}$ instance and let $\Delta'_{\ell'}$ be the number of $(\ell+1)$-cliques in the $\cliquedet{k+1, \ell+1}$ instance. 

We first compare exponents related to $S_a$ and $S'_{a'}$.
\begin{itemize}
    \item If $a \ge \ell$. Then $S_a$ in $\cliquedet{k, \ell}$ is bounded by $\tO(\Delta_\ell^{a / \ell})$. In the $\cliquedet{k + 1, \ell + 1}$ algorithm, $S_{a'}'$ is bounded  by $\tO((\Delta'_{\ell'})^{(a+1)/(\ell+1)})$, a smaller exponent. 
    \item If $a < \ell$, the exponent of the running time for bounding $S_a$ in $\cliquedet{k, \ell}$ is $1+x_a (e_h(k-a) - 1)$, and $S_a$ is bounded by $\Delta_\ell^{1-x_a}$. Let $x'_{a'}$ be equal to $x_a$ in the algorithm for $\cliquedet{k + 1, \ell + 1}$. Then notice that the exponent for running time is $1+x'_{a'}(e_h(k + 1 - a') - 1) = 1+x_a(e_h(k-a)-1)$ and the bound on $S'_{a'}$ is $(\Delta'_{\ell'})^{1-x_a}$, both with same exponents as previous bounds. 
\end{itemize}
We then compare exponents related to $S_b$ and $S'_{b'}$. 
\begin{itemize}
    \item If $b > \ell$. Then $b' = b \ge \ell+1 = \ell'$. Then $S_b$ in $\cliquedet{k, \ell}$ is bounded by $\tO(\Delta_\ell^{b / \ell})$. In the $\cliquedet{k + 1, \ell + 1}$ algorithm, $S_{b'}'$ is bounded  by $\tO((\Delta'_{\ell'})^{b/(\ell+1)})$, a smaller exponent. 
    \item If $b = \ell$. In this case, $S_b = \tO(\Delta_\ell)$ and we will have $b' < \ell'$. Let $x'_{b'} = 0$ in $\cliquedet{k + 1, \ell + 1}$. Then $S'_{b'}$ is bounded by $\tO((\Delta'_{\ell'})^1)$, the same exponent as the bound of $S_b$. Also, the cost for having this bound is $\tO((\Delta'_{\ell'})^{1+x'_{b'}(e_h(k+1-b'))}) = \tO(\Delta'_{\ell'})$, so we can ignore the cost as it is near-linear time.
    \item If $b < \ell$, the exponent of the running time for bounding $S_b$ in $\cliquedet{k, \ell}$ is $1+x_b (e_h(k-b) - 1)$, and $S_b$ is bounded by $\Delta_\ell^{1-x_b}$. Let $x'_{b'}$ be equal to $x_b$ in the algorithm for $\cliquedet{k + 1, \ell + 1}$. Then notice that the exponent for running time is $1+x'_{b'}(e_h(k + 1 - b') - 1) = 1+x_b(e_h(k-b+1)-1)$. By the induction assumption, 
    $e_h(k-b+1) \le e_h(k-b)$, so $1+x_b(e_h(k-b+1)-1)$ is upper bounded by the running time exponent of the corresponding case in $\cliquedet{k, \ell}$. Note that this case does not happen in the base case $k=h+1$, as $b < \ell = 1$ can never happen, so we can safely apply the induction assumption. 
    The bound on $S'_{b'}$ is $(\Delta'_{\ell'})^{1-x_b}$,  with the same exponent as $S_b$ in $\cliquedet{k, \ell}$. 
\end{itemize}
The comparison of the exponents related to $S_c$ and $S'_{c'}$ works similarly. Thus, $e_h(k+1) \le e_h(k)$. 
\end{proof}

\begin{proposition}
\label{prop:eh_upper_bound}
$e_h(k) = 1+O\left(1/ k^{\log_{\frac{3}{2}}(\frac{\omega}{\omega-1})}\right)$. 
\end{proposition}
\begin{proof}
Let $\ell = k - h$. 
 Let $k_0 = 100h$. For all $k \le k_0$, $e_h(k) = O(1)$. 

For $k > k_0$, we choose $a, b, c$ in our $\cliquedet{k, k - h}$ algorithm so that $\lfloor k/3\rfloor = c \le b \le a = \lceil k/3\rceil$. Clearly, $a, b, c < \ell = k - h$. The running time of the algorithm is thus 
$$\tO\left(\Delta_\ell^{1+x_a  \cdot (e_h(k-a)-1)} 
+ \Delta_\ell^{1+x_b  \cdot (e_h(k-b)-1)} 
+ \Delta_\ell^{1+x_c  \cdot (e_h(k-c)-1)} 
+ MM\left(\Delta_\ell^{1-x_a}, \Delta_\ell^{1-x_b},\Delta_\ell^{1-x_c} \right)\right).$$
By Lemma~\ref{lem:det_exponent_monotone}, $e_h(k-c) \le e_h(k-b) \le e_h(k-a)$, so the running time is bounded by 
$$\tO\left(\Delta_\ell^{1+\max\{x_a, x_b, x_c\}  \cdot (e_h(k-a)-1)} 
+MM\left(\Delta_\ell^{1-x_a}, \Delta_\ell^{1-x_b},\Delta_\ell^{1-x_c} \right)\right).$$
Set $x_a = x_b = x_c = \frac{\omega -1}{\omega + e_h(k-a) - 1}$. The running time then becomes 
$$\tO\left(\Delta_\ell^{\frac{\omega \cdot e_h(k-a)}{\omega + e_h(k-a) - 1}}\right).$$
Thus,
$e_h(k) \le \frac{\omega \cdot e_h(k-a)}{\omega + e_h(k-a) - 1}$. Consequently, $$e_h(k) - 1 \le \frac{(\omega - 1) \cdot (e_h(k-a) - 1)}{\omega + e_h(k-a) - 1} \le \frac{\omega - 1}{\omega}  \cdot (e_h(k-a) - 1) = \frac{\omega - 1}{\omega}  \cdot (e_h(k-\lceil k/3\rceil) - 1).$$
Therefore $e_h(k) - 1 \le O\left(\left(\frac{\omega - 1}{\omega}\right)^{\log_{\frac{3}{2}} k}\right) = O\left(1/ k^{\log_{\frac{3}{2}}(\frac{\omega}{\omega-1})}\right)$.
\end{proof}

We also show that our choices of $a, b, c$ are not too far away from optimal, at least when $\omega = 2$. In the following proposition, recall $e_h(k)$ is the exponent of our algorithm, instead of the best exponent for $\cliquedet{k, k - h}$. 

\begin{proposition}
$e_h(k) = 1+\Omega\left(1/ k^{\log_{\frac{3}{2}}(2)}\right)$. 
\end{proposition}
\begin{proof}
Let $\ell = k - h$,  $\rho = \log_{\frac{3}{2}}(2)$, and $f_h(k) = \frac{1}{e_h(k) - 1}$. 
 Let $k_0 = 100h$. It is not difficult to see that for all $k \le k_0$, $f_h(k) \le M k^\rho - 1$ for some sufficiently large constant $M > 1$ because our algorithm does not achieve almost linear time, i.e., it always has $e_h(k) > 1$ and thus $f_h(k) < \infty$. 
 
 Let $k > k_0$, and let $a, b, c$ be the optimal choices for $\cliquedet{k, k - h}$. We will show by induction that $f_h(k) \le M k^\rho - 1$. 
 Consider two cases. 
 
 For the first case, assume $a < \ell$. Let $x_a, x_b, x_c$ be the optimal parameters for $\cliquedet{k, k - h}$, and if there are multiple choices, we choose one set of parameters with smallest $x_a+x_b+x_c$. 
 Then, the bound of our running time is (up to $\tO(1)$ factors) 
$$\Delta_\ell^{1+x_a  \cdot (e_h(k-a)-1)} 
+ \Delta_\ell^{1+x_b  \cdot (e_h(k-b)-1)} 
+ \Delta_\ell^{1+x_c  \cdot (e_h(k-c)-1)} 
+ MM\left(\Delta_\ell^{1-x_a}, \Delta_\ell^{1-x_b},\Delta_\ell^{1-x_c}\right).$$ 

Suppose $x_a > x_b$. By Lemma~\ref{lem:det_exponent_monotone}, $e_h(k-a) \ge e_h(k-b)$. Therefore, we can slightly increase $x_b$, and the running time of the algorithm will not be worse. This contradicts with the optimality of $x_a, x_b, x_c$ and minimality of $x_a+x_b+x_c$. Thus, we must have $x_a \le x_b$. Similarly, we have $x_b \le x_c$. 

Then we can lower bound $MM\left(\Delta_\ell^{1-x_a}, \Delta_\ell^{1-x_b},\Delta_\ell^{1-x_c}\right)$ by   $\Delta_\ell^{2-x_a-x_b}$.

The optimal way to balance $\Delta_\ell^{1+x_a  \cdot (e_h(k-a)-1)}, \Delta_\ell^{1+x_b  \cdot (e_h(k-b)-1)}, \Delta_\ell^{1+x_c  \cdot (e_h(k-c)-1)}$ and $ \Delta_\ell^{2-x_a-x_b}$ is to set 
$x_a = \frac{e_h(k-b)-1}{e_h(k-a)e_h(k-b)-1}$, $x_b = \frac{e_h(k-a)-1}{e_h(k-a)e_h(k-b)-1}$ and $x_c = \min\{1, \frac{(e_h(k-a)-1)(e_h(k-b)-1)}{(e_h(k-a)e_h(k-b)-1)(e_h(k-c)-1)}\}$, which gives \[e_h(k) \ge \frac{2e_h(k-a)e_h(k-b)-e_h(k-a)-e_h(k-b)}{e_h(k-a)e_h(k-b)-1}.\] Substituting $e_h$ by $f_h$ gives the following cleaner formula:
$$f_h(k) \le 1+f_h(k-a)+f_h(k-b).$$
As the algorithm chooses the optimal $a, b, c$, we have that 
$$f_h(k) \le \max_{\substack{1 \leq c \leq b \leq a \leq k\\ a + b + c = k}} \left\{ 1 + f_h(k-a) + f_h(k-b)\right\}.$$

By Lemma~\ref{lem:det_exponent_monotone}, $f_h(k-b)$ is nondecreasing when $b$ increases, so we can pick $b$ to be as large as possible for fixed $a$. Therefore, for fixed $a$, we choose  $c = \lfloor \frac{k-a}{2} \rfloor$ and $b = \lceil \frac{k-a}{2} \rceil$. Therefore, we can rewrite
$$f_h(k) \le \max_{ k/3  \leq a \leq k-2} \left\{1 + f_h(k-a) + f_h\left(\left\lfloor \frac{k+a}{2}\right\rfloor\right)\right\}.$$

By the induction assumption, $f_h(k') \le M (k')^\rho - 1$ for all $k'<k$.

Then,
\begin{align*}
    f_h(k)&\leq \max_{k/3 \le a \le k-2} \left\{1 + f_h(k-a) + f_h\left(\left\lfloor\frac{k+a}{2}\right\rfloor\right)\right\} \\
    & \leq \max_{0 \le p \le k/3}\left\{ 1 + M \left(\frac{2k}{3} - 2p\right)^\rho + M\left(\frac{2k}{3} + p\right)^\rho - 2 \right\}\\
    & \leq Mk^\rho \cdot \max_{0 \le p' \le 1/3} \left\{\left(\frac{2}{3}-2p'\right)^\rho + \left(\frac{2}{3}+p'\right)^\rho \right\} - 1\\
    &\le Mk^\rho - 1,
\end{align*}
which completes the induction step for this case. 

For the other case, assume $a \ge \ell$. Note that we must have $b, c < \ell$ as $2\ell > k$. Let $x_b, x_c$ be the optimal parameters. Similar as before, we can assume $x_b \le x_c$. 
 Then, the bound of our running time is (up to $\tO(1)$ factors) 
\begin{align*}
&\Delta_\ell^{1+x_b  \cdot (e_h(k-b)-1)} 
+ \Delta_\ell^{1+x_c  \cdot (e_h(k-c)-1)} 
+ MM\left(\Delta_\ell^{a/\ell}, \Delta_\ell^{1-x_b},\Delta_\ell^{1-x_c}\right)\\
\ge & \Delta_\ell^{1+x_b  \cdot (e_h(k-b)-1)} 
+ \Delta_\ell^{1+x_c  \cdot (e_h(k-c)-1)} 
+ \Delta_\ell^{a/\ell + 1 - x_b}
\end{align*}
The optimal way to balance is to set $x_b = \frac{a}{\ell e_h(k-b)}$ and $x_c = \min\{1, \frac{a(e_h(k-b)-1)}{\ell e_h(k-b)(e_h(k-c)-1)}\}$. 
This gives $e_h(k) \ge \frac{ae_h(k-b)-a}{\ell e_h(k-b)}+1$. Note that it is possible that $x_b > 1$ in this setting, but if that happens, $e_h(k) > e_h(k-b)$, which by Lemma~\ref{lem:det_exponent_monotone}, can never be optimal. In terms of $f_h$, this implies that $f_h(k) \le \frac{\ell  (f_h(k-b)+1)}{a}$. 
As the algorithm chooses the optimal $a, b, c$, we have that 
$$f_h(k) \le \max_{\substack{1 \leq c \leq b < \ell \leq a \leq k\\ a + b + c = k}}  \frac{\ell  (f_h(k-b)+1)}{a}.$$
By Lemma~\ref{lem:det_exponent_monotone}, $f_h(k-b)$ is nondecreasing when $b$ increases, so we can pick $b$ to be as large as possible for fixed $a$. Therefore, for fixed $a$, we choose  $c = \lfloor \frac{k-a}{2} \rfloor$ and $b = \lceil \frac{k-a}{2} \rceil$. Thus, we can rewrite
$$f_h(k) \le \max_{ \ell \leq a \leq k-2} \frac{\ell  (f_h\left(\left\lfloor \frac{k+a}{2}\right\rfloor\right)+1)}{a} \le \max_{ \ell \leq a \leq k-2}  \left\{f_h\left(\left\lfloor \frac{k+a}{2}\right\rfloor\right)+1\right\}.$$
By induction, it can be further upper bounded by 
$$\max_{ \ell \leq a \leq k-2}  \left\{M\left(\left\lfloor \frac{k+a}{2}\right\rfloor\right)^\rho -1 +1\right\} \le M(k-1)^\rho < Mk^\rho - 1,$$
as $M, \rho > 1$. This finishes the induction step for this case. 

Overall, we have shown that $f_h(k) \le M k^\rho - 1$ for all $k$, which implies $e_h(k) = 1+\Omega\left(1/ k^{\log_{\frac{3}{2}}(2)}\right)$. 
\end{proof}

\subsection{Upper Bound for \texorpdfstring{$\cliquedet{C\ell, \ell}$}{(Cl, l)-Clique-Detection}}\label{sec:Cl_l_detectionbound}
Define a sequence of functions $(f_i)_{i \ge 0}$ as follows:
$$f_i(C) = \frac{2^i \omega^{i+1} C}{3^{i+1}(\omega-1)^i+\left(3 (2^i - 3^i) (\omega-1)^i - 2^i (\omega-1)^i \omega + 2^i \omega^{i + 1}\right)C}. $$
The functions have the following recurrence relation, whose proof we omit as it is straightforward algebra. 
\begin{claim}
$f_0(C) = \frac{\omega C}{3}$ and $f_i(C) = \frac{\omega}{1+\frac{\omega - 1}{f_{i-1}\left(\frac{2C}{3-C}\right)}}$ for $i > 0$.
\end{claim}

Then we can express the running time of $\cliquedet{C\ell, \ell}$ for sufficiently large $\ell$ in terms of the functions $f_i$:
\begin{theorem}\label{thm:det_mult_upper_bound}
Let $C > 1$ be any constant such that $\frac{1}{C} \in \left(1-\left(\frac{2}{3}\right)^i, 1-\left(\frac{2}{3}\right)^{i+1} \right]$ for some constant integer $i \ge 0$. Then for any $\ell \ge 1$ and $C\ell \le  k \le (C+o_\ell(1)) \ell$, $g(k, \ell) \le f_i(C) + o_\ell(1)$. 
\end{theorem}
\begin{proof}
We prove by induction on $i$. 

When $i = 0$, $k \ge C\ell \ge 3\ell$. Therefore, we can apply the  $\cliquedet{k,\ell}$ example in Section~\ref{sec:detection_examples} for $\ell \leq \lfloor k / 3\rfloor$  to get $g(k, \ell) \le \omega(\lceil k/3 \rceil,   \lceil (k-1)/3\rceil, \lfloor k/3 \rfloor)/\ell$. This leads to 
\begin{align*}
    g(k, \ell) &\le \omega(k/3+1, k/3+1,k/3+1) / \ell \\
    & = \frac{(k/3+1)\omega}{\ell}\\
    & \le \frac{((C+o_\ell(1)) \ell / 3 + 1) \omega}{\ell}\\
    & \le \frac{\omega C}{3} + o_\ell(1) = f_0(C) + o_\ell(1).
\end{align*}

When $i > 0$, assume the claim is correct for $i-1$. Similar to the proof of Proposition~\ref{prop:eh_upper_bound}, we choose $a, b, c$ in our $\cliquedet{k, \ell}$ algorithm so that $\lfloor k/3\rfloor = c \le b \le a = \lceil k/3\rceil$. By the same analysis, the running time exponent can then be bounded by 
$\frac{\omega \cdot g(k-a, \ell - a)}{\omega + g(k-a, \ell - a) - 1}$. Let $C' = \frac{2C}{3-C}$. It is not difficult to verify that $\frac{1}{C'} \in \left(1-\left(\frac{2}{3}\right)^{i-1}, 1-\left(\frac{2}{3}\right)^{i} \right]$. 

Also, 
\begin{align*}
    \frac{k-a}{\ell - a} &\ge \frac{k - k/3}{\ell - k/3} \ge \frac{C \ell - (C\ell) / 3}{\ell - (C \ell) / 3} = \frac{2C}{3 - C} = C',
\end{align*}
and 
\begin{align*}
    \frac{k-a}{\ell - a} &\le \frac{k - (k/3 + 1)}{\ell - (k/3 + 1)} \le \frac{(C + o_\ell(1))\ell - ((C + o_\ell(1))\ell) / 3}{\ell - ((C + o_\ell(1)) \ell) / 3} = \frac{2C + o_\ell(1)}{3 - C - o_\ell(1)} = C' + o_\ell(1).
\end{align*}
Thus, $C'(\ell - a) \le k-a \le (C'+o_\ell(1))(\ell-a)$, so $g(k-a, \ell - a) \le f_{i-1}(C') + o_\ell(1)$ by induction. Therefore, the running time exponent of $\cliquedet{k, \ell}$ can be bounded by 
\begin{align*}
    \frac{\omega \cdot g(k-a, \ell - a)}{\omega + g(k-a, \ell - a) - 1} &= \frac{\omega}{1 + \frac{\omega - 1}{g(k-a, \ell - a)}}
    \le \frac{\omega}{1 + \frac{\omega - 1}{f_{i-1}(C') + o_\ell(1)}}
     \le \frac{\omega}{1 + \frac{\omega - 1}{f_{i-1}(\frac{2C}{3-C})}} +  o_\ell(1) = f_i(C) + o_\ell(1).
\end{align*}
\end{proof}

\begin{figure}[ht]
    \centering
    \includegraphics[width=0.6\textwidth]{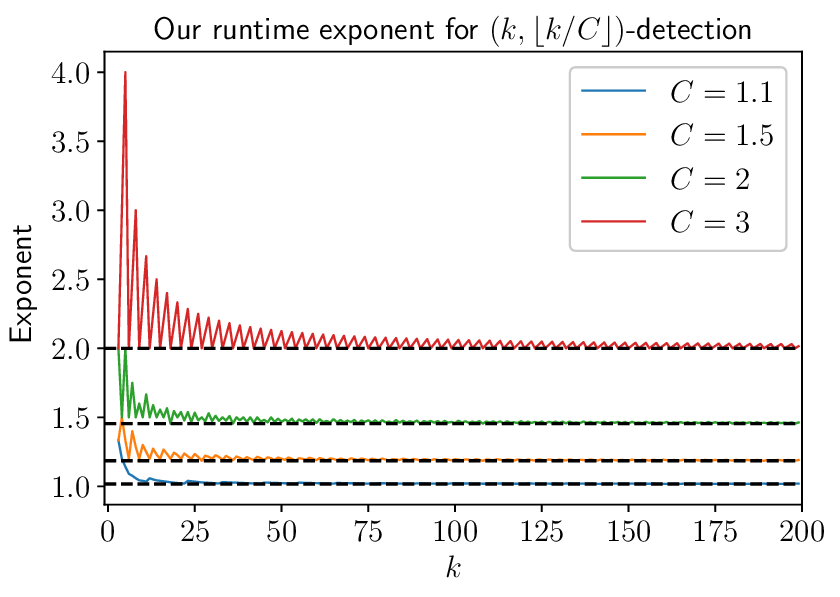}
    \caption{
    Comparison of our detection time for $(k, \lfloor k/C \rfloor)$ (assuming $\omega= 2$) between the actual running time exponent of Algorithm~\ref{alg:generic_detection} and the upper bound from Theorem~\ref{thm:det_mult_upper_bound} (without the $o_\ell(1)$ term; here $\ell=\lfloor k/C \rfloor$). The actual exponents are computed by dynamic programming. The colored lines denote the actual running time exponent, and the dashed lines denote the values of $f_i(C)$ for appropriately chosen $i$. 
    Note that the upper bound from Theorem~\ref{thm:det_mult_upper_bound} can be lower than the actual exponent because we omitted the $o_\ell(1)$ factors.}
    \label{fig:detection_upper_bound}
\end{figure}

In Figure~\ref{fig:detection_upper_bound}, we compare the bound obtained from Theorem~\ref{thm:det_mult_upper_bound} with the actual running time of Algorithm~\ref{alg:generic_detection} computed by dynamic programming for $3 \leq k \leq 200$.
In particular, for various values of $C$, we plot the exponent of $\cliquedet{k, \lfloor k/C\rfloor}$ against the upper bound obtained from Theorem~\ref{thm:det_mult_upper_bound} (without the $o_\ell(1)$ factor). Figure~\ref{fig:detection_upper_bound} shows that the estimates given by Theorem~\ref{thm:det_mult_upper_bound} are actually quite close to the actual exponents, and the values  indeed converge to our bound.

%% file: 4-lower-bound.tex
In this section, we will show our conditional lower bound for $\cliquelist{k,\ell}$ under the Exact-$k$-Clique hypothesis.

\begin{theorem}[Theorem~\ref{thm:lb_intro}]\label{thm:lower_bound}
For any $k \ge 3, 1 \le \ell < k$, and $\gamma \in [0, k/\ell]$, $\cliquelist{k, \ell}$ for instances with $t = \tilde{\Theta}(\Delta_\ell^\gamma)$
requires $$\left(\Delta_\ell^{\frac{2}{\ell(k-\ell)}} t^{1 - \frac{2}{k(k-\ell)}}\right)^{1-o(1)}$$
    time, where $\Delta_\ell$ is the number of $\ell$-cliques and $t$ is the number of $k$-cliques required to list, assuming Hypothesis~\ref{hyp:exact_k_clique}.
\end{theorem}
    
\begin{proof}
First, we can assume 
$\frac{2}{\ell(k-\ell)} + \gamma(1 - \frac{2}{k(k-\ell)}) > 1$, as otherwise the lower bound is trivial. 

    Let $G=(V = V_1 \sqcup \cdots \sqcup V_k, E, w)$ be a $k$-partite Exact-$k$-Clique instance on $k \cdot n$ nodes. Without loss of generality, we assume the edge weights $w$ of $G$ are from $\FF_p$ for some sufficiently large prime $p = n^{O(k)}=n^{O(1)}$. Then we sample $x \sim \FF_p$ uniformly at random. For every $i \in [k]$, and every node $v \in V_i$, we sample $k-1$ random variables $(y_{v, j})_{j \in [k] \setminus\{i\}} \sim \FF_p$ where $\sum_{j \in [k] \setminus\{i\}} y_{v, j} = 0$ uniformly at random. Note that $(y_{v, j})_{j \in [k] \setminus\{i\}}$ are $(k-2)$-wise independent. 
    
    For every $1 \le i < j \le k$, and $(v_i, v_j) \in V_i \times V_j$, let $$w'(v_i, v_j) = x \cdot w(v_i, v_j) + y_{v_i, j} + y_{v_j, i}.$$ It is not difficult to verify that, whenever $x \ne 0$, the sets of exact-$k$-cliques in the graph with weight $w$ and with weight $w'$ are the same. 
    
    Then we partition $\FF_p$ into $s$ contiguous intervals, each of length $O(p/s)$, for some $s$ to be chosen later where $\Omega(1) \le s \le O(n^{\frac{2}{k-1}})$. Consider all combinations of intervals $(L_{i, j})_{1 \le i < j \le k}$, where $0 \in \sum_{1 \le i < j \le k} L_{i, j}$. If we fix an arbitrary choice of the first $\binom{k}{2} - 1$ intervals, their sumset is an interval of length $O(p/s)$. Thus, there is only $O(1)$ choices for the last interval in order for their sumset to contain $0$. Hence, there are only $O(s^{\binom{k}{2} - 1})$ such combinations. For each such combination, we construct an instance of $\cliquelist{k, \ell}$  as follows: create an unweighted graph $H$ such that an edge $(v_i, v_j) \in V_i \times V_j$ for $i < j$ in $G$ is added to $H$ if and only if $w'(v_i, v_j) \in L_{i, j}$. The high level idea then is to list a certain number of $k$-cliques in $H$ and verify whether any of them is an exact-$k$-clique in $G$. Clearly, this algorithm never finds an exact-$k$-clique if $G$ does not have one, so it suffices to show that when $G$ does have an exact-$k$-clique, the algorithm finds it with decent probability. 
    
    Let $(u_1, \ldots, u_k)$ be an arbitrary exact-$k$-clique in $G$. Clearly, there exists one combination of intervals such that all edges in this exact-$k$-clique are in the corresponding subgraph $H_0$. For any $i, j$ and edge $(v_i, v_j) \in V_i \times V_j$, the edge is in $H_0$ only if $w'(u_i, u_j) - w'(v_i, v_j) \in [-O(p/s), O(p/s)]$, which happens with probability $O(1/s)$ as long as $\{u_i, u_j\} \ne \{v_i, v_j\}$. The following lemma shows that the random variables $w'(u_i, u_j) - w'(v_i, v_j)$ are fairly independent. 
    
    \begin{lemma}
    \label{lem:independent}
    If $(v_1, \ldots, v_k)$ is not an exact-$k$-clique w.r.t. $w$, and $(u_1, \ldots, u_k)$ shares exactly $c$ nodes indexed by $S$ with $(v_1, \ldots, v_k)$, then the random variables $$\left\{w'(u_i,u_j) - w'(v_i, v_j)\right\}_{\substack{1 \le i < j \le k \\ i \not \in S \text{ or } j \not \in S}}$$ are independent.     
    \end{lemma}
    \begin{proof}
    By symmetry, we can assume $|S| = [c]$, and we need to show that 
        $$\left\{w'(u_i,u_j) - w'(v_i, v_j)\right\}_{\substack{c+1 \le j \le k \\ 1 \le i < j}} = \left\{ x\cdot (w(u_i,u_j) - w(v_i, v_j)) + y_{u_i, j} + y_{u_j, i} - y_{v_i, j} - y_{v_j, i}\right\}_{\substack{c+1 \le j \le k \\ 1 \le i < j}}$$
        are independent. 
        
        Define $\alpha_{i, j} = w'(u_i,u_j) - w'(v_i, v_j)$. Let $\beta$ be the sum of all the $\alpha_{i, j}$:
        \begin{align*}
            \beta &=\sum_{\substack{c+1 \le j \le k \\ 1 \le i < j}} \alpha_{i, j} 
            = \sum_{1 \le i < j \le k} \alpha_{i, j}\\
            &= x \cdot \left(\sum_{1 \le i < j \le k} \left( w(u_i, u_j) - w(v_i, v_j)\right)\right) + \sum_{i=1}^k \sum_{j \in [k] \setminus \{i\}} y_{u_i, j} - \sum_{i=1}^k \sum_{j \in [k] \setminus \{i\}} y_{v_i, j}\\
            &= x \cdot \left(\sum_{1 \le i < j \le k} \left( w(u_i, u_j) - w(v_i, v_j)\right)\right).
        \end{align*}
        Since $(u_1, \ldots, u_k)$ is an exact-$k$-clique whereas $(v_1, \ldots, v_k)$ is not, we have $\sum_{1 \le i < j \le k} \left( w(u_i, u_j) - w(v_i, v_j)\right) \ne 0$. Therefore, $\beta$ is uniformly random. 

        Showing $\left\{\alpha_{i, j}\right\}_{\substack{c+1 \le j \le k \\ 1 \le i < j}}$ are independent is equivalent to showing that the variables are independent when one of the variables is replaced with the sums of the variables. Namely, it suffices to show $\left\{\alpha_{i, j}\right\}_{\substack{c+1 \le j \le k \\ 1 \le i < j \\ (i, j) \ne (k-1, k)}} \cup \{\beta\}$ are independent. 
        
        Consider the following ordering of the variables:
        $$\beta, \alpha_{c+1, 1}, \ldots, \alpha_{c+1, c}, \alpha_{c+2, 1}, \ldots, \alpha_{c+2, c+1}, \ldots,
        \alpha_{k-1, 1}, \ldots, \alpha_{k-1, k-2}, 
        \alpha_{k, 1}, \ldots, \alpha_{k, k- 2}.$$
        Conditioned on the previous variables, all $\alpha_{j, i}$ variables in this list has an additive term $y_{u_j, i}$ that is independent of all previous variables. Thus, this list of variables is independent. 
    \end{proof}
    \begin{corollary}
    \label{cor:independent}
    For any $\ell$-clique on nodes $(v_i)_{i \in T}$ in $G$ that shares exactly $c$ nodes indexed by $S$ with $(u_1, \ldots, u_k)$, the random variables 
    $$\left\{w'(u_i,u_j) - w'(v_i, v_j)\right\}_{\substack{i, j \in T\\
    i < j\\ i \not \in S \text{ or } j \not \in S}}$$
    are independent. 
    \end{corollary}
    \begin{proof}
        By symmetry, we can assume $[T] =[\ell]$ and $[S] = [c]$. We can complete this $\ell$-clique to a nonzero $k$-clique $(v_1, \ldots, v_k)$ (we can  assume any $\ell$-clique is in some nonzero $k$-clique by adding hypothetical nodes to the graph in this analysis). 
        
        By Lemma~\ref{lem:independent}, $\left\{w'(u_i,u_j) - w'(v_i, v_j)\right\}_{\substack{c+1 \le j \le k \\ 1 \le i < j}}$ are independent, so 
        $\left\{w'(u_i,u_j) - w'(v_i, v_j)\right\}_{\substack{c+1 \le j \le \ell \\ 1 \le i < j}} $ are also independent. 
    \end{proof}
    
    Now we can compute the expected number of $\ell$-cliques in $H_0$. 
    The number of $\ell$-cliques in $G$ that share exactly $c$ nodes with $(u_1, \ldots, u_k)$ is $O(n^{\ell-c})$. By Corollary~\ref{cor:independent}, each of them is in $H_0$ with probability $O\left(1/s^{\binom{\ell}{2} - \binom{c}{2}}\right)$. Therefore, the expected number of $\ell$-cliques in $H_0$ is 
    $$O\left(\sum_{c=0}^\ell n^{\ell-c} / s^{\binom{\ell}{2} - \binom{c}{2}}\right) = O\left(n^\ell / s^{\binom{\ell}{2}}\right),$$
    since by our choice of $s = O(n^{2/(k-1)})$, we have that $s^{\binom{c}{2}} = O(n^c)$.

    Similarly, the expected number of $k$-cliques in $H_0$ that do not correspond to exact-$k$-cliques in $G$ is $O\left(n^k / s^{\binom{k}{2}}\right)$. 
    
    Therefore, by Markov's inequality and union bound, with probability at least $1-1/\Omega(\log n)$, the number of $\ell$-cliques in $H_0$ is at most $n^\ell \log n/ s^{\binom{\ell}{2}}$ and the number of $k$-cliques in $H_0$ that do not correspond to exact-$k$-cliques in $G$ is at most $n^k \log n / s^{\binom{k}{2}}$. 
    
    Let $s = n^{\frac{k-\gamma \ell}{\binom{k}{2} - \gamma \binom{\ell}{2}}}$, so that $\frac{n^k}{s^{\binom{k}{2}}} = \left(\frac{n^\ell}{s^{\binom{\ell}{2}}}\right)^\gamma$. We can verify that indeed $\Omega(1) \le s \le O(n^{\frac{2}{k-1}})$. In fact, since $\frac{2}{\ell(k-\ell)} + \gamma(1 - \frac{2}{k(k-\ell)}) > 1$, we can obtain a stronger upper bound $s = O\left(n^{\frac{k-\ell}{\binom{k}{2}-\binom{\ell}{2}-1}-\delta}\right)$ for $\delta>0$.

    Suppose for the sake of contradiction that there is a $\cliquelist{k, \ell}$ algorithm $\mathcal{A}$ for instances with specified $t = \tilde{\Theta}(\Delta_\ell^\gamma)$  with running time $$T(\Delta_\ell, t) = O\left(\left(\Delta_\ell^{\frac{2}{\ell(k-\ell)}} t^{1 - \frac{2}{k(k-\ell)}}\right)^{1-\eps}\right)$$
    for some $\eps > 0$. Then consider the following algorithm for Exact-$k$-Clique:
    \begin{enumerate}
        \item 
        First, hash the weights of the graph and enumerate $O(s^{\binom{k}{2} - 1})$ graphs $H$ as described earlier.
        \item Enumerate all $\ell$-cliques in $G$, and pre-compute which graphs $H$ contain each $\ell$-clique. Since each $\ell$-clique exists in $s^{\binom{k}{2}-1 -\binom{\ell}{2}}$ graphs $H$, and this list of graphs can be listed efficiently, this step costs 
        $$\tO\left(n^\ell \cdot s^{\binom{k}{2}-1 -\binom{\ell}{2}}\right) \le \tO\left(n^\ell \cdot \left(n^{\frac{k-\ell}{\binom{k}{2}-\binom{\ell}{2}-1}-\delta}\right)^{\binom{k}{2}-1 -\binom{\ell}{2}}\right)\le \tO(n^{k-\delta'})$$
        for some $\delta' > 0$.
        
        \item \label{item:lower-bound-skip} From the previous step, we have a list of $\ell$-cliques for each graph $H$. If some $H$ contains more than $n^\ell \log n/ s^{\binom{\ell}{2}}$ $\ell$-cliques, we skip it. If it contains fewer than $0.99 n^\ell / s^{\binom{\ell}{2}}$ $\ell$-cliques, we add a complete $\ell$-partite graphs with $n'$ nodes on each part, for some $n'$, so that the total number of $\ell$-cliques in the new graph reaches $0.99(n+n')^\ell / s^{\binom{\ell}{2}}$. Clearly, $n' = O(n)$. 
        \item \label{item:lower-bound-listing} For graphs $H$ which we did not skip in the previous step, we run $\mathcal{A}$ on it with $t = n^k \log n / s^{\binom{k}{2}} + 1$. For any $k$-clique listed by $\mathcal{A}$, we test whether it is an exact-$k$-clique in $G$. This step takes $s^{\binom{k}{2}-1} \cdot T(\tO(n^\ell / s^{\binom{\ell}{2}}), \tO(n^k / s^{\binom{k}{2}}))$ time.
        \item If any exact-$k$-clique is found in the previous step, we return YES for the Exact-$k$-Clique instance; otherwise, we return NO. 
    \end{enumerate}
    
    Clearly, if $G$ contains no exact-$k$-clique, our algorithm is always correct. If $G$ contains any exact-$k$-clique, let $H_0$ be the constructed graph containing it. As discussed previously, with probability $1-1/\Omega(\log n)$, the number of $\ell$-cliques in $H_0$ is at most $n^\ell \log n/ s^{\binom{\ell}{2}}$ and the number of $k$-cliques in $H_0$ that do not correspond to exact-$k$-cliques in $G$ is at most $n^k \log n / s^{\binom{k}{2}}$. In this case, we will not skip $H_0$ in Step~\ref{item:lower-bound-skip}, and listing $t = n^k \log n / s^{\binom{k}{2}} + 1$ $k$-cliques in Step~\ref{item:lower-bound-listing} guarantees an exact-$k$-clique. Thus, we will find an exact-$k$-clique with probability $1-1/\Omega(\log n)$, which can be boosted to $1-1/\poly(n)$ by repeating the algorithm $O(\log n)$ times. 
    
    Overall, this algorithm only needs time (besides the previous $\tO(n^{k-\delta'})$ time)
    \begin{align*}
        &\tilde{O}\left(s^{\binom{k}{2}-1}\cdot \left(\left(n^\ell / s^{\binom{\ell}{2}}\right)^{\frac{2}{\ell(k-\ell)}} \left(n^k / s^{\binom{k}{2}}\right)^{1-\frac{2}{k(k-\ell)}} \right)^{1-\eps}\right) =  \tilde{O}\left(n^k \cdot \left(\frac{s^{\binom{k}{2}-1}}{n^k}\right)^\eps\right).
    \end{align*}
    As $s = O(n^{\frac{2}{k-1}})$, the above running time can be further upper bounded by 
    \begin{align*}
        \tilde{O}\left(n^k \cdot \left(\frac{(n^{\frac{2}{k-1}})^{\binom{k}{2}-1}}{n^k}\right)^\eps\right) = \tO\left(n^{k-\frac{2\eps}{k-1}}\right),
    \end{align*}
    contradicting the Exact-$k$-Clique hypothesis. 
\end{proof}

%% file: 5-upper-bound.tex
In this section, we give a $\cliquelist{k, 1}$ algorithm that is  optimal for graphs with many $k$-cliques under Hypothesis~\ref{hyp:exact_k_clique}. This algorithm can be seen as a generalization of the densifying and sparsifying paradigm of  \cite{bjorklund2014listing}. 

We then show how we can extend this algorithm to obtain the conditionally optimal algorithms for all $\cliquelist{k, \ell}$ for graphs with many $k$-cliques.

\subsection{Algorithm}

First, we describe the algorithm for $\cliquelist{k, 1}$ in Algorithm~\ref{alg:large_t_sparse_dense}. 

 \begin{breakablealgorithm}
        \caption{$\cliquelist{k, 1}$ Algorithm for large $t \geq n^{\gamma_k}$, where $\gamma_k$ is defined in Theorem~\ref{thm:k_1_optimal}}\label{alg:large_t_sparse_dense}

        \begin{algorithmic}
            \item \dense{}$(G:= (V, E), n, t)$:
            \begin{itemize}
            \item \textbf{Input:} Graph $G = (V, E)$ with $|V| \leq n$ and at most $t$ $k$-cliques.
            \item \textbf{Output:} List of $k$-cliques in $G$.
            \item \textbf{The Algorithm:}
        \begin{enumerate}
            \item If $n < k$, it returns no $k$-cliques.
            \item Choose a parameter $\lambda$. Let an edge be $\lambda$-light if it is in fewer than $\lambda$ $k$-cliques.
            \item Use the algorithm in Lemma~\ref{lem:simple_list_ub} to obtain a list $L$ of all $(k-2)$-cliques (there are at most $n^{k-2}$ such cliques).
            \item Initialize an empty list $T$.
            \item \label{step:matmul_sampling}Repeat the following $O(\lambda \log n)$ times:
            \begin{itemize}
                \item Sample a subset $L'$ of $L$ of size $|L|/\lambda$.
                \item Construct adjacency matrices $A$ and $\overline{A}$ where the rows are indexed by $V$ and columns are indexed by $L'$.
                \item Let $A[v, C] = 1$ if node $v$ is distinct from and adjacent to every node in the $(k-1)$-clique $C$, and set $A[v, C] = 0$ otherwise.
                \item Let $\overline{A}[v, C] = A[v, C] \cdot C$, i.e. column $C$ contains entries 0 or $C$. 
                \item Compute $B = A \cdot A^T$ and $\overline{B} = A \cdot \overline{A}^T$. This takes $O(\MM(n, |L'|, n))$ time.
                \item For every edge $(u, v) \in E$ that is $\lambda$-light, if $B[u, v] = 1$, add $(u, v, \overline{B}[u, v])$ to $T$.
            \end{itemize}
            \item Output $T$.
            \item Delete all $\lambda$-light edges from $E$ to obtain $E'$ (all $\lambda$-light edges are found in Step \ref{step:matmul_sampling} w.h.p.).
            \item Call {\tt Sparse}$(G' := (V, E'), \binom{k}{2}t/\lambda, t)$.
        \end{enumerate}
        \end{itemize}
    \item\sparse{}$(G:= (V, E), m, t)$:
       \begin{itemize}
       \item \textbf{Input:} Graph $G = (V, E)$ with $|E| \leq m$ and at most $t$ $k$-cliques.
       \item \textbf{Output:} List of $k$-cliques in $G$.
       \item \textbf{The Algorithm:}
        \begin{enumerate}
            \item If $m < \binom{k}{2}$, it returns no $k$-cliques.    
            \item Choose a parameter $x$.
            \item Find all nodes such that $\deg(v) \leq x$, and call the $\cliquelist{k-1,1}$ algorithm in the neighbourhoods of all such nodes with $n' = \deg(v)$.
            \item Delete all nodes in $V$ of degree less than $x$ to obtain set $V'$.
            \item Call {\tt Dense}$(G' := (V', E \cap (V' \times V')), 2m/x, t)$
        \end{enumerate}
    \end{itemize}
    \end{algorithmic}
    \end{breakablealgorithm}

In the $\dense$ algorithm, we use matrix multiplication to enumerate all $k$-cliques containing light edges, i.e. edges that are part of very few $k$-cliques. These edges are then removed to result in a  \emph{sparse} graph with only edges that are part of many $k$-cliques.

In the $\sparse$ algorithm, we enumerate all $k$-cliques containing low-degree nodes by recursively listing all $(k-1)$-cliques in their neighborhoods, and delete all such nodes. Deleting these nodes results in a \emph{dense} graph with only high degree nodes. While one could brute-force the $(k-1)$-cliques in the neighborhoods, our key insight is that we can instead recursively use a $\cliquelist{k-1, 1}$ algorithm to be more efficient. 

We first show the correctness of Algorithm~\ref{alg:large_t_sparse_dense} and defer its runtime analysis to Section~\ref{sec:k_1_opt}.

\paragraph{Correctness.} It is clear that the $\sparse$ algorithm finds all $k$-cliques in the neighborhoods of low-degree nodes. At the end of the algorithm, since the graph has $\Delta_\ell$ $\ell$-cliques and only nodes with degree at least $x$, there are at most $2m/x$ nodes left in the graph.

Now, we argue that the $\dense$ algorithm lists all $k$-cliques containing $\lambda$-light edges.

We argue that Step~\ref{step:matmul_sampling} finds all $\lambda$-light edges with high probability. For every $(u, v) \in E$, let $L_{u, v}$ denote the set of all $(k-2)$-cliques that form $k$-cliques with nodes $u$ and $v$. Since we sample $L'$ of size $|L|/\lambda$, the probability that $L_{u, v} \cap L' = K_{k-2}$ for any fixed $K_{k-2} \in L_{u, v}$ is \[\frac{1}{\lambda} \cdot \left(1 - \frac{1}{\lambda}\right)^{|L_{u, v}| - 1} \geq \frac{1}{\lambda} \cdot \left (1 - \frac{1}{\lambda}\right)^{\lambda-1} \geq \frac{1}{e \lambda}.\]
Therefore, by choosing $O(\lambda \log  n)$ random sets of size $|L|/\lambda$, with high probability, we find all $k$-cliques containing $\lambda$-light edges.

\paragraph{$\cliquelist{k, \ell}$ when $\ell \geq 2$.} To generalize this algorithm to $\cliquelist{k, \ell}$ for $\ell \geq 2$, we recursively use $\cliquelist{k-1, \ell-1}$ to reduce the problem to $\cliquelist{k, 1}.$ At a high level, the algorithm considers all nodes $v$ in fewer than $x$ $\ell$-cliques and recursively calls $\cliquelist{k-1,\ell-1}$ to list all $k$-cliques containing $v$. See Algorithm~\ref{alg:optimal_kl}. The correctness of Algorithm~\ref{alg:optimal_kl} can be shown as follows. 

\begin{algorithm}[ht]
    \caption{$\cliquelist{k, \ell}$ Algorithm for large $t \geq n^{\gamma_{k,\ell}}$, where $\gamma_{k, \ell}$ is defined in Theorem~\ref{thm:k-l-large-t-listing}}\label{alg:optimal_kl}
    \begin{algorithmic}
    \item \textbf{Input:} A graph $G$ and a list $L$ of all $\ell$-cliques.
    \item \textbf{Output:} All $k$-cliques in the graph.
    \item \textbf{The Algorithm:}
    \begin{enumerate}
        \item Call a node $v$ light if $\Delta_\ell(v) \leq x$, for some parameter $x$. 
        \item\label{step:kl_opt_lightnodes} For all light nodes, call $\cliquelist{k-1, \ell-1}$ in the neighbourhoods to find all $k$-cliques incident to $x$. 
        \item Delete all light nodes and incident edges from $G$. 
        \item\label{step:kl_densecall} Call the $\cliquelist{k, 1}$ algorithm $\dense(G' := (V', E'), \ell \Delta_\ell/x, t)$ (from Algorithm~\ref{alg:large_t_sparse_dense}).
    \end{enumerate}
    \end{algorithmic}
\end{algorithm}

\paragraph{Correctness.} It is clear that the algorithm lists all $k$-cliques incident to low-degree nodes. Since all remaining nodes are in at least $x$ $\ell$-cliques, and each $\ell$-cliques contains at most $\ell$ nodes, we can bound the remaining number of nodes by $\ell \Delta_\ell/x$.

To illustrate these algorithms, we first show simplified analyses of Algorithms~\ref{alg:large_t_sparse_dense} and \ref{alg:optimal_kl} for the case of $k = 4$ and $k=5$ assuming that $\omega = 2$ in Section~\ref{sec:4_5_l_listing}. We give more detailed analyses in terms of $\omega$ in Sections~\ref{sec:k_1_opt} and \ref{sec:k_l_opt}. 

\subsection{Analysis for \texorpdfstring{$k = 4$}{k = 4} and \texorpdfstring{$k = 5$}{k = 5} assuming \texorpdfstring{$\omega = 2$}{omega = 2}}\label{sec:4_5_l_listing}
In this section, we illustrate how to analyze the runtime for listing algorithm by considering the cases where $k = 4$ or $k=5$.
\begin{proposition}
    Suppose $\omega = 2$. Then, given a graph $G$ with $t$ 4-cliques,
    \begin{itemize}
        \item $\cliquelist{4,1}$ can be solved in $\tilde{O}(n^3 + n^{2/3}t^{5/6})$ if $G$ has $n$ nodes.
        \item $\cliquelist{4,2}$ can be solved in $\tilde{O}(m^{3/2} + mt^{2/5} + m^{1/2}t^{3/4})$ if $G$ has $m$ edges.
        \item $\cliquelist{4, 3}$ can be solved in $\tilde{O}(\Delta^{6/5} + \Delta t^{1/5} + \Delta^{2/3}t^{1/2})$ if $G$ has $\Delta = \Delta_3$ triangles.
     \end{itemize}  
\end{proposition}

\begin{proof}
    Consider the $\dense$ algorithm. In this case, $L$ is a list of all (up to $n^2$) edges. Therefore, the runtime of this step can be bounded by 
    \[D'(n, m, t) \leq n^2 + \lambda \log n \cdot \MM(n, m/\lambda, n) + S(6t/\lambda, t) = \tilde{O}(n^2 + \lambda n^2 + nm) + S(6t/\lambda).\]
    We can also upper bound $m$ by $n^2$ to obtain the following bound without a dependence on $n$:
\[
    D(n, t) \leq D(n, n^2, t) \leq \tilde{O}(n^3 + \lambda n^2) + S(6t/\lambda, t).
\]
assuming $\omega = 2$. 
    
Consider the $\sparse$ algorithm. In this case, we call $(3,1)$-listing, which takes time $\tO(n^2 + nt^{2/3})$. Therefore (ignoring $\tO(1)$ factors), 
\begin{align*}
    S(m, t) &\leq \sum_{v: \deg(v) \leq x} \left(\deg(v)^2 + \deg(v) \Delta_4(v)^{2/3}\right)+ D'(2m/x, m, t)\\
    &\leq  \sum_{v: \deg(v) \leq x} \left(\deg(v) \cdot x + \deg(v)^{1/3} \Delta_4(v)^{2/3} x^{2/3}\right)+ D'(2m/x, m, t)\\
    &\leq mx + m^{1/3}t^{2/3}x^{2/3} + D'(2m/x, m, t),
\end{align*}
where we applied H\"{o}lder's inequality as seen in Corollary~\ref{cor:holders_useful}. 

\paragraph{$\cliquelist{4, 1}$ analysis.} To obtain a runtime for $\cliquelist{4,1}$, we unravel the recursion in $D(n, t)$. Ignoring $\tilde{O}(1)$ factors in the following inequalities, we have
\begin{align*}
    D(n, t) &\leq n^3 + \lambda n^2 + S(6t/\lambda, t)\\
    &\leq n^3 + \lambda n^2 + \frac{tx}{\lambda} + \frac{tx^{2/3}}{\lambda^{1/3}} + D\left(\frac{12t}{\lambda x}, t\right)
\end{align*}
Choosing $\lambda = \max\{1, \frac{24t}{nx}\}$, we have $\frac{12t}{\lambda x} \leq n/2$, and the above runtime will be dominated by the first four terms up to $\tilde{O}(1)$ factors. Substituting this value of $\lambda$, we obtain a runtime of 
    \[D(n, t) \leq n^3 + \frac{tn}{x} + nx^2 + n^{1/3}t^{2/3}x.
    \]
Choosing \[x = 
    \begin{cases}
        n & t \leq n^{5/2}\\
        n^{8/3}/t^{2/3} & n^{5/2} \leq t \leq n^{14/5}\\
        n^{1/3}t^{1/6} & t \geq n^{14/5}
    \end{cases},\]
we obtain $D(n, t) = n^3 + n^{2/3}t^{5/6}$, as desired.

\paragraph{$\cliquelist{4, 2}$ analysis.} To obtain a runtime for $\cliquelist{4, 2}$, we analyze the runtime of $S(m, t)$.
Here, we use our bound $D'$ in terms of $n$, $m$ and $t$ to get a tighter analysis
(instead of just $n$ and $t$). 

\begin{align*}
    S(m, t) &\leq  mx + m^{1/3}t^{2/3}x^{2/3} + D'(2m/x,m, t)\\
    &\leq mx + m^{1/3}t^{2/3}x^{2/3} + \frac{\lambda m^2}{x^2} + \frac{m^2}{x} + S(6t/\lambda, t).
\end{align*}
Choosing $\lambda = \max\{1, \frac{12t}{m}\}$, we have that $6t/\lambda \leq m/2$. Therefore, the first four terms dominate up to $\tilde{O}(1)$ factors, so (ignoring $\tO(1)$ factors)
\begin{align*}
    S(m, t) &\leq mx + m^{1/3}t^{2/3}x^{2/3} + \frac{mt}{x^2} + \frac{m^2}{x}.
\end{align*}
By choosing 
\[
    x = \begin{cases}
        m^{1/2} &\text{if $t \leq m^{5/4}$}\\
        m/t^{2/5} &\text{if $m^{5/4} < t \leq m^{10/7}$}\\
        m^{1/4}t^{1/8} & \text{if $t > m^{10/7}$,}
    \end{cases}
\] 
we get a runtime of 
\[S(m, t) \leq m^{3/2} + mt^{2/5} + m^{1/2}t^{3/4}.\]

\paragraph{$\cliquelist{4,3}$ analysis.} Note that while we can use Algorithm~\ref{alg:optimal_kl} to bound the runtime in this case, we instead provide a more efficient algorithm shown in Algorithm~\ref{alg:4_3} for $\cliquelist{4, 3}$.

\begin{algorithm}[ht]
    \caption{$\cliquelist{4, 3}$ algorithm}\label{alg:4_3}
    \begin{algorithmic}
    \item \textbf{Input:} Graph $G = (V, E)$, and a list $L$ of all triangles in $G$.
    \item \textbf{Output:} A list of all $k$-cliques in $G$.
    \begin{enumerate}
        \item Call an edge light if it occurs in fewer than $x$ triangles, i.e. $\Delta(e) \leq x$. 
        \item\label{step:light_edges_in_triangles} For all light edges $e$, consider all pairs of nodes in its neighbourhoods to find all 4-cliques containing $e$.
        \item Delete all light edges from $G$.
        \item Call $\cliquelist{4, 2}$ algorithm $\sparse(G' := (V', E'), 3\Delta/x, t).$ 
    \end{enumerate}
    \end{algorithmic}
\end{algorithm}

The runtime of Step~\ref{step:light_edges_in_triangles} is bounded by $\sum_{e:\Delta(e) \le x} \Delta(e)^2 \leq \Delta x.$ Now, we call $\cliquelist{4, 2}$ with a graph with at most $3\Delta/x$ edges and $t$ $4$-cliques, giving a runtime of 
\[
    \left(\Delta/x\right)^{3/2} + \left(\Delta/x\right) t^{2/5} + \left(\Delta/x\right)^{1/2} t^{3/4}.
\]
Therefore, choosing $x = \max\{\Delta^{1/5}, t^{1/5}, t^{1/2}/\Delta^{1/3}\}$, we get a runtime of $\tilde{O}\left(\Delta^{6/5} + \Delta t^{1/5} + \Delta^{2/3} t^{1/2}\right).$ 

\end{proof}

\begin{proposition}
    Suppose $\omega = 2$. Then, given a graph $G$ with $t$ 5-cliques,
    \begin{itemize}
        \item $\cliquelist{5,1}$ can be solved in $\tilde{O}(n^4 + n^{1/2}t^{9/10})$ if $G$ has $n$ nodes.
        \item $\cliquelist{5,2}$ can be solved in $\tilde{O}(m^2 + m^{17/18}t^{10/18} + m^{1/3}t^{13/15})$ if $G$ has $m$ edges.
     \end{itemize}  
\end{proposition}

\begin{proof}
Consider the $\dense$ algorithm. In this case, $L$ is a list of all (up to $n^3$) triangles. Therefore, the runtime of this step is 
\[D'(n, \Delta_3, t) \leq n^3 + \lambda \log n \cdot \MM(n, \Delta_3/\lambda, n) + S(10t/\lambda, t) = \tilde{O}(n^3+n\Delta_3 + \lambda n^2) + S(10t/\lambda, t)\]
assuming $\omega = 2$.
Upper bounding $\Delta_3 \leq O(n^3)$, we get a bound without dependence on $\Delta_3$ of \[D(n, t) \leq D'(n, n^3, t) = \tilde{O}(n^4 + \lambda n^2) +S(10t/\lambda).\]
Consider the $\sparse$ algorithm. In this case, we call $\cliquelist{4, 1}$ in the neighborhoods of all low-degree nodes, so
\begin{align*}
    S(m, t) \leq \sum_{v: \deg(v) \leq x} \left(\deg(v)^3 + \deg(v)^{2/3} \Delta_5(v)^{5/6}\right) + D(2m/x, t) &\leq mx^2 + m^{1/6}t^{5/6}x^{1/2} +  D(2m/x, t)
\end{align*}
by using H\"{o}lder's inequality as in Corollary~\ref{cor:holders_useful}.

\paragraph{$\cliquelist{5, 1}$ analysis.} To obtain a runtime for $\cliquelist{5, 1}$, we analyze the runtime of $D(n, t) $ by unravelling the recursion. Therefore, we have the following inequalities (omitting $\tilde{O}(1)$ factors):
\begin{align*}
    D(n, t) &\leq n^4 + \lambda n^2 + \frac{tx^2}{\lambda} + \left(\frac{t}{\lambda}\right)^{1/6} t^{5/6}x^{1/2} + D\left(\frac{20t}{\lambda x}, t\right)
\end{align*}
By choosing $\lambda = \max\left\{5, \frac{40 t}{nx}\right\}$, we would have $\frac{20t}{\lambda x} \leq \frac{n}{2}$, and the running time will therefore be dominated by the first 4 terms. Choosing 
\begin{align*}
    x = 
    \begin{cases}
        n & \text{if } t \leq n^{19/15}\\
        n^{23/4}/t^{5/4} & \text{if } n^{19/5} \leq t \leq n^{35/9}\\
        n^{1/2}t^{1/10} & \text{if } t \geq n^{35/9}
    \end{cases},
\end{align*}
we obtain a runtime of $\tilde{O}(n^4 + n^{1/2}t^{9/10}).$

\paragraph{$\cliquelist{5, 2}$ analysis.} We now analyze the runtime of $S(m, t).$ Note that the graph has at most $\Delta_3 = O(m^{3/2})$ triangles. Here, we use $D'(n, \Delta_3, t)$ to bound the runtime instead.
Therefore, unrolling the recursion, we have (up to $\tilde{O}(1)$ factors) 
\begin{align*}
    S(m, t) &\leq mx + m^{1/6}t^{5/6}x^{1/2} + D'(m/x, m^{3/2}, t)\\
    &\leq mx + m^{1/6}t^{5/6}x^{1/2} + (m/x) \cdot m^{3/2} + (m/x)^{1/2}t^{9/10}.
\end{align*}
Setting 
\[
    x = \begin{cases}
        m^{1/2} & \text{if } t \leq m^{19/10}\\
        m^{17/18}t^{10/18} & \text{if }m^{19/10} \leq t \leq m^{55/28}\\
        m^{1/3}t^{13/15} & \text{if } t \geq m^{55/28}
    \end{cases},
\]
we get a runtime of $\tilde{O}(m^2  + m^{17/18}t^{10/18} + m^{1/3}t^{13/15}).$ 
\end{proof}

\subsection{Analysis for \texorpdfstring{$\cliquelist{k, 1}$}{(k,1)-Clique-Listing}}\label{sec:k_1_opt} 

For $k \geq 2$, define 
\begin{align}
    x_k &= k \prod_{j = 2}^k ((5 - 2j) + (j - 2)\omega)\label{eq:alpha_num}\\
    y_k &= (3-\omega)^{k-2} + \sum_{j = 2}^{k-1} (3-\omega)^{k-1-j} x_j
    \label{eq:alpha_denom}
\end{align}

The following identities are immediate. 
\begin{claim}
    For any $k \ge 3$, $x_k = x_{k-1} \cdot \frac{k}{k-1} \cdot ((5-2k)+(k-2)\omega)$ and $y_k = (3  - \omega) \cdot y_{k-1} +  x_{k-1}$.
\end{claim}

\begin{theorem}\label{thm:k_1_optimal}
    Let $\alpha_k = x_k/y_k$. For any $k \geq 2$ and large $t \geq n^{\gamma_k}$ where 
    \[\gamma_k = 
    \begin{cases}
    0 & \text{if } k = 2\\
    k \left(1 - \frac{3 - \omega}{k - \alpha_k}\right) & \text{if } k \geq 3
    \end{cases},\]  
    there exists an algorithm that lists all $t$ $k$-cliques in time $\tO(n^{\alpha_k}t^{1-\frac{\alpha_k}{k}})$. If $\omega=2$,  we have that $x_k = k$ and $y_k = \frac{k(k-1)}{2}$, therefore giving a runtime of $\tO(n^{\frac{2}{k-1}} t^{1 - \frac{2}{k(k-1)}})$ for $t \geq n^{k - 1 - \frac{2}{k^2 - k - 2}}.$
\end{theorem}

\begin{proof}
    For $k = 2$, the brute-force algorithm runs in $n^2$ time, and it is easy to check that $x_k=2$ and $y_k = 1$. Moreover, this bound holds for all values of $t$, so we can set $\gamma_k = 0.$
    
    For $k = 3$, \cite{bjorklund2014listing} give an algorithm that runs in time $O(n^\omega + n^{\frac{3(\omega - 1)}{5 - \omega}}t^{\frac{2(3-\omega)}{5 - \omega}})$, which can easily be verified to match the form of the theorem statement. Rewriting this as $O(n^\omega + n^{\alpha_3}t^{1- \frac{\alpha_3}{3}})$, it is easy to see that this term dominates exactly when $t \geq n^{3\left(1 - \frac{3-\omega}{3-\alpha_3}\right)}$, which corresponds exactly to our setting of $\gamma_3.$
    Now suppose  $k \geq 3$ and that the theorem statement is true for all $\cliquelist{r, 1}$ for all $r < k$.
    In particular, suppose the runtime of $\cliquelist{k-1, 1}$ is bounded by $$T_{k-1}(n, \Delta_{k-1}) \leq n^{\alpha_{k-1}}\Delta_{k-1}^{1 - \frac{\alpha_{k-1}}{k-1}},$$ when $\Delta_{k-1} \ge n^{\gamma_{k-1}}$ for some $\gamma_{k-1}$. In fact, since the runtime is non-decreasing in the parameter $\Delta_{k-1}$ by Lemma~\ref{lem:list_exponent_monotone}, one can bound the above runtime for any $\Delta_k$  by:
    \begin{align*}
        T_{k-1}(n, \Delta_{k-1}) \leq n^{\alpha_{k-1}}\left(n^{\gamma_{k-1}}\right)^{1 - \frac{\alpha_{k-1}}{k-1}} + n^{\alpha_{k-1}}\Delta_{k-1}^{1 - \frac{\alpha_{k-1}}{k-1}}.
    \end{align*}

    \paragraph{Runtime analysis.} Let $D(n, t)$ be the running time of $\dense(G, n, t)$, and let $S(m, t)$ be the running time of $\sparse(G, m, t)$. Note that ignoring $\tilde{O}(1)$ factors
    \begin{align*}
        D(n, t) &\le n^{k-1} + \lambda \MM(n, n^{k-2}/\lambda, n) + S\left(\binom{k}{2}t/\lambda, t\right)
    \end{align*}
    
    By the standard trick of decomposing a rectangular matrix product into smaller square matrix products, one can bound $$\MM(n, n^{k-2}/\lambda, n) \leq \left(\frac{n^{k-2}/\lambda}{n}\right) \cdot n^\omega + \left(\frac{n}{n^{k-2}/\lambda}\right)^2 \cdot \left(\frac{n^{k-2}}{\lambda}\right)^\omega = \frac{n^{k-3+\omega}}{\lambda} + \frac{n^{\omega(k-2) -2k + 6}}{\lambda^{\omega - 2}}.$$
    Therefore, we can rewrite 
    $$D(n, t) \leq n^{k-3+\omega} + \lambda^{3 - \omega} n^{\omega(k-2) -2k + 6} + S\left(\binom{k}{2}t/\lambda, t\right).$$

    For $\sparse(G, m, t)$, note that the runtime is bounded by:
    \begin{align*}
        S(m, t) &\leq \sum_{v : \deg(v) \leq x} T_{k-1}(\deg(v), \Delta_k(v)) + D(2m/x, t).\\
        &\leq \sum_{v : \deg(v) \leq x} \left(\deg(v)^{\alpha_{k-1} + \gamma_{k-1}\left(1 - \frac{\alpha_{k-1}}{k-1}\right)} + \deg(v)^{\alpha_{k-1}}\Delta_k(v)^{1 - \frac{\alpha_{k-1}}{k-1}}\right) + D(2m/x, t).
    \end{align*}
    
    One can use H\"{o}lder's inequality as in Corollary~\ref{cor:holders_useful} to bound 
    \begin{align*}
        \sum_{v: \deg(v) \leq x} \deg(v)^{\alpha_{k-1}} \Delta_k(v)^{1 - \frac{\alpha_{k-1}}{k-1}} 
        &\leq x^{\alpha_{k-1} - \frac{\alpha_{k-1}}{k-1}}\sum_{v: \deg(v) \leq x} \deg(v)^{\frac{\alpha_{k-1}}{k-1}}\Delta_k(v)^{1 - \frac{\alpha_{k-1}}{k-1}} \\
        &\leq 
        x^{\alpha_{k-1} \cdot \frac{k-2}{k-1}} \left(\sum_{v: \deg(v) \leq x} \deg(v)\right)^{\frac{\alpha_{k-1}}{k-1}} \left(\sum_{v: \deg(v) \leq x} \Delta_k(v) \right)^{1- \frac{\alpha_{k-1}}{k-1}}\\
        & \leq
        O\left(x^{\alpha_{k-1} \cdot \frac{k-2}{k-1}} m^{\frac{\alpha_{k-1}}{k-1}} 
        t^{1- \frac{\alpha_{k-1}}{k-1}}\right).
    \end{align*}
    
    Thus, we have that (once again ignoring $\tilde{O}(1)$ factors)
    $$S(m, t) \leq m \cdot x^{\alpha_{k-1} + \gamma_{k-1}\left(1 - \frac{\alpha_{k-1}}{k-1}\right) - 1} + x^{\alpha_{k-1} \cdot \frac{k-2}{k-1}} m^{\frac{\alpha_{k-1}}{k-1}} 
        t^{1- \frac{\alpha_{k-1}}{k-1}} + D(2m/x, t).$$
   
    Unravelling the runtime of $\dense(G, n, t)$, we therefore have
    \begin{align*}
        D(n, t) \leq & n^{k-3+\omega} + \lambda^{3-\omega}n^{\omega(k-2) -2k + 6}\\ & + (t/\lambda) \cdot x^{\alpha_{k-1} + \gamma_{k-1}\left(1 - \frac{\alpha_{k-1}}{k-1}\right) - 1} + x^{\alpha_{k-1} \cdot \frac{k-2}{k-1}} (t/\lambda)^{\frac{\alpha_{k-1}}{k-1}} 
        t^{1- \frac{\alpha_{k-1}}{k-1}} \\
        &+ D\left(\frac{2\cdot \binom{k}{2}t}{\lambda x}, t\right).
    \end{align*}
    If one chooses $\lambda$ and $x$ so that $\frac{2 \cdot \binom{k}{2} t}{\lambda x} \leq \frac{n}{2}$, then the runtime is dominated by the first 4 terms up to $\tilde{O}(1)$ factors. Therefore, we choose $\lambda = \max\{1, \frac{4 \cdot \binom{k}{2} t}{n x}\}$ (note that for $t \geq n^{\gamma_k}$, this value will always be equal to $\frac{4 \cdot \binom{k}{2} t}{n x}$ for our setting of $x$). Hence, ignoring $\tilde{O}(1)$ factors, this gives us for $t \geq n^{\gamma_k}$, 
    \begin{align*}
        D(n, t) &\leq n^{k-3+\omega} + \left(\frac{t}{nx}\right)^{3 - \omega}n^{\omega(k-2) -2k + 6} \\
        &+ (n\cdot x) \cdot x^{\alpha_{k-1} + \gamma_{k-1}\left(1 - \frac{\alpha_{k-1}}{k-1}\right) - 1} + x^{\alpha_{k-1} \cdot \frac{k-2}{k-1}} (n\cdot x)^{\frac{\alpha_{k-1}}{k-1}} 
        t^{1- \frac{\alpha_{k-1}}{k-1}}
    \end{align*}
    
    First, suppose that the term $\left(\frac{t}{nx}\right)^{3 - \omega}n^{\omega(k-2) -2k + 6}$
    dominates $n^{k-3+\omega}$ and the term $x^{\alpha_{k-1} \cdot \frac{k-2}{k-1}} (n \cdot x)^{\frac{\alpha_{k-1}}{k-1}} t^{1- \frac{\alpha_{k-1}}{k-1}}$ dominates $(n\cdot x) \cdot x^{\alpha_{k-1} + \gamma_{k-1}\left(1 - \frac{\alpha_{k-1}}{k-1}\right) - 1}$ (we show that this is in fact true for our choice of $\gamma_k$ later) . Then, 
    \begin{equation}\label{eq:bounding_dense}
        D(n, t) \leq \left(\frac{t}{nx}\right)^{3- \omega}n^{\omega(k-2) -2k + 6} + x^{\alpha_{k-1} \cdot \frac{k-2}{k-1}} (n \cdot x)^{\frac{\alpha_{k-1}}{k-1}}t^{1- \frac{\alpha_{k-1}}{k-1}}.
    \end{equation}
    Choosing $x$ to equate the two terms, we set 
    \begin{equation}\label{eq:x_solution_k1_runtime}
        x = \left({t^{\alpha_{k-1} - (\omega-2)(k-1)}n^{(k-1)^2 \omega - (\alpha_{k-1} + 2k^2 - 5k + 3)}}\right)^{\frac{1}{(k-1)(3+\alpha_{k-1} - \omega)}}.
    \end{equation}
    Substituting this into \eqref{eq:bounding_dense}, we have \[D(n, t) \leq n^{\frac{k\alpha_{k-1}((5-2k)+(k-2)\omega)}{(k-1) \cdot (3 + \alpha_{k-1} - \omega)}} t^{\frac{(3-\omega)((k-2)\alpha_{k-1} + (k-1))}{(k-1) \cdot (3 + \alpha_{k-1} - \omega)}}.
    \]
    Setting 
    \begin{equation}\label{eq:alpha_recursive}
    \alpha_k = \frac{k\alpha_{k-1}((5-2k)+(k-2)\omega)}{(k-1) \cdot (3 + \alpha_{k-1} - \omega)},
    \end{equation}
    it is easy to verify that the above bound is in fact of the form $n^{\alpha_k} t^{1-\frac{\alpha_k}{k}}.$ Moreover, note that
    \begin{align*}
        \alpha_k &= \frac{k\alpha_{k-1}((5-2k)+(k-2)\omega)}{(k-1) \cdot (3 + \alpha_{k-1} - \omega)}\\
        &= \frac{k \cdot \frac{x_{k-1}}{y_{k-1}} \cdot ((5-2k)+(k-2)\omega)}{(k-1) \cdot (3  - \omega +  \frac{x_{k-1}}{y_{k-1}})}\\
        &= \frac{x_{k-1} \cdot \frac{k}{k-1} \cdot ((5-2k)+(k-2)\omega)}{(3  - \omega) \cdot y_{k-1} +  x_{k-1}} = \frac{x_k}{y_k}
    \end{align*}
    as desired.
    
    \paragraph{Bound on $\gamma_k$.} Now, it suffices to show that for $t \geq n^{\gamma_k}$, for our choice of $x$,
    \begin{align}
        n^{k-3+\omega} &\leq \left(\frac{t}{nx}\right)^{3-\omega} n^{\omega(k-2) -2k + 6} = n^{\alpha_k} t^{1 - \frac{\alpha_{k}}{k}}\label{eq:first_bound}\\
        (n\cdot x) \cdot x^{\alpha_{k-1} + \gamma_{k-1}\left(1 - \frac{\alpha_{k-1}}{k-1}\right) - 1} &\leq  x^{\alpha_{k-1} \cdot \frac{k-2}{k-1}} (n\cdot x)^{\frac{\alpha_{k-1}}{k-1}} 
        t^{1- \frac{\alpha_{k-1}}{k-1}}\label{eq:second_bound}
    \end{align}
    The first inequality \eqref{eq:first_bound} is trivially satisfied since we chose $\gamma_k \geq k \left(1 - \frac{3-\omega}{k-\alpha_k}\right).$ 
    
    It suffices to show that the second inequality \eqref{eq:second_bound} is also satisfied. Rearranging, we see that it suffices to show that $x^{\gamma_{k-1}} \leq \frac{t}{n}.$
    Plugging in $x$ from \eqref{eq:x_solution_k1_runtime} and $\gamma_{k-1}$ and rearranging, we obtain that this holds as long as 
    \[ t \geq n^{\frac{(k-1) \cdot ((3 - \omega)^2 + k(\omega - 2)) + \alpha_{k-1} (k(2-\omega) + \omega - 3)}{11-(\omega - 2)\alpha_{k-1} + k(\omega - 2) - 7\omega + \omega^2}}.\]
    To show that all $t \geq n^{\gamma_k}$ satisfies the above inequality, it suffices to check that the exponent above is at most $\gamma_k$, i.e., it suffices to check that 
    \[\frac{(k-1) \cdot ((3 - \omega)^2 + k(\omega - 2)) + \alpha_{k-1} (k(2-\omega) + \omega - 3)}{(k-1-\alpha_{k-1})(\omega - 2) + (3 - \omega)^2}
    \leq \gamma_k = k \left(1 - \frac{3-\omega}{k - \alpha_k}\right).\]
    If $\omega = 2$, we can rewrite \eqref{eq:alpha_recursive} as $\alpha_k = \frac{k\alpha_{k-1}}{(k-1)(1 + \alpha_{k-1})}$ to obtain the following equivalent inequality:
    \begin{align*}
        k - (\alpha_{k-1} + 1) \leq k - \frac{(k-1)}{(k-1) + (k-2) \alpha_{k-1}} \cdot (\alpha_{k-1} + 1),
    \end{align*}
    which clearly holds since $(k-2)\alpha_{k-1} \geq 0.$
    
    When $\omega > 2$, we substitute our recursive formula for $\alpha_k$ from \eqref{eq:alpha_recursive} and rearrange to obtain that the inequality is satisfied for all $k > 0$ as long as 
    \begin{align}\label{eq:bound_alpha_for_gamma}
        (k-1)(\omega - 2) \leq \alpha_{k-1} \leq k-1 + \frac{(3-\omega)^2}{\omega - 2}.
    \end{align}
    
    \begin{claim}\label{claim:alpha_k_lb}
    For $2 \le \omega \leq 3$ and $k \geq 2$, we have $k(\omega - 2) \leq \alpha_k$.
    \end{claim}
    \begin{proof}
         We show this by induction. When $k = 2$, the equation is clearly true because $0 \leq \omega - 2 \leq 1$. 
         Therefore, $2 (\omega - 2) \leq 2,$ and the lower bound clearly holds.
    
    Now suppose $k \ge 3$ and that $\alpha_{k-1} \geq (k-1)(\omega - 2)$. Now, we want $\alpha_k \geq k (\omega - 2)$. Substituting the recursion from \eqref{eq:alpha_recursive}, we have
    \begin{align*}
        \frac{k \alpha_{k-1} ((5-2k) + (k-2)\omega)}{(k-1) \cdot (3 + \alpha_{k-1} - \omega)} \geq k (\omega - 2).
    \end{align*}
    Rearranging the equation using the fact that $3 - \omega + \alpha_{k-1}  > 0$, we have that the above equation holds if and only if
    \begin{align*}
        (3 - \omega) (\alpha_{k-1} - (k-1)(\omega - 2)) \geq 0.
    \end{align*}
    Since $\omega \leq 3$ and $\alpha_{k-1} \geq (k-1)(\omega - 2)$, we have that the equation indeed holds. 
    \end{proof}
   \begin{claim}
        For $2 \leq \omega \leq 3$ and $k \geq 2$, $\alpha_{k} \leq k.$
   \end{claim}
   \begin{proof}
       We proceed by induction. First, we note that $\alpha_2 = 2$ and $\alpha_3 = \frac{3(\omega - 1)}{5 - \omega} \leq 3$ since $\omega \leq 3$. 
        Now, suppose $\alpha_{k-1} \leq k-1$. Then, note that
        \begin{align*}
            \alpha_k &\leq \frac{k \alpha_{k-1} ((5-2k) + (k-2)\omega)}{(k-1) \cdot (3 + \alpha_{k-1} - \omega)} \leq \frac{k ((5-2k) + (k-2)\omega)}{3 + \alpha_{k-1} - \omega}. 
        \end{align*}
        Therefore, $\alpha_k \leq k$ as long as 
        \begin{align*}
            &(5-2k) + (k-2)\omega \leq 3 + \alpha_{k-1} - \omega\\
            \iff & \alpha_{k-1} \geq (k-1)(\omega - 2),
        \end{align*}
        which is indeed true by Claim~\ref{claim:alpha_k_lb}.
   \end{proof}
   Therefore, the bounds in \eqref{eq:bound_alpha_for_gamma} indeed hold, thereby completing the proof.
\end{proof}

Using the bound of Theorem~\ref{thm:k_1_optimal} and note that the runtime of $\cliquelist{k, \ell}$ is monotone with respect to $t$ (Lemma~\ref{lem:list_exponent_monotone}), we immediately get the following corollaries. 
\begin{corollary}[Theorem~\ref{thm:4_1_opt}]\label{cor:4-1-opt}
    Given a graph on $n$ nodes, one can list $t$ 4-cliques in $$\tilde{O}\left(n^{\omega + 1} + n^{\frac{4(\omega - 1)(2\omega - 3)}{\omega^2 - 5 \omega + 12}}t^{1 - \frac{(\omega - 1)(2\omega - 3)}{\omega^2 - 5 \omega + 12}}\right)$$ time. If $\omega = 2$, the runtime is $\tilde{O}(n^3 + n^{2/3}t^{5/6}).$
\end{corollary}

\begin{corollary}[Theorem~\ref{thm:5_1_opt}]\label{cor:5_1_opt}
     Given a graph on $n$ nodes, one can list $t$ 5-cliques in $$\tilde{O}\left(n^{\omega + 2} + n^{\frac{5(\omega - 1)(2\omega - 3)(3\omega - 5)}{48-47\omega + 16\omega^2 - \omega^3}}t^{1 - \frac{(\omega - 1)(2\omega - 3)(3\omega - 5)}{48-47\omega + 16\omega^2 - \omega^3}}\right)$$
    time. If $\omega = 2$, the runtime is $\tilde{O}(n^4 +n^{1/2}t^{9/10}).$
\end{corollary}
    
\subsection{Analysis for \texorpdfstring{$\cliquelist{k, \ell}$}{(k, l)-Clique-Listing} for \texorpdfstring{$\ell \geq 2$}{l >= 2}}\label{sec:k_l_opt}
We have shown an algorithm for $\cliquelist{k, 1}$ that is conditionally optimal when $\Delta_k \geq n^{\gamma_k}$. Now, we use this to show that there exists a $\cliquelist{k, \ell}$ algorithm for all $\ell$ that is conditionally optimal for $\Delta_k \geq n^{\gamma_{k, \ell}}$, for some $0 \leq \gamma_{k, \ell} < \frac{k}{\ell}.$ First, we define the following variable 
\[z_{k,\ell} = x_k \sum_{i=0}^{\ell-1}  \frac{k-\ell}{k-i-1} \cdot  \frac{y_{k-i}}{x_{k-i}}\]

where $x_k$ and $y_k$ are just as defined in \eqref{eq:alpha_num} and \eqref{eq:alpha_denom}. From this definition, the following identity is immediate.
\begin{claim}\label{claim:zkl_recursive}
For $\ell \ge 2$ and $k > \ell$, 
$z_{k,\ell} =\frac{x_k}{x_{k-1}} z_{k-1,\ell-1} +  \frac{k-\ell}{k-1}y_k$. 
\end{claim}

\begin{theorem}
\label{thm:k-l-large-t-listing}
    Fix any constant integers $k-1 \ge \ell \ge 1$.
    Let $\alpha_{k,\ell} = x_k/z_{k,\ell}$. Then, there exists
    some $\gamma_{k, \ell} = (1 - \epsilon_{k,\ell})k/\ell$ for
    $\epsilon_{k,\ell} > 0$ such that for large $t \geq n^{\gamma_k}$ 
    there exists an algorithm that lists all $t$ $k$-cliques given the $\ell$-cliques in time $\tO(\Delta_\ell^{\alpha_{k,\ell}}t^{1-\frac{\ell \alpha_{k,\ell}}{k}})$. 
    
    If $\omega = 2$, we have $x_k = k$ and $z_{k, \ell} = \frac{k\ell(k-\ell)}{2}$, giving a runtime of $\tO(\Delta_\ell^{\frac{2}{\ell(k-\ell)}}t^{1 - \frac{2}{k(k-\ell)}})$ for all $t \geq n^{\gamma_{k, \ell}}$, where
    \[\gamma_{k, \ell} = \frac{k(k^2 - 2k - 1)}{\ell(k^2 - k - \ell - 1)}.
    \]
\end{theorem}

\begin{proof}
    We show this inductively on $\ell$. For $\ell = 1$, we have $z_{k, 1} = y_k$, and this simply reduces to Theorem~\ref{thm:k_1_optimal}.
    
    For some $\ell > 1$, suppose that the theorem statement is true for all $\ell' < \ell$ and $k' > \ell'$. In particular, we assume that $\cliquelist{k',\ell'}$ takes time \[\tO(\Delta_{\ell'}^{\alpha_{k',\ell'}}\Delta_{k'}^{1 - \frac{\ell' \alpha_{k', \ell'}}{k'}})\]
    for $\Delta_{k'} \geq \Delta_{\ell'}^{\gamma_{k', \ell'}},$ and that the runtime is 
    \[\tO\left(\Delta_{\ell'}^{\alpha_{k',\ell'} + \gamma_{k',\ell'} \left(1 - \frac{\ell' \alpha_{k', \ell'}}{k'}\right)}\right)\]
    for $\Delta_{k'} \leq \Delta_{\ell'}^{\gamma_{k', \ell'}}.$
    We may assume this because the runtime is non-decreasing in $\Delta_{k'}$ by Lemma~\ref{lem:list_exponent_monotone}. 

    Recall that, 
    at a high level, Algorithm~\ref{alg:optimal_kl} first lists all $k$-cliques containing nodes that are contained in at most $x$ $\ell$-cliques, and deletes all such nodes. Now, there are at most $k\Delta_\ell/x$ nodes left in the graph, and we call the $\dense$ algorithm from Algorithm~\ref{alg:large_t_sparse_dense}.
    
    \paragraph{Runtime analysis.} Fix any $k > \ell$. In Step~\ref{step:kl_opt_lightnodes} of the algorithm, the runtime is given by (omitting $\tO(1)$ factors):
    \begin{align}
        &\sum_{v: \Delta_\ell(v) \leq x}
        \left(\Delta_\ell(v)^{\alpha_{k-1,\ell-1} + \gamma_{k-1,\ell-1} \left(1 - \frac{(\ell-1)\alpha_{k-1,\ell-1}}{k-1}\right)} + \Delta_\ell(v)^{\alpha_{k-1,\ell-1}} \Delta_k(v)^{1 - \frac{(\ell-1)\alpha_{k-1,\ell-1}}{k-1}}
        \right)\nonumber\\
        &\leq \Delta_\ell x^{\alpha_{k-1,\ell-1} + \gamma_{k-1,\ell-1} \left(1 - \frac{(\ell-1)\alpha_{k-1,\ell-1}}{k-1}\right) - 1} + \Delta_\ell^{\frac{(\ell-1)\alpha_{k-1,\ell-1}}{k-1}} t^{1 - \frac{(\ell-1)\alpha_{k-1,\ell-1}}{k-1}} x^{\frac{(k-\ell)\alpha_{k-1,\ell-1}}{k-1}},\label{eq:k_l_holders}   
    \end{align}
    where we use H\"{o}lder's inequality to bound the second term. Suppose for now that $t$ is large enough so that the second term dominates. 
    
    In Step~\ref{step:kl_densecall}, by Theorem~\ref{thm:k_1_optimal}, the runtime can be bounded by $(k\Delta_\ell/x)^{\alpha_k} t^{1 - \frac{\alpha_k}{k}}$  up to $\tilde{O}(1)$ factors, if we have
    \begin{equation}\label{eq:k_l_secondbound}
        t \geq (\Delta_\ell/x)^{\gamma_k}.
    \end{equation}
    
    Suppose that $t$ is large enough so that this inequality holds. Then, the runtime of the algorithm is
    \begin{align*}
        \Delta_\ell^{\frac{(\ell-1)\alpha_{k-1,\ell-1}}{k-1}} t^{1 - \frac{(\ell-1)\alpha_{k-1,\ell-1}}{k-1}} x^{\frac{(k-\ell)\alpha_{k-1,\ell-1}}{k-1}} +  \left(\frac{\Delta_\ell}{x}\right)^{\alpha_k} t^{1 - \frac{\alpha_k}{k}}.
    \end{align*}
    Choosing \begin{equation}\label{eqn:x_k_l_listing}
    x = \Delta_\ell^{\frac{(k-1) \alpha_k - (\ell-1)  \alpha_{k-1,\ell-1} }{(k-1)\alpha_k + (k-\ell)\alpha_{k-1,\ell-1}}} 
    t^{\frac{k(\ell-1) \alpha_{k-1,\ell-1} - (k-1) \alpha_k}{k((k-1)\alpha_k + (k-\ell)\alpha_{k-1,\ell-1})}},\end{equation}
    we get a runtime of (omitting $\tO(1)$ factors)
    \[
        \Delta_\ell^{\frac{\alpha_k \alpha_{k-1,\ell-1} (k-1)}{\alpha_k(k-1) + \alpha_{k-1,\ell-1} (k-\ell))}} t^{1 - \frac{\ell}{k} \cdot \frac{\alpha_k \alpha_{k-1,\ell-1} (k-1)}{\alpha_k(k-1) + \alpha_{k-1,\ell-1} (k-\ell))}},
    \]
    therefore giving 
    \begin{align*}
        \alpha_{k,\ell} &= \frac{\alpha_k \alpha_{k-1,\ell-1} (k-1)}{\alpha_k(k-1) + \alpha_{k-1,\ell-1} (k-\ell)}\\
        &= \frac{\frac{x_k}{y_k}\cdot \frac{x_{k-1}}{z_{k-1,\ell-1}}}{\frac{x_k}{y_k} + \frac{k-\ell}{k-1} \cdot \frac{x_{k-1}}{z_{k-1,\ell-1}}}\\
        &= \frac{x_k}{\frac{x_k}{x_{k-1}} z_{k-1,\ell-1} + \frac{k-\ell}{k-1}\cdot y_k} = \frac{x_k}{z_{k,\ell}}
    \end{align*}
    by Claim~\ref{claim:zkl_recursive}.
    
    \paragraph{Bound on $\gamma_{k, \ell}$ if $\omega = 2$.} 
    If $\ell = 1$, then $\gamma_{k, 1}$ matches the value of $\gamma_k$ we obtained from Theorem~\ref{thm:k_1_optimal}. Thus, we assume $\ell > 1$ and the bound for all $\ell' < \ell$ holds. 
    Recall that when $\omega = 2$, $\alpha_k = \frac{2}{k-1}$ and $\alpha_{k-1, \ell-1} = \frac{2}{(\ell-1)(k-\ell)}$. Therefore, substituting this into \eqref{eqn:x_k_l_listing}, we obtain
    \[
        x = \Delta_\ell^{\frac{(k-\ell-1)(\ell - 1)}{\ell(k-\ell)}} t^{\frac{\ell - 1}{k(k-\ell)}}.
    \]
    First, we check that \eqref{eq:k_l_secondbound} holds. In fact,
    \begin{align*}
        &t \geq \left(\frac{\Delta_\ell}{x}\right)^{\gamma_k} = \left(\Delta_\ell \cdot \Delta_\ell^{-\frac{(k-\ell-1)(\ell - 1)}{\ell(k-\ell)}} t^{-\frac{\ell - 1}{k(k-\ell)}}\right)^
        {\gamma_k}\\
        \iff & t^{1 + \frac{\gamma_k(\ell - 1)}{k(k-\ell)}} \geq \Delta_\ell^{\frac{(k-1)\gamma_k}{\ell(k-\ell)}}\\
        \iff & t \geq \Delta_\ell^{\frac{k(k-1)\gamma_k}{k\ell(k- \ell) + \ell (\ell - 1)\gamma_k}}.
    \end{align*}
    Substituting $\gamma_k = k - 1 - \frac{2}{k^2 - k - 2}$ from Theorem~\ref{thm:k_1_optimal}, we get the inequality
    $t \geq \Delta_\ell^{ \frac{k(k^2 - 2k - 1)}{\ell(k^2 - k - \ell - 1)}} = \Delta_\ell^{\gamma_{k, \ell}}$, which is indeed true by our choice of $\gamma_{k, \ell}$. 
    
    Now, it suffices to show that if $t \geq \Delta_\ell^{\gamma_{k ,\ell}}$, then the second term dominates in \eqref{eq:k_l_holders}. In fact, the second term dominates as long as 
    \begin{align*}
        &t \geq \Delta_\ell x^{\gamma_{k-1,\ell-1} - 1}\\
        \iff &t \geq \Delta_\ell^\frac{1+\frac{(k-\ell-1)(\ell - 1)}{\ell(k-\ell)}(\gamma_{k-1,\ell-1}-1)}{1 - \frac{\ell - 1}{k(k-\ell)}(\gamma_{k-1,\ell-1}-1)} = \Delta_\ell^{\frac{k(k-1)(k-3)}{\ell(k^2 - 3k + 3 - \ell)}},
    \end{align*}
    where the last equality holds because $\gamma_{k-1,\ell-1} = \frac{(k-1)((k-1)^2 - 2(k-1)-1)}{(\ell-1) ((k-1)^2 - k - \ell + 1)}$
    by induction. 
    Hence, it suffices to show that $\gamma_{k,\ell}$ is at least the exponent on the right-hand side.
    \begin{align*}
        &\gamma_{k, \ell} = \frac{k(k^2 - 2k - 1)}{\ell(k^2 - k - \ell - 1)} \geq \frac{k(k-1)(k-3)}{\ell(k^2 - 3k + 3 - \ell)}\\
        \iff & 1 - \frac{k-\ell}{k^2-k-1-\ell} \geq 1 - \frac{k-\ell}{k^2 - 3k - \ell + 3}\\
        \iff & k^2 - k - 1 - \ell \geq k^2 - 3k - \ell + 3\\
        \iff & k \geq 2,
    \end{align*}
    which is true since $k \geq 3.$
    
    \paragraph{Bound on $\gamma_{k, \ell}$ if $\omega > 2$.} In this case, we show that there exists some $\epsilon_{k,\ell} > 0$ such that $\gamma_{k,\ell} \leq \frac{k}{\ell}(1 - \epsilon_{k, \ell})$.

    First, consider \eqref{eq:k_l_secondbound}. Note that one can rewrite 
    \begin{align*}
       x = \Delta_\ell^{1 - \frac{\alpha_{k,\ell}}{\alpha_k}} t^{\frac{\ell \alpha_{k,\ell}}{k\alpha_k} - \frac{1}{k}}.
    \end{align*}
    Rewriting $q = \frac{\alpha_{k, \ell}}{\alpha_k}$, and substituting this into \eqref{eq:k_l_secondbound}, we obtain
    \begin{align*}
        t \geq \left(\frac{\Delta_\ell}{x}\right)^{\gamma_k} = \left(\frac{\Delta_\ell^q}{t^{\frac{\ell q}{k} - \frac{1}{k}}}\right)^{\gamma_k} \iff t \geq \Delta_\ell^{\frac{kq \gamma_k}{\ell q \gamma_k + k - \gamma_k}}.
    \end{align*}
By choosing $\epsilon_1 = \frac{k - \gamma_k}{\ell q \gamma_k + k - \gamma_k}$ (which is positive since $k > \gamma_k$ by Theorem~\ref{thm:k-l-large-t-listing}, it is easy to check that the right-hand side is equal to $\Delta_\ell^{\frac{k}{\ell}(1 - \epsilon_1)}$.

Now, consider \eqref{eq:k_l_holders}. For the second term to dominate, we can rewrite the inequality as 
\begin{align*}
    t \geq \Delta_\ell \cdot x^{\gamma_{k-1,\ell-1} - 1} = \Delta_\ell^{1 + (\gamma_{k-1,\ell-1} - 1)(1-q)} t^{\frac{1}{k} \cdot (\gamma_{k-1,\ell-1} - 1)(\ell q - 1)}.
\end{align*}
Rearranging this, we see that we require
\begin{align*}
    t \geq \Delta_\ell^{\frac{1 + (\gamma_{k-1,\ell-1} - 1) (1 - q)}{1 - \frac{1}{k}(\gamma_{k-1,\ell-1} - 1)(\ell q - 1)}}.
\end{align*}
By the induction hypothesis, we know there exists some $0 < \epsilon' < 1$ such that $\gamma_{k-1,\ell-1} = \frac{k-1}{\ell - 1}(1 - \epsilon').$ Therefore, substituting this into the above equation and rearranging, we require
\begin{align}\label{eq:gamma_kl_holderbound}
    t \geq \Delta_\ell^{\frac{k}{\ell}\left(1 - \frac{(k-1)(\ell - 1)\epsilon'}{(q\ell - 1)(k-1) \epsilon' + \ell((k-1) - q(k-\ell))}\right)}.
\end{align}
Let $\epsilon_{num} = (k-1)(\ell-1)\epsilon'$ and and $\epsilon_{den} = (q\ell - 1)(k-1) \epsilon' + \ell((k-1) - q(k-\ell))$. Clearly, since $k \geq 3$ and $\ell \geq 2$, $\epsilon_{num} > 0$. Now, consider two cases.
\begin{itemize}
    \item $q\ell - 1 \geq 0$. Then, since $\epsilon' > 0$, we have $\epsilon_{den} \geq \ell((k-1) - q(k-\ell)) > 0$ since $q = \frac{\alpha_{k,\ell}}{\alpha_k} = \frac{y_k}{z_{k,\ell}} < \frac{k-1}{k-\ell}$ by Claim~\ref{claim:zkl_recursive}.
    \item $q\ell - 1 < 0$. Then, since $\epsilon' < 1,$ $\ell \geq 2,$ $q > 0$ and $k \geq 3$
    \begin{align*}
        \epsilon_{den} &> (q\ell - 1)(k-1)  + \ell((k-1) - q(k-\ell))\\
        & = q\ell (\ell - 1) + (\ell-1)(k-1) > 0.
    \end{align*}
\end{itemize}

    Therefore, let $\epsilon_2 = \frac{\epsilon_{num}}{\epsilon_{den}}$. Clearly, $\epsilon_2 > 0$. Then, if $t \geq \Delta_\ell^{\frac{k}{\ell}(1 - \epsilon_2)}$, then \eqref{eq:gamma_kl_holderbound} holds. Hence, we can pick $\epsilon_{k, \ell} = \min\{\epsilon_1, \epsilon_2\} > 0$ to ensure both conditions \eqref{eq:k_l_holders} and \eqref{eq:k_l_secondbound} hold.
\end{proof}

%% file: 6-general-list.tex
In this section, we show how to apply our algorithm in Section~\ref{sec:upper-bound} which only works for very large $t$ (or rather, does not have improved runtime for smaller $t$) to other ranges of $t$ as well, via black-box reductions.

\begin{theorem}
\label{thm:reduction_from_small_t}
Suppose for every $1 \le \ell < k$, $\cliquelist{k, \ell}$ can be solved in $\tO(\Delta_\ell^{\alpha_{k, \ell}} t^{1-\frac{\ell \alpha_{k, \ell}}{k}})$ time when $t \ge \Delta_\ell^{\gamma_{k, \ell}}$. Then for every $1 \le \ell \le k$ and $1 \le s < k$ where $\lceil \frac{k}{s} \rceil \ne \lceil \frac{\ell}{s} \rceil$, $\cliquelist{k, \ell}$ can be solved in 
$$\tO\left( \left(\Delta_{\ell}^{\frac{s}{\ell} \lceil \frac{\ell}{s}\rceil} \right)^{\alpha_{k', \ell'}} t^{1-\frac{\ell' \alpha_{k', \ell'}}{k'}} \right)$$
time for $t \ge \left(\Delta_{\ell}^{\frac{s}{\ell} \lceil \frac{\ell}{s}\rceil} \right)^{\gamma_{k', \ell'}}$, where $k' = \lceil \frac{k}{s}\rceil$ and $\ell' = \lceil \frac{\ell}{s}\rceil$. 
\end{theorem}
\begin{proof}
Let $G$ be the input of a $\cliquelist{k, \ell}$ instance. Without loss of generality, assume $G$ is $k$-partite with parts $V_1, \ldots, V_k$. Create a new $k'$-partite graph $G'$ on node parts $U_1 = V_1 \times \cdots \times V_s,  U_2 = V_{s+1} \times \cdots \times V_{2s}, \ldots, U_{k'} = V_{s(k' - 1) + 1} \times \cdots \times V_k$ (each node corresponds to a set of at most $s$ nodes). Keep a node $(v_1, v_2, \ldots, v_i)$ if and only if $(v_1, v_2, \ldots, v_i)$ forms a clique in $G$.  
Add an edge between two nodes $(v_1, v_2, \ldots, v_i)$ and $(v_1', v_2', \ldots, v_{i'}')$ belonging to two different parts if and only if the nodes $(v_1, v_2, \ldots, v_i, v_1', v_2', \ldots, v_{i'}')$ form a clique in $G$. 

Clearly, $k'$-cliques in $G'$ have one-to-one correspondence with $k$-cliques in $G$, so it suffices to list $t$ $k'$-cliques in $G'$ in order to list $t$ $k$-cliques in $G$. Furthermore, distinct $\ell'$-clique in $G'$ corresponds to distinct clique in $G$. Depending on whether an $\ell'$-clique uses a node in $U_{k'}$, it corresponds to either an $(s\ell')$-clique in $G$ or an $(s\ell'+k-sk')$-clique in $G$. Either way, it is a clique of size at most $s\ell'$. Thus, by Lemma~\ref{lem:simple_list_ub}, there are $\tO(\Delta_\ell^{\frac{s \ell'}{\ell}}) = \tO(\Delta_\ell^{\frac{s}{\ell} \lceil \frac{\ell}{s}\rceil})$ such cliques in $G$ and we can list them in $\tO(\Delta_\ell^{\frac{s}{\ell} \lceil \frac{\ell}{s}\rceil})$ time as well. 

Thus, to solve $\cliquelist{k, \ell}$ on $G$ with $t$ $k$-cliques, it suffices to solve $\cliquelist{k', \ell'}$ on $G'$ with $t$ $k'$-cliques. The theorem thus easily follows. 
\end{proof}

Let us give some examples to show how to use Theorem~\ref{thm:reduction_from_small_t}. 

First, for any $1 \le \ell \le k$, let us take the extreme example $s = 1$. In this case, the running time is exactly the running time given in Theorem~\ref{thm:k-l-large-t-listing}. On the other extreme end, $s = k - 1$. Then $k' = 2$ and $\ell' = 1$. In this extreme case, we have $\alpha_{k', \ell'} = 2$ and $\gamma_{k', \ell'} = 0$. Thus, we get an algorithm that works for any $t \ge 1$, although its running time $\tO\left(\Delta_\ell^{\frac{2(k-1)}{\ell}}\right)$ is not great. One can imagine when increasing $s$ from $1$ to $k-1$, we achieve a trade-off between the bound for $t$ and running time of the algorithm. 

Let us give the following more concrete examples. For simplicity, we assume $\omega = 2$. 

\begin{corollary}
\label{cor:k-1-listing-general}
Assume $\omega = 2$. 
Fix any integer  $k \ge 2$, and any integer $1 \le s < \frac{k}{2}$. $\cliquelist{k, 1}$ can be solved in $\tO\left(n^{\frac{2s}{k'-1}} t^{1-\frac{2}{k'(k'-1)}} \right)$ time when $t \ge n^{s(k'-1-\frac{2}{k'^2-k'-2})}$, where $k' = \lceil \frac{k}{s}\rceil$.
\end{corollary}
\begin{proof}
Apply Theorem~\ref{thm:reduction_from_small_t}. We get $k' = \lceil \frac{k}{s}\rceil\ge 3$, $\ell' = 1$, and an algorithm for $\cliquelist{k, 1}$ that runs in $\tO(\left(n^{s}\right)^{\alpha_{k', 1}} t^{1-\frac{\alpha_{k',1}}{k'}})$ time when $t \ge \left(n^{s}\right)^{\gamma_{k', 1}}$. By Theorem~\ref{thm:k_1_optimal}, when $\omega = 2$, $\alpha_{k', 1} = \frac{2}{k'-1}$ and $\gamma_{k', 1} = k'-1-\frac{2}{k'^2 - k' - 2}$. Thus, we get an 
$\tO\left(n^{\frac{2s}{k'-1}} t^{1-\frac{2}{k'(k'-1)}} \right)$ time algorithm for $t \ge n^{s(k'-1-\frac{2}{k'^2-k'-2})}$. 
\end{proof}

\begin{example}[$\cliquelist{12, 1}$]\label{eg:12_1_list}
$\cliquelist{12, 1}$ has the following running times (by setting $s = 1, 2, 3, 4$ in Corollary~\ref{cor:k-1-listing-general}):
\begin{itemize}
    \item $\tO\left(n^{\frac{2}{11}} t^{\frac{65}{66}}\right)$ when $t \ge n^{11-\frac{1}{65}}$;
    \item $\tO\left(n^{\frac{4}{5}} t^{\frac{14}{15}}\right)$ when $t \ge n^{10-\frac{1}{7}}$;
    \item $\tO\left(n^{2} t^{\frac{5}{6}}\right)$ when $t \ge n^{9-\frac{3}{5}}$;
    \item $\tO\left(n^{4} t^{\frac{2}{3}}\right)$ when $t \ge n^{6}$;
    \item $\tO\left(n^8\right)$ when $t<n^6$.
\end{itemize}
\end{example}
Figure~\ref{fig:12_1} shows a pictorial representation of the $\cliquelist{12,1}$ runtime.

\begin{figure}[ht]
    \centering
    \includegraphics[width=0.5\textwidth]{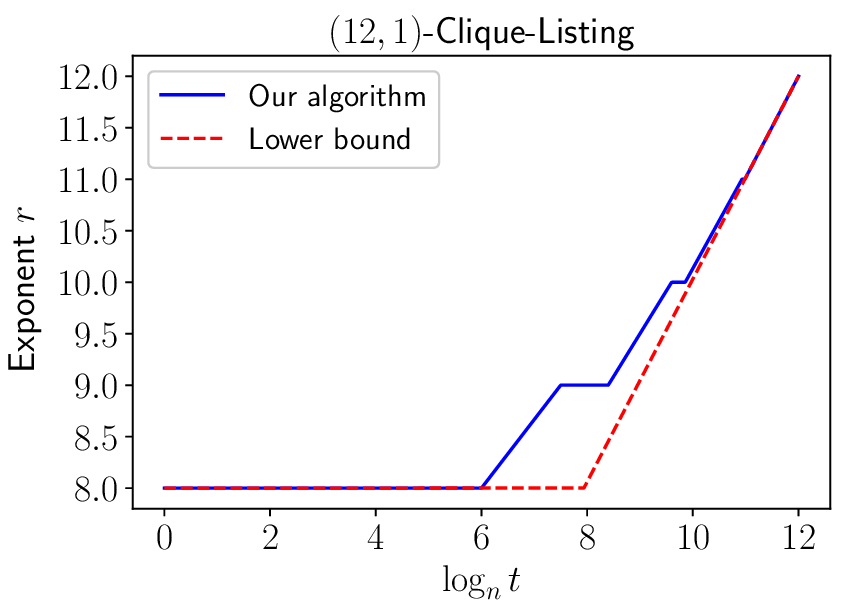}
    \caption{Upper and lower bounds for the runtime $r$ for $\cliquelist{12, 1}$ for a graph with $t$ $k$-cliques, assuming $\omega = 2$. The upper bound is from Example~\ref{eg:12_1_list} and the lower bounds are from Proposition~\ref{prop:eh_upper_bound} and Theorem~\ref{thm:lower_bound}.}\label{fig:12_1}
\end{figure}

We can similarly obtain the following corollary for $\cliquelist{k, 2}$. 

\begin{corollary}
\label{cor:k-2-listing-general}
Assume $\omega = 2$. 
Fix any integer  $k \ge 2$, and any integer $2 \le s < \frac{k}{2}$. $\cliquelist{k, 2}$ can be solved in $\tO\left(m^{\frac{s}{k'-1}} t^{1-\frac{2}{k'(k'-1)}} \right)$ time when $t \ge m^{\frac{s}{2} (k'-1-\frac{2}{k'^2-k'-2})}$, where $k' = \lceil \frac{k}{s}\rceil$.
\end{corollary}
\begin{proof}
Apply Theorem~\ref{thm:reduction_from_small_t}. We get $k' = \lceil \frac{k}{s}\rceil\ge 3$, $\ell' = 1$, and an algorithm for $\cliquelist{k, 2}$ that runs in $\tO\left(\left(m^{\frac{s}{2}}\right)^{\alpha_{k', 1}} t^{1-\frac{\alpha_{k',1}}{k'}}\right)$ time when $t \ge \left(m^{\frac{s}{2}}\right)^{\gamma_{k', 1}}$. By Theorem~\ref{thm:k_1_optimal}, when $\omega = 2$, $\alpha_{k', 1} = \frac{2}{k'-1}$ and $\gamma_{k', 1} = k'-1-\frac{2}{k'^2 - k' - 2}$. The corollary then follows. 
\end{proof}

%% file: 7-6-clique.tex
In this section, we show that our algorithm in Section~\ref{sec:upper-bound} is improvable by showing a faster algorithm for $\cliquelist{6, 1}$. See Figure~\ref{fig:6_1} for a comparison of the bounds achieved by the algorithm in Section~\ref{sec:general-list} and this section. 

The new algorithm for $\cliquelist{6, 1}$ comprises two parts:  Algorithms~\ref{alg:6_clique_Algo_1} and \ref{alg:6_clique_Algo_2}. 

\begin{algorithm}
    \caption{$\cliquelist{6, 1}$ Algorithm I.}\label{alg:6_clique_Algo_1}
    \begin{algorithmic}
    \item \textbf{Input:} $(G := (V, E), n, t)$
    \item \textbf{Output:} The list of all $\le t$ $6$-cliques.
    \begin{enumerate}
    \item If $n \le 5$, list nothing and return. 
    \item Call a $K_4$ light if it is contained in at most $\rho$ $K_6$ for some $\rho \ge 1$. Clearly, there are at most $15t/\rho$ dense $K_4$. 
    \item Just as before, we can list all $K_6$ containing light $K_4$ in $\tO(\rho \MM(n^2, \frac{n^2}{\rho}, n^2))$ time. 
    \label{line:6-1-list-denseK4}
    \item Call an edge light if it is contained in at most $\lambda$ $K_6$ for some $\lambda \ge 1$. All other edges are dense. There are at most $15 t/\lambda$ dense edges. 
    \item Just as before, we can list all $K_6$ containing one light edge and one dense $K_4$ that is disjoint with the light edge in $\tO(\lambda \MM(n, \frac{t/\rho}{\lambda}, n))$ time. 
    \label{line:6-1-list-denseedge}
    \item For each node $v$ connected to  $d_v \le x$ dense edges for some $x \ge 1$, run the $\cliquelist{5, 1}$ algorithm from Corollary~\ref{cor:5_1_opt} in its neighbors connected to it by dense edges. Delete this node afterwards. 
    \item The number of remaining nodes is at most $30t/x\lambda$; recurse. 
    \end{enumerate}
    \end{algorithmic}
\end{algorithm}

\begin{remark}
    Intuitively, Algorithm~\ref{alg:6_clique_Algo_1} is similar to Algorithm~\ref{alg:optimal_kl} with one main difference: we first bound the number of 4-cliques in the graph by $O(t/\rho)$ by getting rid of light 4-cliques rather than simply bounding the number of 4-cliques by $n^4$. This idea allows us to get a better bound on $\gamma_6$ than in Theorem~\ref{thm:k_1_optimal}. This idea can also be extended to all $k \geq 6$.
\end{remark}

\begin{lemma}
\label{lem:6_clique_Algo_1}
Algorithm~\ref{alg:6_clique_Algo_1} is correct and runs in $\tO(n^4 + n^{5/2}t^{1/2}+n^{2/5}t^{14/15})$ time if $\omega = 2$. 
\end{lemma}
\begin{proof}
After Line~\ref{line:6-1-list-denseK4}, the algorithm has listed all $K_6$ containing at least one light $K_4$. After Line~\ref{line:6-1-list-denseedge}, the algorithm  has also listed all $K_6$ containing a dense $K_4$ and a disjoint light edge. Thus, after this point, only $K_6$ containing no light edges are not listed. Then clearly, the next two steps list all such $K_6$. 

The running time, excluding the recursion, is (if $\omega = 2$)
\begin{align*}
    &\tO\left(\rho \MM\left(n^2, \frac{n^2}{\rho}, n^2\right) + \lambda \MM\left(n, \frac{t/\rho}{\lambda}, n\right) + \sum_{v: d_v \le x} \left( d_v^4 + d_v^{1/2}\Delta_6(v)^{9/10} \right) \right)\\
    \le & \tO\left(\rho n^4 + \lambda n^2 + \frac{nt}{\rho}  +  (t/\lambda) x^3 + (t/\lambda)^{1/10} t^{9/10} x^{2/5} \right).
\end{align*}
The inequality is due to $\sum_{v}d_v \le O(t/\lambda)$,  $\sum_v \Delta_6(v) \le O(t)$ and H\"{o}lder's inequality. 

We also set $\lambda = \max\{1, \frac{15t}{xn}\}$, so that each recursion level decreases $n$ by a factor of at least $2$. The overall time complexity is thus within $\tO(1)$ of the time complexity of the first recursion level. The running time then becomes (assuming $\frac{15t}{xn} \ge 1$)
\begin{align*}
    \tO\left(\rho n^4 + \frac{nt}{x} + \frac{nt}{\rho}  +  x^4 n + t^{9/10}x^{1/2}n^{1/10} \right).
\end{align*}
The running time of the algorithm is thus
\begin{itemize}
    \item $\tO(n^4)$ when $t \le n^3$ by setting $\rho = 1$ and $x=1$ (even though in this setting, $\frac{15t}{xn}$ will be less than $1$ if $t < n / 15$, the running time still holds by setting $\lambda = 1$);
    \item $\tO(n^{5/2} t^{1/2})$ when $n^3 < t \le n^{63/13}$ by setting $x = \rho = n^{-3/2} t^{1/2}$;
    \item and $\tO(n^{2/5} t^{14/15})$ when $t > n^{63/13}$ by setting $x = n^{3/5} t^{1/15}$ and $\rho = n^{-18/5}t^{14/15}$. \qedhere
\end{itemize}
\end{proof}

In Algorithm~\ref{alg:6_clique_Algo_2}, we show another alternative algorithm for $\cliquelist{6, 1}$ that performs better for different ranges of $t$.  

\begin{algorithm}
    \caption{$\cliquelist{6, 1}$ Algorithm II.}\label{alg:6_clique_Algo_2}
    \begin{algorithmic}
    \item \textbf{Input:} $(G := (V, E), n, t)$
    \item \textbf{Output:} The list of all $\le t$ $6$-cliques
    \item \textbf{The Algorithm:}
    \begin{enumerate}
    \item If $n \le 5$, list nothing and return. 
    \item Call a $K_4$ light if it is contained in at most $\rho$ $K_6$ for some $\rho \ge 1$. Clearly, there are at most $15 t/\rho$ dense $K_4$. 
    \item Just as before, we can list all $K_6$ containing light $K_4$ in $\tO(\rho \MM(n^2, \frac{n^2}{\rho}, n^2))$ time. 
    \item Call an edge $e$ light if it is contained in $q_e \le x$ dense $K_4$ for some $x \ge 1$. All other edges are dense. There are at most $90 t/\rho x$ dense edges. 
    \item For every light edge $e$, we call the $\cliquelist{4, 2}$ algorithm in Section~\ref{sec:4_5_l_listing} using all edges that are disjoint with $e$ and form a dense $K_4$ with $e$. We remove edge $e$ afterwards. 
    \item For each node $v$ connected to  $d_v \le y$ dense edges for some $y \ge 1$, run the $\cliquelist{5, 1}$ algorithm from Corollary~\ref{cor:5_1_opt} in its neighbors connected to it by dense edges. Delete this node afterwards. 
    \item The number of remaining nodes is at most $180 t/\rho x y $; recurse. 
    \end{enumerate}
    \end{algorithmic}
\end{algorithm}

\begin{remark}
    Intuitively, Algorithm~\ref{alg:6_clique_Algo_2} is similar to the algorithm obtained from Theorem~\ref{thm:reduction_from_small_t} by setting $k = 6$ and $s = 2$ and reducing the problem to $\cliquelist{3, 1}$. However, instead of calling $\cliquelist{2, 1}$ (as one would in the usual $\cliquelist{3, 1}$ algorithm), we   call $\cliquelist{4, 2}$ instead. This is better because $\cliquelist{4, 2}$ takes advantage of matrix multiplication whereas $\cliquelist{2, 1}$ simply uses brute-force.
\end{remark}

\begin{lemma}
\label{lem:6_clique_Algo_2}
Algorithm~\ref{alg:6_clique_Algo_2} is correct and runs in $\tO(n^4 + n^{15/7}t^{4/7}+n^{37/21}t^{2/3}+n^{29/25}t^{4/5}+n^{9/10}t^{17/20})$ time. 
\end{lemma}
\begin{proof}
The correctness of the algorithm is almost immediate. The running time of the algorithm, excluding the recursion, is (if $\omega = 2$)
\begin{align*}
    &\tO\left(\rho \MM\left(n^2, \frac{n^2}{\rho}, n^2\right) + \sum_{e: q_e \le x} \left(q_e^{3/2} + q_e \Delta_6(e)^{2/5} + q_e^{1/2}\Delta_6(e)^{3/4} \right)+ \sum_{v: d_v \le y} \left( d_v^4 + d_v^{1/2}\Delta_6(v)^{9/10} \right) \right)\\
    \le & \tO\left(\rho n^4 + (t/\rho)x^{1/2} + (t/\rho)^{3/5}t^{2/5}x^{2/5}+(t/\rho)^{1/4}t^{3/4}x^{1/4}+ (t/\rho x) y^3 + (t/\rho x)^{1/10} t^{9/10}y^{2/5}\right).
\end{align*}
The inequality is due to 
$\sum_e q_e = O(t/\rho), \sum_v d_v = O(t/\rho x)$ and H\"{o}lder's inequality. 
We also set $\rho = \max\{1, \frac{90  t}{xyn}\}$ so that each recursion level decreases $n$ by a factor of at least $2$. The overall time complexity is thus within $\tO(1)$ of the time complexity of the first recursion level. The running time then becomes (assuming $\rho =\frac{90  t}{xyn}$)
$$\tO\left(\frac{n^3 t}{xy} + x^{3/2}yn +xy^{3/5}n^{3/5}t^{2/5}+x^{1/2}y^{1/4}n^{1/4}t^{3/4}+ y^4 n + y^{1/2}n^{1/10}t^{9/10}\right).$$
The running time of the algorithm is thus
\begin{itemize}
    \item $\tO(n^4)$ when $t \le n^{13/4}$ by setting $x = n^{3/2}, y = n^{3/4}$ (even though in this setting, $\frac{90t}{xyn}$ will be less than $1$ if $t < n^{7/4} / 90$,  the running time still holds by setting $\rho = 1$);
    \item $\tO(n^{15/7}t^{4/7})$ when $n^{13/4} < t \le n^4$ by setting $x=n^{4/7}t^{2/7}$ and $y = n^{2/7} t^{1/7}$;
    \item $\tO(n^{37/21}t^{2/3})$ when $n^4 < t \le n^{158/35}$ by setting $x=n^{22/21}t^{1/6}$ and $y = n^{4/21}t^{1/6}$;
    \item $\tO(n^{29/25}t^{4/5})$ when $n^{158/35} < t \le n^{26/5}$ by setting $x=n^{9/5}$ and $y = n^{1/25}t^{1/5}$;
    \item $\tO(n^{9/10}t^{17/20})$ when $n^{26/5} < t \le n^{6}$ by setting $x=n^{1/2}t^{1/4}$ and $y = n^{8/5}t^{-1/10}$ (note that $y$ can be $> n$ sometimes. This would mean all nodes are ``dense nodes'', and we could improve the running time by decreasing $y$ to $n$. Nevertheless, Algorithm I performs better in this regime. )
\end{itemize}
\end{proof}

Combining Lemma~\ref{lem:6_clique_Algo_1} and Lemma~\ref{lem:6_clique_Algo_2}, we obtain the following upper bound for $\cliquelist{6, 1}$. 
\begin{proposition}
If $\omega = 2$, $\cliquelist{6, 1}$ can be solved in time $$\tO\left(\min\left\{n^4 + n^{5/2}t^{1/2}+n^{2/5}t^{14/15}, n^4 + n^{15/7}t^{4/7}+n^{37/21}t^{2/3}+n^{29/25}t^{4/5}+n^{9/10}t^{17/20} \right\}\right).$$
\end{proposition}